\documentclass[a4paper, 11pt]{article}
\usepackage{amsmath, amsthm, bbm, booktabs, multirow, threeparttable, soul}
\usepackage[margin=1in]{geometry}

\usepackage{float, subcaption, standalone}
\usepackage{tikz, pgfplots}
\usetikzlibrary{arrows.meta, shapes, patterns, calc, positioning, plotmarks, pgfplots.groupplots}
\pgfplotsset{compat=1.18}

\usepackage{algorithm, algpseudocode}
\algnewcommand\algorithmicinput{\textbf{Input:}}
\algnewcommand\algorithmicoutput{\textbf{Output:}}
\algnewcommand\Input{\item[\algorithmicinput]}
\algnewcommand\Output{\item[\algorithmicoutput]}

\newtheorem{theorem}{Theorem}

\newtheorem{lemma}{Lemma}

\newtheorem{remark}{Remark}
\newtheorem{corollary}{Corollary}

\usepackage{apacite}
\usepackage{natbib}
\usepackage{hyperref}
\usepackage{cleveref}
\usepackage[labelfont=bf]{caption}
\usepackage{authblk}

\usepackage{tabularx, environ}
\makeatletter
\newcommand{\problemtitle}[1]{\gdef\@problemtitle{#1}}
\newcommand{\probleminstance}[1]{\gdef\@probleminstance{#1}}
\newcommand{\problemquestion}[1]{\gdef\@problemquestion{#1}}
\NewEnviron{dproblem}{
    \problemtitle{}\probleminstance{}\problemquestion{}
    \BODY
    \par\addvspace{.5\baselineskip}
    \noindent
    \begin{tabularx}{\textwidth}{@{\hspace{\parindent}} l X c}
        \multicolumn{2}{@{\hspace{\parindent}}l}{\textbf{\@problemtitle}} \\
        \text{INSTANCE:} & \@probleminstance \\
        \text{QUESTION:} & \@problemquestion
    \end{tabularx}
    \par\addvspace{.5\baselineskip}
}
\makeatother

\title{Designing Efficient and Reachable Routes:\\ The $k$-Step-Central Shortest Path Problem}
\author[1,2]{Johnson Phosavanh}
\author[2]{Dmytro Matsypura\thanks{Corresponding author: dmytro.matsypura@sydney.edu.au}}
\affil[1]{The Hong Kong Polytechnic University}
\affil[2]{The University of Sydney}

\begin{document}

\maketitle

\begin{abstract}
Designing rapid transportation routes requires balancing efficiency and reachability. Shortest-path models ensure direct, cost-efficient routes but ignore coverage, while centrality-based approaches maximize accessibility but do not enforce operational constraints. We study the problem of selecting a shortest path that maximizes reachability, measured as the number of nodes within a fixed distance of the path. To do this, we introduce the \emph{$k$-step-central shortest path} problem and analyse its structural properties. We show that optimal solutions on unweighted graphs can be found in polynomial time and propose an algorithm with a novel pruning rule. We also prove that the problem becomes NP-hard when edge weights are introduced. Additionally, we show that our algorithm can be used to solve the NP-hard problem of finding the closeness-central shortest path in a graph. We demonstrate the efficiency and scalability of our algorithm on synthetic and real-world networks with up to 2,000 nodes. Our results show that improving reachability can substitute for route expansion: increasing the reach of transit lines drastically increases their coverage with shorter routes. This suggests that investments in active transport infrastructure that improve reachability can be more effective than extending primary routes, providing a data-driven basis for allocating resources in network design.
\end{abstract}

\section{Introduction}

\subsection{Motivation}

When planning mass transit routes in dense urban areas, planners must maximise ridership to maximise the return on investment. Directness, convenience, proximity, and destination choice are a natural set of spatial concepts that encapsulate basic travel needs and should be considered simultaneously when planning transport options \citep{TransportChoice}. Broadly speaking, these concepts can be grouped into two goals: direct routing and accessibility/connectivity. Indirect routes increase travel times, prompting commuters to utilise alternative forms of transport. Similarly, the location of stations and stops affects public transport uptake, with dense, high-population catchment areas near stations leading to higher potential ridership. However, these goals may not be simultaneously satisfiable, and ``there may need to be trade-offs between, for example, directness and linking destinations'' \citep[p.~5]{TransportChoice}. Direct routes tend to bypass peripheral areas, whereas routes designed to maximise coverage often become longer and slower. We model this problem as a most-central shortest path problem. The goal is to find a shortest path in a network that is most central with respect to a measure of centrality; in our case, the $k$-step reach centrality, which counts the number of nodes at most $k$ hops from the selected path. We use the path to represent the planned route and the $k$-step reach neighbourhood to model its catchment area. We restrict the feasible set to the shortest paths between all pairs of nodes in the graph, ensuring that our proposed routes are direct and cost-efficient to construct. 

We focus on the $k$-step reach centrality measure because it offers flexibility. It can serve as both a local and a global measure of centrality, depending on the choice of $k$. A small $k$ emphasises the immediate neighbourhood of a path, while a sufficiently large $k$ can cover the entire network. In the transport setting, $k$ can be interpreted as a proxy for the connectivity of an area. Varying $k$ allows urban planners to model different commuter behaviours: a higher $k$ corresponds to cities with better accessibility standards and larger catchment areas. On a grid network, $k$ corresponds directly to the number of blocks that commuters are willing to walk to reach their nearest station. In practice, the value of $k$ can be calibrated for a specific context. For example, urban planners have proposed guidelines calling for railway stations and bus stops to be within 800 and 400 metres of transport users, respectively \citep{Taylor2024Distances, NSWTransportAccessibility}. In these cases, $k$ can be set based on the mode of transport and the number of blocks a person can travel by active modes of transport within a given distance or time. Studying the route-planning problem within this framework also allows us to quantify the cost-benefit trade-off between improving accessibility and optimising route design by varying $k$. For example, planners can assess the effects of implementing lower-cost active transport projects, such as establishing dedicated bike lanes, designing pedestrian-friendly infrastructure (e.g., covering, lighting, and/or paving pathways, establishing shared traffic zones), and promoting active lifestyles, compared with more expensive mass transit options.

Our problem also arises in many other applications, such as last-mile logistics, social network analysis, and security. In last-mile logistics, planners are interested in establishing distribution centres that serve consumers. Our model can be used to select the locations of distribution centres that can be efficiently linked while maximising their coverage. In social network analysis, our proposed model can be used to determine the most efficient way to send a message between users while maximising its potential reach. In security applications, actors with nefarious intentions could use this model to distribute malicious information or compromise software or hardware to achieve maximum impact with minimal chance of detection. On the other hand, government agencies can also use this model to develop effective intervention plans. During disasters, our model can be used to identify locations that can be readily linked to expedite the delivery of relief and emergency supplies.

\subsection{Related literature and problem complexity}

In this subsection, we outline key literature and show that the problem's complexity depends strongly on the centrality measure, restrictions on the feasible set, and the family of graphs under consideration.

Naturally, the concept of centrality has been studied on street networks to identify the most important streets. For example, \citet{CrucittiPaolo2006} conducted a spatial analysis of eighteen urban street networks, examining the distributions of node-based centrality measures, including closeness, betweenness, straightness, and information centrality, to cluster similar cities. Some of these node-based measures can be extended to groups of nodes \citep{Everett1999}, providing tools to answer questions such as which set of nodes is most important in a network \citep{Borgatti2006}. Inherently, these problems fall under the umbrella of facility-location problems, in which central planners seek to develop central distribution hubs to meet demand across the network. However, in many transportation applications, additional constraints may need to be placed on the chosen locations. For example, if planners require a route for a road or railway, the selected nodes must be connected by a path. This has been studied under the `hub line location problem' \citep{deSa2015}. In this paper, the authors aim to find a path that minimizes the total weighted travel time, whereas in our study, we require the path to be the shortest and employ a different centrality measure. More generally, the problem of finding the most central set of nodes in a network has been widely studied. We refer the reader to a survey by \citet{Camus2024}.

The problem we consider is closely related to the classical NP-hard minimum vertex cover problem \citep{GareyJohnson1979}, in which the goal is to find the smallest set of nodes that covers all vertices in the graph. Similarly, the $k$-centre problem, which seeks to select $k$ ``centre'' nodes in a graph to minimize the maximum distance from any vertex to its nearest centre, is NP-hard \citep[p.~47]{Vazirani2003}.

Broader optimisation studies of group centrality measures have been popular in the last decade. For instance, \citet{dawande2012} introduce the portal problem, in which the objective is to find a set of $k$ nodes that maximises the normalised betweenness centrality. The authors show that the portal problem is strongly NP-hard on an arbitrary graph. Later, \citet{Vogiatzis2014} study the problem of finding the most and least central cliques in a graph with respect to group degree, betweenness and closeness centralities. The authors show that this problem is NP-hard regardless of the centrality measure employed, and they provide linear binary programs for all three centrality measures. In a later paper, \citet{Vogiatzis2019} study the problem of finding the most degree-central induced star and show that it is NP-hard. They provide two binary programs and approximation algorithms, and evaluate their performance in protein-protein interaction networks. This work is later extended by \citet{Camur2021}, who utilise decomposition methods to establish faster binary programs.

The first studies in which the set of selected nodes was restricted to form a path were conducted by \citet{hedetniemi1981jordan} and \citet{Slater1982}. In both studies, the authors restricted their analysis to trees and the eccentricity centrality measure. \citet{hedetniemi1981jordan} showed that the Jordan centre (the centre of a graph: either a single node or two connected nodes) can be found in linear time, and \citet{Slater1982} later showed that the most central path can also be found in linear time. \citet{Gomez2023} studies the same problem on different families of graphs for which finding the most central path is known to be NP-hard (families of graphs for which the Hamiltonian path problem is known to be NP-complete), and provides upper bounds for the minimum eccentricity for arbitrary and $k$-connected graphs. In a similar line of research, \citet{Current1989coveringTSP} introduce the NP-hard covering salesman problem, where the goal is to find a minimum-cost tour such that all cities not on the tour are within a certain distance of the tour. \citet{Zeng2019coveringpath} later studied a variant of this problem, termed the covering path problem, on a grid. Instead of finding a tour, they aimed to find the minimum-cost path such that a predetermined subset of nodes lies within a certain distance of the path. The authors provide polynomial-time approximation algorithms based on simple construction rules.

In the studies outlined above, the feasible set is restricted to paths between any pair of nodes. \citet{Matsypura2023} later studied the problem of finding the most degree-central set of nodes under various restrictions on the feasible set. In this paper, the authors restrict the feasible set to shortest paths, induced paths, simple paths, trails, and walks, and show that all cases are NP-hard, except for the polynomial-time solvable case of determining the most degree-central shortest path. \citet{PhosavanhMatsypura2025} extended the latter setting of finding the most central shortest paths in a network to three centrality measures, namely the degree, betweenness, and closeness centrality. The authors provide a faster algorithm for finding the most degree-central shortest path and show that the same problem is NP-hard on graphs with weighted edges. It was also shown that finding the most betweenness-central shortest path is polynomial-time solvable, while finding the most closeness-central shortest path is NP-hard regardless of whether the graph has weighted edges or not.

\subsection{Our contributions and structure of the paper}

In this paper, we restrict the feasible set to all pairwise shortest paths and use $k$-step reach centrality as our centrality measure. The remainder of the paper is organised as follows. In Section~\ref{sec:problem-definition}, we introduce the notation and formally define the problem. Section~\ref{sec:2-step} explores properties of the $2$-step-central shortest path. We generalise these properties in Section~\ref{sec:k-step} for the $k$-step-central shortest path and present a polynomial-time algorithm in Section~\ref{sec:algorithm}. Our theoretical analysis continues with a discussion of the problem's complexity on graphs with weighted edges in Section~\ref{sec:weighted}. We then show how our analysis can be extended to find the closeness-central shortest path in Section~\ref{sec:closeness}. In Section~\ref{sec:numerical}, we conduct extensive numerical experiments with our algorithm on synthetic and real-world urban street networks and provide insights into the results. Finally, Section~\ref{sec:conclusion} concludes the paper.

\section{Problem definition and notation}\label{sec:problem-definition}

Let $G$ be a simple unweighted graph with vertex set $V$ and edge set $E$, which we sometimes write as $G = (V, E)$. Our goal is to find a shortest path in $G$ with the highest $k$-step reach centrality. 

We begin by defining the $k$-step reach centrality of a path. We use $P$ to denote a path in $G$, i.e., a finite sequence of distinct, adjacent vertices in $G$, and write $P = \langle p_0, p_1, \dots, p_\ell\rangle$ for a path starting at $p_0$ and visiting $p_1, \dots, p_\ell$. Let $\mathcal{N}[P]$ denote the \emph{closed neighbourhood} of $P$, i.e., 
\[
\mathcal{N}[P] := \{v \colon (u, v) \in E, u \in P\}\cup P.
\]

Then, the \emph{closed $k$-step neighbourhood} of $P$, $\mathcal{N}_k[P]$, can be defined recursively as 
\[
\mathcal{N}_k[P] = \mathcal{N}[\mathcal{N}_{k-1}[P]],
\]
with $\mathcal{N}_1[P] = \mathcal{N}[P]$. Similarly, let $\mathcal{N}(P)$ denote the \emph{open neighbourhood} of $P$, i.e., 
\[
\mathcal{N}(P) := \mathcal{N}[P] \setminus P.
\]
In words, the open neighbourhood of $P$ is the closed neighbourhood of $P$ minus the path $P$ itself.  To obtain the \emph{open $k$-step neighbourhood} of $P$, we remove the path from its closed $k$-step neighbourhood:
\[
\mathcal{N}_k(P) = \mathcal{N}_k[P] \setminus P.
\]
It then follows that the \emph{open $k$-step reach centrality} of $P$ is the size of its open $k$-step neighbourhood, i.e.,
\[
C_{k}(P) = |\mathcal{N}_k (P)|.
\]
We focus our attention on the open neighbourhood because it is more intuitive. To streamline the exposition, we will drop the word \emph{open} and refer to $\mathcal{N}_k(P)$ and $C_{k}(P)$ as $k$-step neighbourhood and $k$-step reach centrality, respectively. 

Note that for unweighted graphs, the $k$-step neighbourhood of $P$ can also be defined in terms of distance. Let $d(u,v)$ be the distance between nodes $u$ and $v$. Similarly, let $d(u, P)$ be the distance between node $u$ and path $P$, which we define as
\[
d(u, P) = \min_v\{d(u, v) \colon v \in P\}.
\]
Then, the $k$-step neighbourhood of $P$ can be equivalently expressed as
\[
\mathcal{N}_k(P) = \{u \in V \colon d(u, P) \leq k\} \setminus P.
\]
The equivalence implies that the definition of $k$-step reach centrality remains the same.

Then, the problem of finding the $k$-step-central shortest path in $G$ can be written as
\begin{equation}\label{eq:problem}
    \max_P \left\{C_{k}(P) : P \in \mathcal{SP}(G)\right\},
\end{equation}
where $\mathcal{SP}(G)$ is the set of all pairwise shortest paths in $G$. This problem is closely related to the central path problem \citep{Slater1982} and we discuss this connection in detail in Section~\ref{sec:closeness}. 

\section{The 2-step-central shortest path}\label{sec:2-step}

In this section, we consider Problem~\eqref{eq:problem} with $k=2$, i.e., finding the 2-step-central shortest path in a graph. All results in this section can be generalised to arbitrary $k$, as discussed in Section~\ref{sec:k-step}.

Let $P = \langle p_0, \dots, p_\ell\rangle$ be a 2-step-central shortest path of length $\ell$ from $p_0$ to $p_\ell$. For $m \leq \ell$, define the subpaths $P_{m-} := \langle p_0, \dots, p_m\rangle$ and $P_{m+} := \langle p_m, \dots, p_\ell\rangle$. Recall that $d(s,t)$ denotes the distance between $s$ and $t$.

\begin{lemma}\label{lem:2-neighbour-disjoint}
    If $d(s,t) \geq 5$, then the 2-step neighbourhoods of $s$ and $t$ do not overlap.
\end{lemma}
\begin{proof}
    Suppose the 2-step neighbourhoods of $s$ and $t$ overlap. Then $s$ and $t$ are at most four hops apart. That is, if $p_j$ is reachable from $s$ via $p_i$ and from $t$ via $p_{\ell}$, then $\langle s, p_i, p_j, p_\ell, t\rangle$ is a path of length 4 between $s$ and $t$, leading to a contradiction.
\end{proof}
\begin{corollary}\label{cor:2-path-disjoint}
    If $P = \langle p_0, \dots, p_\ell\rangle$ is a shortest path with $\ell \geq 5$ and $0 \leq m \leq \ell-5$, then the 2-step neighbourhoods of $P_{m-}$ and $P_{(m+5)+}$ do not overlap.
\end{corollary}
\begin{lemma}\label{lem:sp2}
    Consider two nodes $s$ and $t$ such that $d(s,t) = \ell$ and $\ell \geq 5$. Let $m$ be a fixed positive integer such that $m \in \{5, \dots, \ell\}$. Assume that there exist two shortest paths between nodes $s$ and $t$, $P'$ and $P''$, that may differ only in the first $m-4$ nodes starting from the node adjacent to $s$. That is,
\begin{align*}
    P' &= \langle p'_0, \ldots, p'_{m-4}, p_{m-3}, p_{m-2}, p_{m-1}, p_m, \dots, p_\ell \rangle, \\
    P'' &= \langle p''_0, \ldots, p''_{m-4}, p_{m-3}, p_{m-2}, p_{m-1}, p_m, \dots, p_\ell \rangle,
\end{align*}
where $s = p'_0 = p''_0$ and $t = p_k$.
Then the shortest path $P'$ is at least as central as the shortest path $P''$ if and only if the
subpath $P'_{m-}$ is at least as central as the subpath $P''_{m-}$, i.e.,
$|\mathcal{N}_2(P')|\geq |\mathcal{N}_2(P'')|$ if and only if $|\mathcal{N}_2(P'_{m-})| \geq |\mathcal{N}_2(P''_{m-})|$.
\end{lemma}
This lemma allows us to reduce the size of the feasible set when constructing 2-step-central shortest paths. For any set of paths that share the same four last nodes, we need to keep only the one with the largest 2-step neighbourhood.
\begin{proof}
    Let $P$ be an arbitrary shortest path of length $\ell$. We can decompose the path into three segments: $P_{(m-4)-}$, $P_{(m+1)+}$ and $\langle p_{m-3}, p_{m-2}, p_{m-1}, p_m\rangle$.  By Corollary~\ref{cor:2-path-disjoint}, note that $\mathcal{N}_2(P_{(m-4)-}) \cap \mathcal{N}_2(P_{(m+1)+}) = \emptyset$.
    
    Observe that we can express the 2-step neighbourhood of $P$ as
    \[
    \mathcal{N}_2(P) = (\mathcal{N}_2(P_{(m-4)-}) \cup \mathcal{N}_2(\langle p_{m-3}, p_{m-2}, p_{m-1},p_{m}\rangle) \cup \mathcal{N}_2(P_{(m+1)+}))\setminus P.
    \]    
    Then, we have
    \begin{align*}
        |\mathcal{N}_2(P)| &= |\mathcal{N}_2(P_{(m-4)-})| + |\mathcal{N}_2(\langle p_{m-3}, p_{m-2}, p_{m-1}, p_{m}\rangle)| + |\mathcal{N}_2(P_{(m+1)+})| \\
        & \quad - |\mathcal{N}_2(P_{(m-4)-}) \cap \mathcal{N}_2(\langle p_{m-3}, p_{m-2}, p_{m-1}, p_{m}\rangle)| \\
        & \quad - |\mathcal{N}_2(P_{(m+1)+}) \cap \mathcal{N}_2(\langle p_{m-3}, p_{m-2}, p_{m-1}, p_{m}\rangle)| - 8,
    \end{align*}
    and as
    \begin{align*}
        \mathcal{N}_2(P_{m-}) &= \left(\mathcal{N}_2(P_{(m-4)-})\setminus \{p_{m-3}, p_{m-2}\} \right) \cup \left(\mathcal{N}_2(\langle p_{m-3}, p_{m-2}, p_{m-1}, p_{m}\rangle)\setminus \{p_{m-5},p_{m-4}\}\right),
    \end{align*}
    we also have
    \begin{align*}
        |\mathcal{N}_2(P_{m-})| &= |\mathcal{N}_2(P_{(m-4)-})| + |\mathcal{N}_2(\langle p_{m-3}, p_{m-2}, p_{m-1}, p_{m}\rangle)| \\
        & \quad - |\mathcal{N}_2(P_{(m-4)-})\cap \mathcal{N}_2(\langle p_{m-3}, p_{m-2}, p_{m-1}, p_{m}\rangle)| - 4.
    \end{align*}

    Combining these two expressions, we have
    \begin{align*}
        |\mathcal{N}_2(P)| &= |\mathcal{N}_2(P_{m-})| + |\mathcal{N}_2(P_{(m+1)+})| - |\mathcal{N}_2(P_{(m+1)+}) \cap \mathcal{N}_2(\langle p_{m-3}, p_{m-2}, p_{m-1}, p_{m}\rangle)| - 4.
    \end{align*}

    Since this holds for an arbitrary shortest path $P$, replacing $P$ with $P'$ and $P''$ yields the desired result.
\end{proof}

The above result allows us to build a polynomial-time algorithm to find the 2-step-central shortest path, which we present in Section~\ref{sec:algorithm}.

\section{\texorpdfstring{The $k$-step-central shortest path}{The k-step-central shortest path}}\label{sec:k-step}

Before we present an algorithm for solving Problem~\eqref{eq:problem} for any $k$, we provide generalisations of the lemmas and corollaries from Section~\ref{sec:2-step}. For brevity, we omit the proofs, as they are direct extensions of the previous section.

\begin{lemma}\label{lem:k-neighbour-disjoint}
    If $d(s,t) \geq 2k+1$, then the $k$-step neighbourhoods of $s$ and $t$ do not overlap.
\end{lemma}
\begin{corollary}\label{cor:k-path-disjoint}
    If $P = \langle p_0, \dots, p_\ell\rangle$ is a shortest path with $\ell \geq 2k+1$ and $0 \leq m \leq \ell-2k-1$, then the $k$-step neighbourhoods of $P_{m-}$ and $P_{(m+2k+1)+}$ do not overlap.
\end{corollary}

\begin{lemma}\label{lem:spk}
    Consider two nodes $s$ and $t$ such that $d(s,t) = \ell$ and $\ell \geq 2k+1$. Let $m$ be a fixed positive integer such that $m \in \{2k+1, \dots, \ell\}$. Assume that there exist two shortest paths between nodes $s$ and $t$, $P'$ and $P''$, that may differ only in the first $m-2k$ nodes starting from the node adjacent to $s$. That is,
    \begin{align*}
        P' &= \langle p'_0, \ldots, p'_{m-2k}, p_{m-2k+1}, \dots, p_m, \dots, p_\ell \rangle, \\
        P'' &= \langle p''_0, \ldots, p''_{m-2k}, p_{m-2k+1}, \dots, p_m, \dots, p_\ell \rangle,
    \end{align*}
    where $s = p'_0 = p''_0$ and $t = p_k$.
    Then the shortest path $P'$ is at least as central as the shortest path $P''$ if and only if the
    subpath $P'_{m-}$ is at least as central as the subpath $P''_{m-}$, i.e.,
    $|\mathcal{N}_k(P')|\geq |\mathcal{N}_k(P'')|$ if and only if $|\mathcal{N}_k(P'_{m-})| \geq |\mathcal{N}_k(P''_{m-})|$.
\end{lemma}

Note that in the analysis above, the choice of $k$ plays a significant role. For a given graph, increasing $k$ can quickly lead to oversaturation of the $k$-step neighbourhood. To see this, consider the graph in Figure~\ref{fig:saturated-graphs}. For any $k\geq 2$, the path (the singular blue square node) will be a $k$-step-central shortest path. 

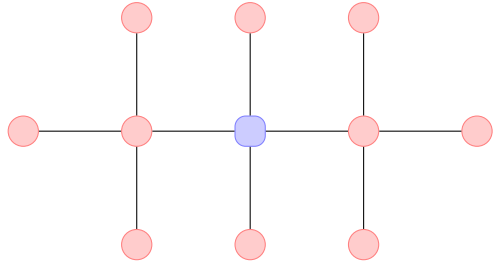
\begin{figure}[!htb]
    \centering
    \begin{tikzpicture}[scale=1.5, square/.style={rectangle, rounded corners}]
    \footnotesize
    \begin{scope}[circle,minimum size=4mm]
        \node at (0,  0) [draw=red!50,fill=red!20] (1)  {};
        \node at (1,  0) [draw=red!50,fill=red!20] (2)  {};
        \node at (2,  0) [draw=blue!50,fill=blue!20, square] (3)  {};
        \node at (3,  0) [draw=red!50,fill=red!20] (4)  {};
        \node at (4,  0) [draw=red!50,fill=red!20] (5)  {};
        \node at (1,  1) [draw=red!50,fill=red!20] (6)  {};
        \node at (2,  1) [draw=red!50,fill=red!20] (7)  {};
        \node at (3,  1) [draw=red!50,fill=red!20] (8)  {};
        \node at (1, -1) [draw=red!50,fill=red!20] (9)  {};
        \node at (2, -1) [draw=red!50,fill=red!20] (10) {};
        \node at (3, -1) [draw=red!50,fill=red!20] (11) {};
    \end{scope}
    
    \path 
    (1) edge (2)
    (2) edge (3)
    (3) edge (4)
    (4) edge (5)
    (2) edge (6)
    (2) edge (9)
    (3) edge (7)
    (3) edge (10)
    (4) edge (8)
    (4) edge (11);
\end{tikzpicture}
    \caption{Example of a graph where the $k$-step-central shortest path is the same when $k$ is increased. For any $k\geq 2$, the blue square node is the $k$-step-central shortest path.}
    \label{fig:saturated-graphs}
\end{figure}

However, this observation does not imply tighter bounds on the length of the $k$-step-central shortest path in a graph. Regardless of the choice of $k$, we can construct a graph in which the $k$-step-central shortest path is a diametral path. To demonstrate this, for a given value of $k$, we first create a wheel graph with $4(k + 1)$ spokes. We then split each edge incident to the central node into $k + 1$ edges. Examples of these graphs are shown in Figure~\ref{fig:diameter-solutions}, with Figure~\ref{fig:diameter-solutions-1} and \ref{fig:diameter-solutions-2} depicting the graphs for $k = 1$ and $k = 2$, respectively. In these graphs, the blue square nodes indicate a $k$-step-central shortest path, which is also a diametral path. The red circle nodes are in the $k$-step neighbourhood, while the orange triangle nodes are not covered. It is clear that the path cannot be extended to additional nodes, as doing so would no longer yield the shortest path. Additionally, the path cannot be shortened, as removing a node from either end would lose two neighbours while gaining only the removed node, resulting in a net loss of one neighbour.

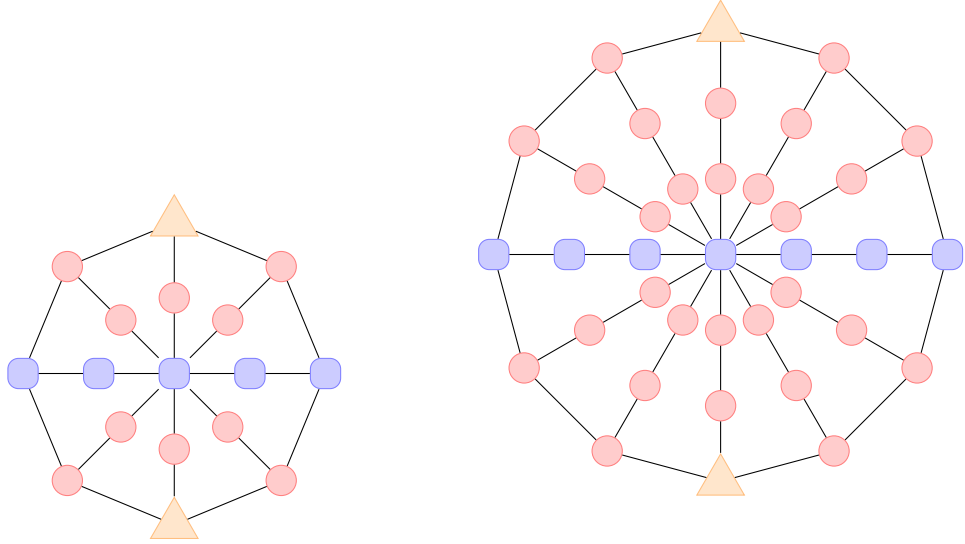
\begin{figure}[!htb]
    \centering
    \begin{subfigure}[b]{0.4\textwidth}
        \centering
        \begin{tikzpicture}[scale=0.5, square/.style={rectangle, rounded corners}, triangle/.style = {regular polygon, regular polygon sides=3},]
    \begin{scope}[circle,minimum size=4mm]
        \draw (0, 0)           node[draw=blue!50, fill=blue!20, square]       (0)  {}; 
        \draw (2.0, 0.0)       node[draw=blue!50, fill=blue!20, square]       (0-0){}; 
        \draw (1.414, 1.414)   node[draw=red!50, fill=red!20]                 (0-1){}; 
        \draw (0.0, 2.0)       node[draw=red!50, fill=red!20]                 (0-2){}; 
        \draw (-1.414, 1.414)  node[draw=red!50, fill=red!20]                 (0-3){}; 
        \draw (-2.0, 0.0)      node[draw=blue!50, fill=blue!20, square]       (0-4){}; 
        \draw (-1.414, -1.414) node[draw=red!50, fill=red!20]                 (0-5){}; 
        \draw (-0.0, -2.0)     node[draw=red!50, fill=red!20]                 (0-6){}; 
        \draw (1.414, -1.414)  node[draw=red!50, fill=red!20]                 (0-7){}; 
        \draw (4.0, 0.0)       node[draw=blue!50, fill=blue!20, square]       (1-0){}; 
        \draw (2.828, 2.828)   node[draw=red!50, fill=red!20]                 (1-1){}; 
        \draw (0.0, 4.0)       node[draw=orange!50, fill=orange!20, triangle] (1-2){}; 
        \draw (-2.828, 2.828)  node[draw=red!50, fill=red!20]                 (1-3){}; 
        \draw (-4.0, 0.0)      node[draw=blue!50, fill=blue!20, square]       (1-4){}; 
        \draw (-2.828, -2.828) node[draw=red!50, fill=red!20]                 (1-5){}; 
        \draw (-0.0, -4.0)     node[draw=orange!50, fill=orange!20, triangle] (1-6){}; 
        \draw (2.828, -2.828)  node[draw=red!50, fill=red!20]                 (1-7){}; 
    \end{scope}
    \begin{scope}[-]
        \draw (0)   to (0-0);
        \draw (0)   to (0-1);
        \draw (0)   to (0-2);
        \draw (0)   to (0-3);
        \draw (0)   to (0-4);
        \draw (0)   to (0-5);
        \draw (0)   to (0-6);
        \draw (0)   to (0-7);
        \draw (0-0) to (1-0);
        \draw (0-1) to (1-1);
        \draw (0-2) to (1-2);
        \draw (0-3) to (1-3);
        \draw (0-4) to (1-4);
        \draw (0-5) to (1-5);
        \draw (0-6) to (1-6);
        \draw (0-7) to (1-7);
        \draw (1-0) to (1-1);
        \draw (1-0) to (1-7);
        \draw (1-1) to (1-2);
        \draw (1-2) to (1-3);
        \draw (1-3) to (1-4);
        \draw (1-4) to (1-5);
        \draw (1-5) to (1-6);
        \draw (1-6) to (1-7);
    \end{scope}
\end{tikzpicture}
        \caption{Example of a graph where a diametral path has the maximum $1$-step reach (degree) centrality.}
        \label{fig:diameter-solutions-1}
    \end{subfigure}%
    \qquad 
    \begin{subfigure}[b]{0.4\textwidth}
        \centering
        \begin{tikzpicture}[scale=0.5, square/.style={rectangle, rounded corners}, triangle/.style = {regular polygon, regular polygon sides=3},]
    \begin{scope}[circle,minimum size=4mm]
        \draw (0, 0)         node[draw=blue!50, fill=blue!20, square]       (0)   {};
        \draw (2.0, 0.0)     node[draw=blue!50, fill=blue!20, square]       (0-0) {};
        \draw (1.732, 1.0)   node[draw=red!50, fill=red!20]                 (0-1) {};
        \draw (1.0, 1.732)   node[draw=red!50, fill=red!20]                 (0-2) {};
        \draw (0.0, 2.0)     node[draw=red!50, fill=red!20]                 (0-3) {};
        \draw (-1.0, 1.732)  node[draw=red!50, fill=red!20]                 (0-4) {};
        \draw (-1.732, 1.0)  node[draw=red!50, fill=red!20]                 (0-5) {};
        \draw (-2.0, 0.0)    node[draw=blue!50, fill=blue!20, square]       (0-6) {};
        \draw (-1.732, -1.0) node[draw=red!50, fill=red!20]                 (0-7) {};
        \draw (-1.0, -1.732) node[draw=red!50, fill=red!20]                 (0-8) {};
        \draw (-0.0, -2.0)   node[draw=red!50, fill=red!20]                 (0-9) {};
        \draw (1.0, -1.732)  node[draw=red!50, fill=red!20]                 (0-10){};
        \draw (1.732, -1.0)  node[draw=red!50, fill=red!20]                 (0-11){};
        \draw (4.0, 0.0)     node[draw=blue!50, fill=blue!20, square]       (1-0) {};
        \draw (3.464, 2.0)   node[draw=red!50, fill=red!20]                 (1-1) {};
        \draw (2.0, 3.464)   node[draw=red!50, fill=red!20]                 (1-2) {};
        \draw (0.0, 4.0)     node[draw=red!50, fill=red!20]                 (1-3) {};
        \draw (-2.0, 3.464)  node[draw=red!50, fill=red!20]                 (1-4) {};
        \draw (-3.464, 2.0)  node[draw=red!50, fill=red!20]                 (1-5) {};
        \draw (-4.0, 0.0)    node[draw=blue!50, fill=blue!20, square]       (1-6) {};
        \draw (-3.464, -2.0) node[draw=red!50, fill=red!20]                 (1-7) {};
        \draw (-2.0, -3.464) node[draw=red!50, fill=red!20]                 (1-8) {};
        \draw (-0.0, -4.0)   node[draw=red!50, fill=red!20]                 (1-9) {};
        \draw (2.0, -3.464)  node[draw=red!50, fill=red!20]                 (1-10){};
        \draw (3.464, -2.0)  node[draw=red!50, fill=red!20]                 (1-11){};
        \draw (6.0, 0.0)     node[draw=blue!50, fill=blue!20, square]       (2-0) {};
        \draw (5.196, 3.0)   node[draw=red!50, fill=red!20]                 (2-1) {};
        \draw (3.0, 5.196)   node[draw=red!50, fill=red!20]                 (2-2) {};
        \draw (0.0, 6.0)     node[draw=orange!50, fill=orange!20, triangle] (2-3) {};
        \draw (-3.0, 5.196)  node[draw=red!50, fill=red!20]                 (2-4) {};
        \draw (-5.196, 3.0)  node[draw=red!50, fill=red!20]                 (2-5) {};
        \draw (-6.0, 0.0)    node[draw=blue!50, fill=blue!20, square]       (2-6) {};
        \draw (-5.196, -3.0) node[draw=red!50, fill=red!20]                 (2-7) {};
        \draw (-3.0, -5.196) node[draw=red!50, fill=red!20]                 (2-8) {};
        \draw (-0.0, -6.0)   node[draw=orange!50, fill=orange!20, triangle] (2-9) {};
        \draw (3.0, -5.196)  node[draw=red!50, fill=red!20]                 (2-10){};
        \draw (5.196, -3.0)  node[draw=red!50, fill=red!20]                 (2-11){};
    \end{scope}
    \begin{scope}[-]
        \draw (0)    to (0-0);
        \draw (0)    to (0-1);
        \draw (0)    to (0-2);
        \draw (0)    to (0-3);
        \draw (0)    to (0-4);
        \draw (0)    to (0-5);
        \draw (0)    to (0-6);
        \draw (0)    to (0-7);
        \draw (0)    to (0-8);
        \draw (0)    to (0-9);
        \draw (0)    to (0-10);
        \draw (0)    to (0-11);
        \draw (0-0)  to (1-0);
        \draw (0-1)  to (1-1);
        \draw (0-2)  to (1-2);
        \draw (0-3)  to (1-3);
        \draw (0-4)  to (1-4);
        \draw (0-5)  to (1-5);
        \draw (0-6)  to (1-6);
        \draw (0-7)  to (1-7);
        \draw (0-8)  to (1-8);
        \draw (0-9)  to (1-9);
        \draw (0-10) to (1-10);
        \draw (0-11) to (1-11);
        \draw (1-0)  to (2-0);
        \draw (1-1)  to (2-1);
        \draw (1-2)  to (2-2);
        \draw (1-3)  to (2-3);
        \draw (1-4)  to (2-4);
        \draw (1-5)  to (2-5);
        \draw (1-6)  to (2-6);
        \draw (1-7)  to (2-7);
        \draw (1-8)  to (2-8);
        \draw (1-9)  to (2-9);
        \draw (1-10) to (2-10);
        \draw (1-11) to (2-11);
        \draw (2-0)  to (2-1);
        \draw (2-0)  to (2-11);
        \draw (2-1)  to (2-2);
        \draw (2-2)  to (2-3);
        \draw (2-3)  to (2-4);
        \draw (2-4)  to (2-5);
        \draw (2-5)  to (2-6);
        \draw (2-6)  to (2-7);
        \draw (2-7)  to (2-8);
        \draw (2-8)  to (2-9);
        \draw (2-9)  to (2-10);
        \draw (2-10) to (2-11);
    \end{scope}
\end{tikzpicture}
        \caption{Example of a graph where a diametral path has the maximum $2$-step reach centrality.}
        \label{fig:diameter-solutions-2}
    \end{subfigure}
    \caption{For any $k$, a diametral path may be a $k$-step-central shortest path. The blue square nodes indicate the $k$-step-central shortest path, the red circular nodes indicate the neighbourhood of the path, and the orange triangle nodes indicate nodes that are not covered.}
    \label{fig:diameter-solutions}
\end{figure}

\section{\texorpdfstring{Finding the $k$-step-central shortest path}{Finding the k-step-central shortest path.}}\label{sec:algorithm}

We now present a polynomial-time algorithm to solve Problem~\eqref{eq:problem}. Algorithm~\ref{alg:algorithm} generalises the algorithm provided by \citet{Matsypura2023}, which relies on breadth-first search. We prefer this as the basis of our algorithm because it is easier to generalise to higher values of $k$ than the algorithm provided by \citet{PhosavanhMatsypura2025}, even though the latter is more efficient.

We first compute all pairwise distances in $G$. Next, given a source node $s$ and a destination node $t$, we compute $d(s, t)-1$ layers. Each layer $L(\ell)$, for $\ell = 1, \dots, d(s, t) - 1$, contains nodes at distance $\ell$ from $s$ and $d(s, t) - \ell$ from $t$.

The next stage of the algorithm uses breadth-first search. At each stage, we extend existing shortest paths of length $\ell-1$ by one node to obtain shortest paths of length $\ell$, as shown in lines 10-15 of the pseudocode. Without intervention, this process may generate an exponential number of shortest paths that must be stored in memory. To prevent this, we add a pruning step that relies on Lemma~\ref{lem:spk}. After more than $2k$ layers have been processed, for each set of paths that share the last $2k$ nodes, our pruning procedure retains only the path with the highest $k$-step reach centrality. This is shown in lines 16-26 of the pseudocode.

\begin{algorithm}[!htbt]
    \caption{Finding the $k$-step-central shortest path (kStep-CSP)}\label{alg:algorithm}
    \begin{algorithmic}[1]
        \Input Graph $G = (V, E)$, source node $s \in V$ and destination node $t \in V$.
        \Output $k$-step-central shortest path $P^\star$ between $s$ and $t$.
        \State compute $d(s, q)$ and $d(q, t)$ for all $q\in V$ \Comment{compute pairwise distances}
        \State $m \gets d(s, t)$
        \For{$\ell = 1, \dots, m-1$}\Comment{compute layers}
            \State $L(\ell) \gets \{q\in V:\ d(s, q) = \ell$ and $d(q, t) = m-\ell\}$
        \EndFor
        \State $\mathcal{P}_0\gets\{\langle s\rangle \}$
        \State ${p}_0 \gets s$
        \For{$\ell=1, \dots, m$} \Comment{construct shortest paths}
            \State $\mathcal{P_{\ell}}= \emptyset$
            \For{$P=\langle p_0,p_1,\ldots,p_{\ell-1}\rangle \in \mathcal{P}_{\ell-1}$}
                \For{$q \in \mathcal{N}(p_{\ell-1})\cap L(\ell)$}\Comment{extend shortest paths of length $\ell-1$ to layer $\ell$}
                    \State $\mathcal{P}_\ell \gets \mathcal{P}_\ell \cup \{\langle p_0,p_1,\ldots,p_{\ell-1},q\rangle\}$
                    \State compute and store $\mathcal{N}_k(\langle p_0,p_1,\ldots,p_{\ell-1},q\rangle)$
                \EndFor
            \EndFor
            \If{$\ell \geq 2k + 1$}\Comment{pruning}
                \State $\mathcal{P}^{pr}_{\ell} \gets \emptyset$
                \For{$P \in \mathcal{P}_{\ell}$}
                    \If{$\mathcal{P}^{pr}_{\ell}$ has no path, say $P^{pr}$, that has the same last $2k$ nodes as $P$}
                        \State $\mathcal{P}^{pr}_{\ell}\leftarrow \mathcal{P}^{pr}_{\ell} \cup \{P\}$
                    \ElsIf{$|\mathcal{N}_k(P)|> |\mathcal{N}_k(P^{pr})|$}
                        \State $\mathcal{P}^{pr}_{\ell}\leftarrow \left\{\mathcal{P}^{pr}_{\ell}\setminus\{P^{pr}\}\right\} \cup \{P\}$
                    \EndIf  
                \EndFor
                \State $\mathcal{P}_\ell\leftarrow\mathcal{P}^{pr}_\ell$
            \EndIf
        \EndFor
        \State $P^\star \leftarrow \arg\max\{ |\mathcal{N}_k(P)| \ :\ P\in \mathcal{P}_{m}\}$
    \end{algorithmic}
\end{algorithm}

\begin{lemma}
    Given a pair of nodes $s$ and $t$ with $d(s, t) = m$, Algorithm~\ref{alg:algorithm} returns a $k$-step-central shortest path in $O(m|V|^{2k+2})$ time.
\end{lemma}

\begin{proof}
    Computing all pairwise distances takes $O(|V|^3)$ time using an all-pairs shortest path algorithm such as the Floyd-Warshall algorithm, and the layers can be computed in $O(k|V|)$ time.

    The for loop in line 8 is executed $m$ times. In this for-loop, we need to extend paths in $\mathcal{P}_{\ell - 1}$. 
    
    To evaluate the time complexity of lines 10-15, consider two cases:
    \begin{itemize}
        \item When $\ell < 2k + 1$, $|\mathcal{P}_{\ell - 1}| = O(|V|^{\ell - 1})$ as this set consists of paths containing $\ell$ nodes starting at $s$. 
        \item When $\ell \geq 2k + 1$, $|\mathcal{P}_{\ell - 1}| = O(|V|^{2k})$ as by Lemma~\ref{lem:spk}, we only need to keep one path where the last $2k$ nodes are the same.
    \end{itemize}
    This means that the for-loop in line 10 is executed $O(|V|^{2k})$ times. The for-loop in line 11 is executed $O(|V|)$ times. In this for-loop, we need to compute the $k$-step neighbourhood of the newly extended path. Assuming that the neighbourhoods of paths $p_{\ell - 1}$ are stored, computing $\mathcal{N}_k(\langle p_0, p_1, \dots, p_{\ell-1}, q\rangle)$ can be completed in $O(|V|)$ time because we are taking the union of two sets. This means that the running time of lines 10-15 is $O(|V|^{2k+2})$.

    The pruning procedure in lines 16-26 takes $O(|V|^{2k + 1})$ time. This is because before pruning, $|\mathcal{P}_\ell| = O(|V|^{2k + 1})$ as $|\mathcal{P}_{\ell-1}| = O(|V|^{2k})$, and each path in $\mathcal{P}_{\ell-1}$ can be extended to at most $|V|$ nodes. This means that the for-loop in line 19 is executed $O(|V|^{2k + 1})$ times, and each iteration is completed in constant time, as we are checking the neighbourhood of paths with the same last $2k$ nodes. This implies an overall running time of $O(m|V|^{2k+2})$.
\end{proof}

\begin{theorem}\label{thm:problem-running-time}
    Problem~\eqref{eq:problem} on an unweighted graph can be solved in $O(m|V|^{2k+4})$ time, where $m$ is the diameter of the graph.
\end{theorem}
\begin{proof}
    To solve Problem~\eqref{eq:problem}, we need to apply Algorithm~\ref{alg:algorithm} $O(|V|^2)$ times over all possible start and end node pairs, each time taking $O(m|V|^{2k+2})$ time. Overall, this leads to a running time of $O(m|V|^{2k+4})$.
\end{proof}
\begin{remark}
    Note that when solving Problem~\eqref{eq:problem}, line 1 of Algorithm~\ref{alg:algorithm} needs to be executed only once if an all-pairs shortest path algorithm is used and the results are stored.
\end{remark}

\section{The case of weighted graphs}\label{sec:weighted}

Let us now assume that $G$ has weighted edges and show that Problem~\eqref{eq:problem} is NP-hard. To do so, we first define the decision variant of Problem~\eqref{eq:problem} on a graph with weighted edges (kStep-CSP-D). 
\begin{dproblem}
  \problemtitle{kStep-CSP-D}
  \probleminstance{Graph $G = (V, E)$ with weighted edges and positive integers $k$ and $\alpha$.}
  \problemquestion{Is there a shortest path between two vertices in $G$ with a $k$-step reach centrality of at least $\alpha?$}
\end{dproblem}

We also define the satisfiability problem, which is well-known to be NP-complete \citep{GareyJohnson1979}.
\begin{dproblem}
  \problemtitle{SAT}
  \probleminstance{A conjunctive normal form  (CNF) formula consisting of a set of $U$ variables, a collection $C$ of clauses over $U$.}
  \problemquestion{Is there a truth assignment for $U$ that simultaneously satisfies all of the clauses in $C$?}
\end{dproblem}

Note that although the graph under consideration has weighted edges, the definition of $k$-step reach centrality remains unchanged, i.e., the $k$-step reach is measured by hop distance, not by weighted distance.

\begin{lemma}\label{lem:weighted-np-complete}
    kStep-CSP-D is NP-complete.
\end{lemma}

\begin{proof}
   It is clear that kStep-CSP-D is in NP because the $k$-step neighbourhood can be computed in polynomial time.

    Let $n$ be the number of variables and $m$ the number of clauses, i.e., $n = |U|$, $m = |C|$. We provide a reduction of SAT to kStep-CSP-D by constructing a graph $G$ as follows. First, create two directed paths $\langle x_1, \dots, x_{n} \rangle$ and $\langle \overline{x}_1, \dots, \overline{x}_{n} \rangle$. Additionally, add edges $(x_i, \overline{x}_{i+1})$ and $(\overline{x}_i, x_{i+1})$ for $i = 1, \dots, n - 1$. We also construct and connect the following paths:
    \begin{itemize}
        \item $\langle s_1, \dots, s_k\rangle$, with $s_1$ connected to $x_1$ and $\overline{x}_1$,
        \item $\langle \overline{s}_1, \dots, \overline{s}_k\rangle$, with $\overline{s}_1$ connected to $x_1$ and $\overline{x}_1$,
        \item $\langle t_1, \dots, t_k\rangle$, with $t_1$ connected to $x_{n}$ and $\overline{x}_{n}$, and
        \item $\langle \overline{t}_1, \dots, \overline{t}_k\rangle$, with $\overline{t}_1$ connected to $x_{n}$ and $\overline{x}_{n}$.
    \end{itemize}
    In the constructions above, all edges have unit weight.

    Then, we add nodes and edges for each clause. For each clause $c \in C$, we first construct a path $\langle y_{c, 1}, \dots, y_{c, k}\rangle$ with unit-weight edges. We then connect $y_{c, 1}$ to its corresponding variables with an edge of weight $n$. An example of the graph for $k=2$ is shown in Figure~\ref{fig:np-hardness-proof}.
    
    \begin{figure}[!htb]
        \centering
        \begin{tikzpicture}[scale=1.2, square/.style={rectangle, rounded corners}]
    \footnotesize
    \begin{scope}[circle, minimum size=10mm, inner sep=-0.2em]
        \draw
        (0, 0.5) node[draw=blue!50, fill=blue!20] (s){$s_1$}
        (0, -0.5) node[draw=blue!50, fill=blue!20] (-s){$\overline{s}_1$}
        (-1.2, 0.5) node[draw=blue!50, fill=blue!20] (s1){$s_2$}
        (-1.2, -0.5) node[draw=blue!50, fill=blue!20] (-s1){$\overline{s}_2$}
        (1.2, -0.5) node[draw=blue!50, fill=blue!20] (-1){$\overline{x}_{1}$}
        (1.2, 0.5) node[draw=blue!50, fill=blue!20] (1){$x_{1}$}
        (2.4, -0.5) node[draw=blue!50, fill=blue!20] (-2){$\overline{x}_{2}$}
        (2.4, 0.5) node[draw=blue!50, fill=blue!20] (2){$x_{2}$}
        (3.6, -0.5) node[draw=blue!50, fill=blue!20] (-3){$\overline{x}_{3}$}
        (3.6, 0.5) node[draw=blue!50, fill=blue!20] (3){$x_{3}$}
        (4.8, -0.5) node[draw=blue!50, fill=blue!20] (-4){$\overline{x}_{4}$}
        (4.8, 0.5) node[draw=blue!50, fill=blue!20] (4){$x_{4}$}
        (7.2, -0.5) node[draw=blue!50, fill=blue!20] (-6){$\overline{x}_{n-1}$}
        (7.2, 0.5) node[draw=blue!50, fill=blue!20] (6){$x_{n-1}$}
        (8.4, -0.5) node[draw=blue!50, fill=blue!20] (-7){$\overline{x}_{n}$}
        (8.4, 0.5) node[draw=blue!50, fill=blue!20] (7){$x_{n}$}
        (9.6, 0.5) node[draw=blue!50, fill=blue!20] (t){$t_1$}
        (9.6, -0.5) node[draw=blue!50, fill=blue!20] (-t){$\overline{t}_1$}
        (10.8, 0.5) node[draw=blue!50, fill=blue!20] (t1){$t_2$}
        (10.8, -0.5) node[draw=blue!50, fill=blue!20] (-t1){$\overline{t}_2$}
        (2.4, 1.75) node[draw=red!50, fill=red!20, square, inner sep=-0.2em] (1000){$y_{1,1}$}
        (2.4, -1.75) node[draw=red!50, fill=red!20, square, inner sep=-0.2em] (1001){$y_{2,1}$}
        (4.8, 1.75) node[draw=red!50, fill=red!20, square, inner sep=-0.2em] (1002){$y_{3,1}$}
        (4.8, -1.75) node[draw=red!50, fill=red!20, square, inner sep=-0.2em] (1003){$y_{4,1}$}
        (7.2, 1.75) node[draw=red!50, fill=red!20, square, inner sep=-0.2em] (1004){$y_{m-1,1}$}
        (7.2, -1.75) node[draw=red!50, fill=red!20, square, inner sep=-0.2em] (1005){$y_{m,1}$}
        (2.4, 3) node[draw=red!50, fill=red!20, square, inner sep=-0.2em] (10001){$y_{1,2}$}
        (2.4, -3) node[draw=red!50, fill=red!20, square, inner sep=-0.2em] (10011){$y_{2,2}$}
        (4.8, 3) node[draw=red!50, fill=red!20, square, inner sep=-0.2em] (10021){$y_{3,2}$}
        (4.8, -3) node[draw=red!50, fill=red!20, square, inner sep=-0.2em] (10031){$y_{4,2}$}
        (7.2, 3) node[draw=red!50, fill=red!20, square, inner sep=-0.2em] (10041){$y_{m-1,2}$}
        (7.2, -3) node[draw=red!50, fill=red!20, square, inner sep=-0.2em] (10051){$y_{m,2}$};
    \end{scope}
    
    \node at ($(4)!.5!(6)$) {\ldots};
    \node at ($(-4)!.5!(-6)$) {\ldots};
    \node at ($(4)!.5!(-6)$) {\ldots};
    

    \begin{scope}[dashed, -Stealth, thick]
        \draw[red] (-1) to (1004);
        \draw[red] (-2) to (1003);
        \draw[red] (-3) to (1001);
        \draw[red] (-6) to (1002);
        \draw[red] (-6) to (1005);
        \draw[red] (-7) to (1004);
        \draw[red] (-7) to (1005);
        \draw[red] (1) to (1000);
        \draw[red] (1) to (1001);
        \draw[red] (2) to (1003);
        \draw[red] (3) to (1000);
        \draw[red] (-6) to (1000);
        \draw[red] (-6) to (1003);
        \draw[red] (4) to (1002);
        \draw[red] (2) to (1002);
    \end{scope}
    \begin{scope}[-Stealth, thick]
        \draw[blue] (-1) to (2);
        \draw[blue] (-1) to (-2);
        \draw[blue] (1) to (-2);
        \draw[blue] (-2) to (3);
        \draw[blue] (-2) to (-3);
        \draw[blue] (2) to (-3);
        \draw[blue] (-3) to (4);
        \draw[blue] (-3) to (-4);
        \draw[blue] (3) to (-4);
        \draw[blue] (-6) to (7);
        \draw[blue] (-6) to (-7);
        \draw[blue] (6) to (-7);
        \draw[blue] (1) to (2);
        \draw[blue] (2) to (3);
        \draw[blue] (3) to (4);
        \draw[blue] (6) to (7);
        \draw[blue] (6) to (7);
        \draw[blue] (7) to (t);
        \draw[blue] (-7) to (t);
        \draw[blue] (7) to (-t);
        \draw[blue] (-7) to (-t);
        \draw[blue] (1) to (s);
        \draw[blue] (-1) to (s);
        \draw[blue] (1) to (-s);
        \draw[blue] (-1) to (-s);
        \draw[blue] (t) to (t1);
        \draw[blue] (-t) to (-t1);
        \draw[blue] (s) to (s1);
        \draw[blue] (-s) to (-s1);
        \draw[red] (1000) to (10001);
        \draw[red] (1001) to (10011);
        \draw[red] (1002) to (10021);
        \draw[red] (1003) to (10031);
        \draw[red] (1004) to (10041);
        \draw[red] (1005) to (10051);
    \end{scope}
\end{tikzpicture}
        \caption{Example of a graph $G$ used in Lemma~\ref{lem:weighted-np-complete}. The blue circular nodes indicate the base structure, and the red square nodes indicate the nodes associated with the clauses. The dashed lines indicate edges with weight $n$, and the remaining edges have unit weight.}
        \label{fig:np-hardness-proof}
    \end{figure}
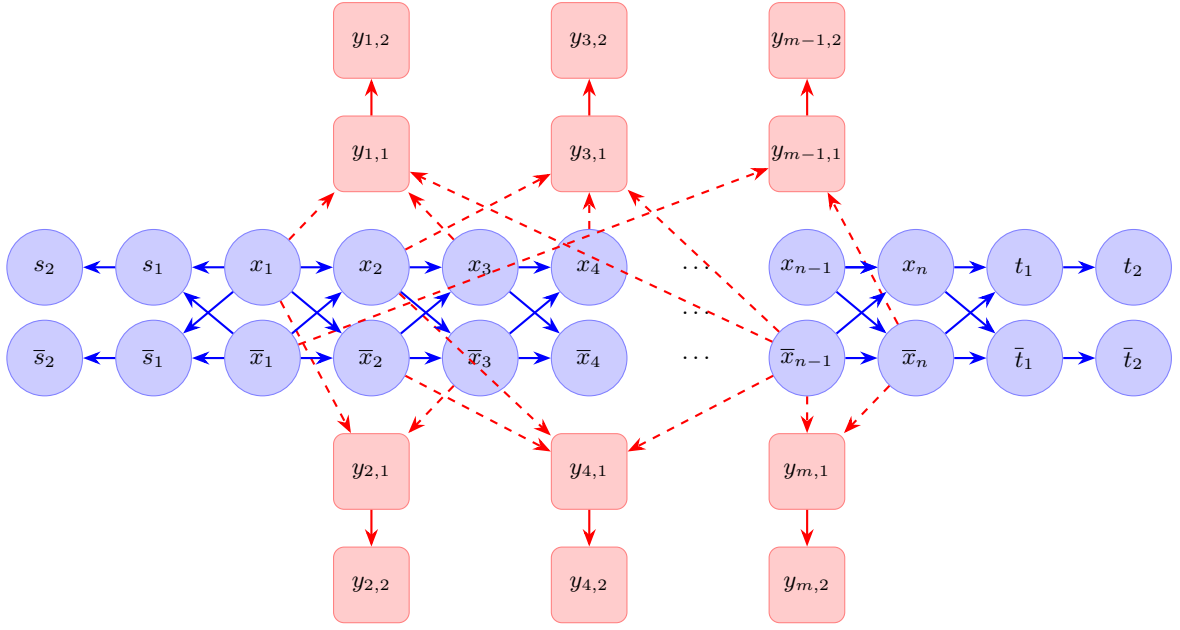
    
    We now show that a satisfying truth assignment for $C$ exists if and only if there exists a shortest path between two vertices in $G$ with a $k$-step reach centrality of $k(m + 4) + n$.

    First, note that by construction, any shortest path that starts and ends at a node in $\cup_{i=1}^{n} \{x_i, \overline{x}_i\}$ can only traverse nodes in $\cup_{i=1}^{n} \{x_i, \overline{x}_i\}$.    

    $\Longrightarrow$ Given a satisfying truth assignment, the corresponding path in $G$ will have a $k$-step reach centrality of at least $k(m + 4) + n$ because we have $km$ nodes from the satisfied clauses, $n$ nodes from the negated through assignments, and $4k$ nodes from the $s$ and $t$ nodes.

    $\Longleftarrow$ Given a solution to kStep-CSP-D with $k$-step reach centrality at least $k(m + 4) + n$, we have the following cases:
    \begin{enumerate}
        \item \label{case1} If the path starts at $x_1$ or $\overline{x}_1$ and ends at $x_{n}$ or $\overline{x}_{n}$, then the corresponding path is a solution to SAT. This is because the $s$ and $t$ nodes contribute $4k$ nodes, the negated variables contribute $n$ nodes, and the clauses contribute an additional $km$ nodes.
        \item \label{case2} If the path starts at a node in $\cup_{i=2}^{n-1} \{x_i, \overline{x}_i\}$ and/or ends at a node in $\cup_{i=2}^{n-2} \{x_i, \overline{x}_i\}$, any arbitrary extension of the path such that it starts at $x_1$ or $\overline{x}_1$ and ends at $x_{n}$ or $\overline{x}_{n}$ will not decrease the $k$-step reach neighbourhood of the path, and we arrive at Case~\ref{case1}.
        \item If the path starts and/or ends at a node in $\cup_{i=1}^{m}\{y_{i, 1}, \dots, y_{i, k}\}$, we can remove the $y$ nodes, maintaining or increasing the path's $k$-step reach neighbourhood, and arrive at Case~\ref{case1} or Case~\ref{case2}.
        \item Similarly, if the path starts and/or ends at a node in $\cup_{i=1}^{k}\{s_i, \overline{s}_i, t_i, \overline{t}_i\}$, we can remove these nodes without decreasing the path's $k$-step reach centrality, and we arrive at Case~\ref{case1} or Case~\ref{case2}.
    \end{enumerate}
    In all the above cases, we can form a path that starts at $x_1$ or $\overline{x}_1$ and ends at $x_{n}$ or $\overline{x}_{n}$, corresponding to a satisfying truth assignment.
\end{proof}

\begin{corollary}
    Problem~\eqref{eq:problem} on a graph with weighted edges is NP-hard.
\end{corollary}

\section{Connection to the closeness-central shortest path}\label{sec:closeness}

We now show how to apply Algorithm~\ref{alg:algorithm} to find the closeness-central shortest path in a graph, which is known to be NP-hard \citep{PhosavanhMatsypura2025}. To do so, we rely on the closed $k$-step neighbourhood, $\mathcal{N}_k[P]$, defined in Section~\ref{sec:problem-definition}. Again, note that for unweighted graphs, we have the equivalent distance-based definition of the closed $k$-step neighbourhood:
\[
\mathcal{N}_k[P] = \{u \in V \colon d(u, P) \leq k\}.
\]
Analogously, we can define the \emph{closed $k$-step reach centrality} of $P$ as $C_k[P] = |\mathcal{N}_k[P]|$. Hence, the problem of finding the $k$-step-central shortest path with this centrality measure is
\begin{equation}\label{eq:problem-inclusive}
    \max_P \left\{C_{k}[P] \colon P \in \mathcal{SP}(G)\right\}.
\end{equation}

Observe that Lemmas~\ref{lem:sp2} and \ref{lem:spk} hold for the \emph{closed} $k$-step reach centrality. To adapt the proof, we remove the $-8$ and $-4$ terms from the corresponding lines in the proof of Lemma~\ref{lem:sp2}. This simplification arises because we no longer need to account for the removal of the path nodes when calculating the neighbourhood. This observation allows us to use Algorithm~\ref{alg:algorithm} to solve Problem~\eqref{eq:problem-inclusive}, and we have the following corollary.

\begin{corollary}\label{cor:solve-problem-inclusive}
    Problem~\eqref{eq:problem-inclusive} on an unweighted graph can be solved in $O(m|V|^{2k+4})$ time, where $m$ is the diameter of the graph.
\end{corollary}

We can now state the problem of finding the closeness-central shortest path. First, we define the closeness centrality of path $P$ as
\[
C_c^{\max}(P) := \max\{d(u, P) \colon u \in V\}.
\]
This metric returns the distance to the node that is farthest from any node in $P$ and is sometimes called \emph{eccentricity} \citep{hedetniemi1981jordan}. Using this metric, we can formally define the problem of finding the closeness-central shortest path as
\begin{align}\label{eq:closest-path-problem}
    \min_P \left\{C_c^{\max}(P) \colon P \in \mathcal{SP}(G)\right\}.
\end{align}
In other words, Problem~\eqref{eq:closest-path-problem} seeks a shortest path $P$ between two nodes such that the maximum distance any user must travel to reach $P$ is minimized. A generalisation of this problem to arbitrary paths ($P$ need not be shortest) is known in the literature as the \emph{central path problem} \citep{puerto2018updatedreview}. However, we prefer the term \emph{closeness-central path problem} and reserve the term \emph{central path problem} for a more general class of problems in which the measure of centrality need not be path eccentricity. Similarly, we use the term \emph{central shortest path} to refer to the problem of finding the central shortest path with respect to some measure of centrality and to its solution. We believe this terminology is more elegant and avoids ambiguity. 

To solve Problem~\eqref{eq:closest-path-problem}, we rely on the following observation. If $C_k[P] = |V|$, then all nodes in $G$ are at most $k$ hops away from $P$. Let $C^\star_k[G]$ denote the optimal value of Problem~\eqref{eq:problem-inclusive} for a fixed $k$. Then, Problem~\eqref{eq:closest-path-problem} can be equivalently written as
\begin{equation}\label{eq:closest-path-problem-reformulated}
    \min \left\{k : C^\star_k[G] = |V|\right\}.
\end{equation}
Hence, Problem~\eqref{eq:closest-path-problem} can be solved by iteratively solving Problem~\eqref{eq:problem-inclusive} for increasing values of $k$ until we find a solution with $C^\star_k[G] = |V|$. We formalise this in Algorithm~\ref{alg:algorithm-closeness}. Note that we only need to consider values of $k$ up to and including the diameter of the graph, as this is the maximum distance between any pair of nodes and therefore an upper bound on closeness centrality. Hence, we have the following result.

\begin{corollary}\label{cor:cls-csp}
    Problem~\eqref{eq:closest-path-problem} on an unweighted graph can be solved in $O(m|V|^{2m+4})$ time, where $m$ is the diameter of the graph.
\end{corollary}
\begin{proof}
    Due to the equivalence of Problem~\eqref{eq:closest-path-problem} and Problem~\eqref{eq:closest-path-problem-reformulated}, we can apply Algorithm~\ref{alg:algorithm} sequentially to solve instances of Problem~\eqref{eq:problem-inclusive} for different values of $k \in \{1,\ldots, m\}$, each with a running time of $O(m|V|^{2k+4})$. This needs to be repeated at most $m$ times, since $m$ is the largest value of $k$ we need to consider.
\end{proof}

\begin{remark}
Corollary~\ref{cor:cls-csp} states that the worst-case running time of Algorithm~\ref{alg:algorithm-closeness} is exponential in the graph diameter. While this poses an obvious scalability challenge, many real-world instances have relatively small diameters, so the problem can be solved in a reasonable time, as we demonstrate in the next section. Moreover, each iteration of the for loop can be executed in parallel.
\end{remark}

\begin{algorithm}
    \caption{Finding the closeness-central shortest path (Cls-CSP)}\label{alg:algorithm-closeness}
    \begin{algorithmic}[1]
        \Input Graph $G = (V, E)$, source node $s \in V$ and destination node $t \in V$.
        \Output Closeness-central shortest path $P^\star$ in $G$.
        \State $m \gets \text{diam}(G)$
        \For{$k = 1, \dots, m$}
            \State Compute $\widetilde{P}$ by solving Problem~\eqref{eq:problem-inclusive} using Algorithm~\ref{alg:algorithm}. 
            \If{$C_k[\widetilde{P}] = |V|$}
                \State $P^\star \gets \widetilde{P}$.
                \State \textbf{break}
            \EndIf
        \EndFor
    \end{algorithmic}
\end{algorithm}

\section{Numerical analysis}\label{sec:numerical}
In this section, we present a comprehensive numerical analysis of the proposed algorithm across a set of real and synthetic graph instances.

\subsection{Data}

We use the following graph instances in our numerical experiments. 
\begin{itemize}
    \item \textit{Watts-Strogatz}: the `small-world' random network generation model proposed by \citet{WattsStrogatz1998}. We generated Watts-Strogatz graphs with 100, 200, 500, 1000 and 2000 nodes. For each graph, the number of initial neighbours was set to 4, and the rewiring probability to 0.1 or 0.2.
    \item \textit{Barab\'{a}si-Albert}: the `scale-free preferential attachment' model proposed by \citet{BarabasiAlbert1999}. This model generates graphs with a power-law degree distribution. We generated graphs with 100, 200, 500, 1000, and 2000 nodes. For each graph, the number of edges added to new nodes was set to 2 to match the density of the Watts-Strogatz graphs.
    \item \textit{Exponential SP}: To analyse the effect of pruning in Algorithm~\ref{alg:algorithm}, we constructed random graphs with a controllable feasible-set size. We describe our graph-generating procedure in Section~\ref{sec:exp-sp-graphs}.
    \item \textit{Urban Streets}: the set of twenty spatial networks, first studied by \citet{CrucittiPaolo2006}. These graphs represent 1-square-mile maps of twenty cities.
\end{itemize}
For the three synthetic instance classes, we generated 30 random instances per parameter set.

\subsection{Generating graphs with exponentially-many shortest paths}\label{sec:exp-sp-graphs}

Here, we describe a procedure for generating graphs with exponentially many shortest paths, building on the observation by \citet{Matsypura2023} that the number of shortest paths between a pair of nodes can be exponentially large. We provide a general procedure that allows control over the number of shortest paths while introducing random structures into the graph.

We first define two parameters, $R$ and $\ell$, which indicate the number of repetitions and the number of linking nodes, respectively, used to construct the backbone of our graphs. This backbone structure enables us to generate a graph with exponentially many shortest paths. For each repetition $r \in \{0, 1, \dots, R-1\}$, we generate a loop composed of two paths of length $\ell + 1$:
\begin{itemize}
	\item $P_1 = \langle r(2 \ell + 1), r(2 \ell + 1) + 1, \dots, r(2 \ell + 1) + \ell, (r+1)(2 \ell + 1) \rangle$,
	\item $P_2 = \langle r(2 \ell + 1), r(2 \ell + 1) + \ell + 1, \dots,  r(2 \ell + 1) + 2\ell, (r+1)(2 \ell + 1)\rangle$.
\end{itemize}
This procedure produces a chain-like graph with $R(2\ell + 1)+ 1$ nodes and $2R(\ell + 1)$ edges. See Figure~\ref{fig:chain-graphs} for an example with $R = 3$ and $\ell = 2$.

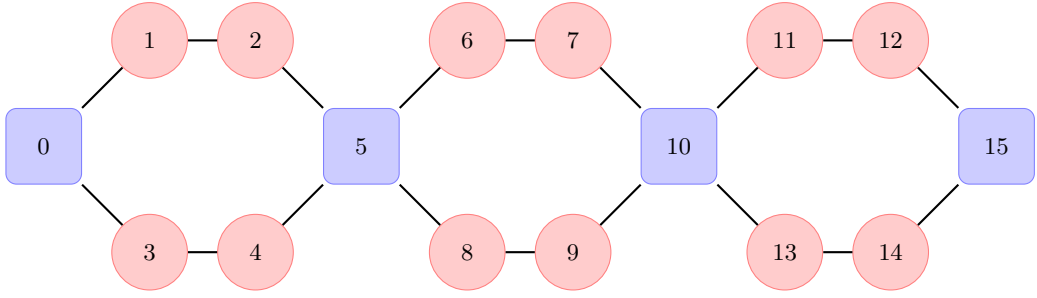
\begin{figure}[!htbt]
    \centering
    \begin{tikzpicture}[scale=1.4, square/.style={rectangle, rounded corners}]
    \footnotesize
    \begin{scope}[circle,minimum size=10mm]
        \node at (1, 1) [draw=red!50,fill=red!20]  (1) {1};
        \node at (2, 1) [draw=red!50,fill=red!20]  (2) {2};
        \node at (1, -1) [draw=red!50,fill=red!20] (3) {3};
        \node at (2, -1) [draw=red!50,fill=red!20] (4) {4};
        \node at (4, 1) [draw=red!50,fill=red!20]  (6) {6};
        \node at (5, 1) [draw=red!50,fill=red!20]  (7) {7};
        \node at (4, -1) [draw=red!50,fill=red!20] (8) {8};
        \node at (5, -1) [draw=red!50,fill=red!20] (9) {9};
        \node at (7, 1) [draw=red!50,fill=red!20]  (11) {11};
        \node at (8, 1) [draw=red!50,fill=red!20]  (12) {12};
        \node at (7, -1) [draw=red!50,fill=red!20] (13) {13};
        \node at (8, -1) [draw=red!50,fill=red!20] (14) {14};
    \end{scope}
    \begin{scope}[circle, minimum size=10mm]
        \node at (0, 0) [draw=blue!50,fill=blue!20, square, inner sep=-0.2em]  (0) {0};
        \node at (3, 0) [draw=blue!50,fill=blue!20, square, inner sep=-0.2em]  (5) {5};
        \node at (6, 0) [draw=blue!50,fill=blue!20, square, inner sep=-0.2em]  (10) {10};
        \node at (9, 0) [draw=blue!50,fill=blue!20, square, inner sep=-0.2em]  (15) {15};
    \end{scope}
    \path 
    (0) edge [black, thick] (1)
    (1) edge [black, thick] (2)
    (2) edge [black, thick] (5)
    (0) edge [black, thick] (3)
    (3) edge [black, thick] (4)
    (4) edge [black, thick] (5)
    (5) edge [black, thick] (6)
    (6) edge [black, thick] (7)
    (7) edge [black, thick] (10)
    (5) edge [black, thick] (8)
    (8) edge [black, thick] (9)
    (9) edge [black, thick] (10)
    (10) edge [black, thick] (11)
    (11) edge [black, thick] (12)
    (12) edge [black, thick] (15)
    (10) edge [black, thick] (13)
    (13) edge [black, thick] (14)
    (14) edge [black, thick] (15);
\end{tikzpicture}
    \caption{Example of a chain graph with $R=3$ and $\ell=2$.}
    \label{fig:chain-graphs}
\end{figure}

Note that for a pair of nodes in different loops, the number of shortest paths is $2^{q-1}$, where $q$ is the number of blue square nodes between them. It is easy to see that the graph produced by this procedure has a lot of structure, and for some problems, including ours, this structure is uninteresting. To make it more interesting, we introduce chaotic structures by adding trees to the base chain. To do this, we generate random trees and randomly select non-linking nodes as their roots. In Figure~\ref{fig:chain-graphs}, the non-linking nodes are red circular nodes. We chose only non-linking nodes so that the resulting graphs are relevant to the central shortest path problem; otherwise, there could be exponentially many solutions.

For the purposes of our paper, we will analyse graphs with approximately 500 nodes. To generate these instances, we first construct a base chain structure comprising about 250 nodes. To choose appropriate values of $\ell$ and $R$, we fix $\ell$ and compute $R = \frac{250}{2\ell+1}$, then round to obtain an integer. To obtain the additional 250 nodes, we generate 25 trees, each with 11 nodes, since each tree contributes 10 additional nodes to the graph when its root is connected to the base structure.

\subsection{Experimental results}

We now present the results of our numerical experiments. Table~\ref{tab:ba} reports results for the Barab\'{a}si-Albert graphs, while Table~\ref{tab:ws} reports results for the Watts-Strogatz instances. Table~\ref{tab:urban-streets} reports results for selected Urban Streets instances. Results for the remaining Urban Streets instances are provided in Appendix~\ref{sec:appendix-tables}. For Tables~\ref{tab:ba} and \ref{tab:ws}, results are averaged over 30 instances. Across all these tables, Algorithm~\ref{alg:algorithm} was applied with pruning. The experiments were conducted on a Mac Studio (2023) with an Apple M2 Ultra and 128 GB of RAM, running macOS Sequoia 15.7. Algorithm~\ref{alg:algorithm} was implemented in Python 3.12. 

In these tables, we report the graph size, the graph diameter, and the relative size of the feasible set $|\mathcal{SP}(G)|/U(G)$, where $U(G)$ is the number of shortest paths in $G$ if each pair of nodes had a single shortest path between them, i.e., the number of pairwise connectable nodes in $G$. The row \textit{diam centrality} gives the maximum centrality of a diametral path, while the row \textit{path length} gives the length of the shortest path with the highest centrality, which is recorded in the row \textit{path centrality}. Finally, the run time of Algorithm~\ref{alg:algorithm} is reported in the row \textit{time}, given in seconds. For the synthetic instances, we test our algorithm for $k=1$ and $k=2$, whereas for the Urban Streets instances, we consider $k$ up to $ 9$ but report results only up to $k=5$ for brevity.

In all the result tables, we see that increasing $k$ quickly increases the neighbourhood of the path. For most of our instances, as $k$ grows, the centrality of the optimal path increases, while the length of the $k$-step reach-central shortest path decreases. In most cases, when $k$ increases from 1 to 2, the neighbourhood of the optimal path doubles or even triples. We also observe that the length of the optimal path is shorter than the diameter of the graphs, implying that, in most graphs, the most influential path is not necessarily the longest.

Tables~\ref{tab:ba} and \ref{tab:ws} demonstrate the scalability of our algorithm on synthetic graphs that mimic social networks. Our algorithm can solve Problem~\ref{eq:problem} on graphs with 2000 nodes in about half an hour. We also find that our algorithm runs faster on networks with smaller diameters, consistent with Theorem~\ref{thm:problem-running-time}. We observe a 10-20\% increase in running time when $k$ increases from 1 to 2, as the algorithm must consider extending more paths.

Our algorithm solved the largest Urban Streets instance (Ahmedabad, 2870 nodes) in under six hours when $k=1$. As expected, the runtime of our algorithm increases with $k$. On the same Ahmedabad instance, when $k=5$, the runtime is just over 17.5 hours. For the Venice graph, our algorithm solves the $k=1$ case in around 100 minutes, but when $k=5$, the running time approximately triples.

\begin{table}[!htb]
    \centering
    \small
    \begin{tabular}{lrrrrr}
\toprule
\multicolumn{1}{c}{} & \multicolumn{5}{c}{Barab\'{a}si-Albert}              \\ \cmidrule(lr){2-6}
$|V|$                      & 100   & 200    & 500    & 1000   & 2000    \\
$|E|$                      & 196   & 396    & 996    & 1996   & 3996    \\
$|\mathcal{SP}(G)|/U(G)$ & 2.13  & 2.39   & 2.67   & 2.87   & 3.06    \\
diam                       & 5.57  & 6.13   & 7.03   & 7.27   & 8.00    \\ \midrule
\multicolumn{1}{c}{}       & \multicolumn{5}{c}{$k=1$}                  \\ \cmidrule(lr){2-6}
diam centrality            & 46.17 & 74.87  & 120.10 & 179.30 & 244.73  \\
path length                & 3.87  & 4.40   & 4.67   & 5.07   & 5.47    \\
path centrality            & 51.87 & 79.97  & 138.33 & 199.83 & 291.47  \\
time                       & 0.12  & 0.86   & 13.20  & 110.19 & 926.01  \\
\midrule
\multicolumn{1}{c}{}       & \multicolumn{5}{c}{$k=2$}                  \\ \cmidrule(lr){2-6}
diam centrality            & 90.10 & 176.63 & 383.77 & 679.63 & 1127.10 \\
path length                & 2.83  & 3.97   & 4.67   & 5.23   & 5.63    \\
path centrality            & 93.93 & 181.87 & 409.20 & 721.53 & 1245.43 \\
time                       & 0.18  & 1.27   & 18.52  & 146.71 & 1197.97 \\
\bottomrule
\end{tabular}
    \caption{Results for Barab\'{a}si-Albert graphs, averaged over 30 instances.}
    \label{tab:ba}
\end{table}

\begin{table}[!htb]
    \centering
    \small
\begin{tabular}{lrrrrrrrrrr}
\toprule
\multicolumn{1}{c}{}       & \multicolumn{5}{c}{Watts-Strogatz (4, 0.1)} & \multicolumn{5}{c}{Watts-Strogatz (4, 0.2)} \\ \cmidrule(lr){2-6}\cmidrule(lr){7-11}
$|V|$                      & 100   & 200   & 500   & 1000   & 2000    & 100   & 200   & 500   & 1000   & 2000          \\
$|E|$                      & 200   & 400   & 1000  & 2000   & 4000    & 200   & 400   & 1000  & 2000   & 4000          \\
$|\mathcal{SP}(G)|/U(G)$   & 2.46  & 2.70  & 3.19  & 3.53   & 3.89    & 2.04  & 2.10  & 2.25  & 2.35   & 2.45          \\
diam                       & 10.43 & 12.47 & 15.40 & 17.73  & 19.80   & 8.43  & 9.53  & 11.50 & 12.97  & 14.23         \\ \midrule
\multicolumn{1}{c}{}       & \multicolumn{10}{c}{$k=1$}                                                                \\ \cmidrule(lr){2-11}
diam centrality            & 22.23 & 26.30 & 31.20 & 35.37  & 38.50   & 21.67 & 24.97 & 28.80 & 30.60  & 34.17         \\
path length                & 9.03  & 10.87 & 12.97 & 14.73  & 16.23   & 7.03  & 8.00  & 9.33  & 10.17  & 11.03         \\
path centrality            & 23.17 & 27.77 & 33.87 & 37.60  & 41.77   & 23.33 & 26.27 & 31.73 & 35.17  & 39.00         \\
time                       & 0.22  & 1.61  & 26.61 & 230.13 & 2015.32 & 0.17  & 1.27  & 20.98 & 181.66 & 1590.98       \\ \midrule
\multicolumn{1}{c}{}       & \multicolumn{10}{c}{$k=2$}                                                                \\ \cmidrule(lr){2-11} 
diam centrality            & 47.80 & 60.57 & 74.00 & 82.20  & 91.77   & 47.80 & 60.57 & 74.00 & 82.20  & 91.77         \\
path length                & 8.30  & 10.13 & 11.83 & 12.67  & 14.37   & 6.60  & 7.37  & 8.70  & 9.70   & 10.27         \\
path centrality            & 52.03 & 67.30 & 84.73 & 96.93  & 109.90  & 58.50 & 74.80 & 97.00 & 109.63 & 125.00        \\
time                       & 0.28  & 2.06  & 32.14 & 260.96 & 2180.42 & 0.21  & 1.49  & 24.06 & 198.97 & 1685.64       \\
\bottomrule
\end{tabular}
    \caption{Results for Watts-Strogatz graphs with a rewiring probability of 0.1 and 0.2, averaged over 30 instances.}
    \label{tab:ws}
\end{table}

\begin{table}[!htb]
    \centering
    \small
\begin{tabular}{lrrrrr}
\toprule
  ~                        & Ahmedabad  & Cairo    & London   & San Francisco & Venice   \\ \cmidrule(lr){2-6}
  $|V|$                    & 2870       & 1496     & 488      & 169           & 1840     \\
  $|E|$                    & 4375       & 2252     & 729      & 271           & 2397     \\
  diam                     & 86         & 51       & 33       & 21            & 102      \\
  $|\mathcal{SP}(G)|/U(G)$ & 63.39      & 12.26    & 14.60    & 105.80        & 36.66    \\ \midrule   \multicolumn{1}{c}{} & \multicolumn{5}{c}{$k=1$} \\ \cmidrule(lr){2-6}
  diam centrality          & 108        & 76       & 48       & 39            & 116      \\
  path length              & 84         & 49       & 29       & 19            & 94       \\
  path centrality          & 110        & 78       & 51       & 41            & 120      \\
  time                     & 20568.23   & 2141.02  & 53.81    & 2.78          & 6096.04  \\ \midrule   \multicolumn{1}{c}{} & \multicolumn{5}{c}{$k=2$} \\ \cmidrule(lr){2-6}
  diam centrality          & 244        & 174      & 101      & 76            & 227      \\
  path length              & 82         & 48       & 30       & 18            & 93       \\
  path centrality          & 247        & 177      & 104      & 79            & 237      \\
  time                     & 23216.90   & 2473.90  & 78.90    & 9.17          & 7080.23  \\ \midrule   \multicolumn{1}{c}{} & \multicolumn{5}{c}{$k=3$} \\ \cmidrule(lr){2-6}
  diam centrality          & 387        & 283      & 153      & 108           & 320      \\
  path length              & 81         & 47       & 30       & 17            & 90       \\
  path centrality          & 392        & 286      & 161      & 113           & 328      \\
  time                     & 29380.25   & 3235.07  & 134.71   & 28.36         & 9096.59  \\ \midrule   \multicolumn{1}{c}{} & \multicolumn{5}{c}{$k=4$} \\ \cmidrule(lr){2-6}
  diam centrality          & 537        & 394      & 194      & 125           & 408      \\
  path length              & 78         & 45       & 29       & 17            & 65       \\
  path centrality          & 544        & 398      & 211      & 134           & 432      \\
  time                     & 40717.58   & 4563.87  & 241.38   & 71.16         & 12579.93 \\ \midrule   \multicolumn{1}{c}{} & \multicolumn{5}{c}{$k=5$} \\ \cmidrule(lr){2-6}
  diam centrality          & 677        & 515      & 237      & 135           & 492      \\
  path length              & 76         & 45       & 29       & 15            & 63       \\
  path centrality          & 686        & 527      & 261      & 144           & 544      \\
  time                     & 62415.42   & 6691.38  & 397.10   & 132.54        & 18580.46 \\ 
\bottomrule
\end{tabular}
    \caption{Results for selected Urban Streets graphs.}
    \label{tab:urban-streets}
\end{table}

\subsection{The effect of pruning}
We now examine the effect of pruning. We apply Algorithm~\ref{alg:algorithm} with $k=1$ and $k=2$ to graphs generated by the procedure in Section~\ref{sec:exp-sp-graphs}. We report the runtime of Algorithm~\ref{alg:algorithm} as a function of the relative size of the feasible set, averaged over 30 randomly generated instances. In these figures, low values of $|\mathcal{SP}(G)|/U(G)$ correspond to graphs with few but long chains, while high values correspond to graphs with many short chains. Again, note that in all instances, the graphs analysed in this subsection have approximately 500 nodes.

\begin{figure}[!htb]
    \centering
    \begin{tikzpicture}[font=\small,]
    \begin{groupplot}[
        group style={
            group size=2 by 1,
            horizontal sep=1em,
            yticklabels at=edge left,
            ylabels at=edge left,
            xlabels at=edge bottom,
        },
        x grid style={darkgray},
        xlabel={$|\mathcal{SP}(G)|/U(G)$},
        xmin=1, xmax=500,
        xmode=log,
        xtick style={color=black},
        y grid style={darkgray},
        ylabel={Time (s)},
        ymin=0, ymax=3700,
        ytick style={color=black},
        ytick={0,500,...,3500},
        legend columns=-1,
        legend style={/tikz/every even column/.append style={column sep=0.5cm}}
    ]
    \nextgroupplot[title={$k=1$},legend to name=legendpos]
        \addplot [draw=red, fill=red, mark=o, only marks]
        table{%
        x  y
        1.2703821711804766 151.66177195130004
        1.530260517634255 169.52882159726724
        1.87813047833709 168.7187363792332
        2.335262396064603 193.78740276526733
        2.9483489530500084 236.41440450970063
        3.667302389705882 232.08257332630066
        4.8762799999999995 214.7513326345339
        6.284536457345222 222.22258788460022
        8.27823103336984 224.4082213097669
        10.790730124800275 212.44874889026735
        18.18775487797227 239.3727983306007
        33.591162661814835 241.84058506800017
        57.66273427041499 227.50142862636665
        203.0379935285314 253.3787610944329
        966.2781939969207 258.6172575279013
        };
    \addlegendentry{With pruning}
    \addplot [draw=blue, fill=blue, mark=*, only marks]
        table{%
        x  y
        1.2703821711804766 108.43496827499912
        1.530260517634255 118.93597676666735
        1.87813047833709 119.55591048476781
        2.335262396064603 139.4818928376015
        2.9483489530500084 177.10518133620013
        3.667302389705882 188.9952760985975
        4.8762799999999995 203.4066966542324
        6.284536457345222 254.65647025546642
        8.27823103336984 342.3658105471996
        10.790730124800275 450.24343375700215
        18.18775487797227 1286.3773675595332
        33.591162661814835 3597.9342334611674
        };
        \addlegendentry{Without pruning}                \coordinate (c1) at (rel axis cs:0,1);
    \nextgroupplot[title={$k=2$}]
    \addplot [draw=red, fill=red, mark=o, only marks]
        table{%
        x  y
        1.27038217118048 157.55795744102
        1.53026051763425 182.190307193034
        1.87813047833709 184.8663469306
        2.3352623960646 216.220941404198
        2.94834895305001 269.518811191631
        3.66730238970588 269.8604486792
        4.87628 254.032340245833
        6.28453645734522 267.5935276806
        8.27823103336984 275.884363806899
        10.7907301248003 264.379686555601
        18.1877548779723 309.602144544468
        33.5911626618148 320.738562211066
        57.662734270415 308.626150184834
        203.037993528531 362.863380805465
        966.278193996921 392.157437240261
        };
    \addplot [draw=blue, fill=blue, mark=*, only marks]
        table{%
        x  y
        1.27038217118048 88.589809256924
        1.53026051763425 127.520664892998
        1.87813047833709 129.156905181931
        2.3352623960646 157.225418433252
        2.94834895305001 202.466580454214
        3.66730238970588 223.670797848593
        4.87628 258.918182909861
        6.28453645734522 352.10070725133
        8.27823103336984 508.983297738743
        10.7907301248003 719.157225355529
        18.1877548779723 2265.2935575584
        };
        \coordinate (c2) at (rel axis cs:1,1);
    \end{groupplot}
    \coordinate (c3) at ($(c1)!.5!(c2)$);
    \node[below] at (c3 |- current bounding box.south)
      {\pgfplotslegendfromname{legendpos}};
\end{tikzpicture}
    \caption{Running time of Algorithm~\ref{alg:algorithm} with $k=1$ and $k=2$ over different sets of graphs generated with exponentially many shortest paths, averaged over 30 instances.}
    \label{fig:pruning-comparison}
\end{figure}
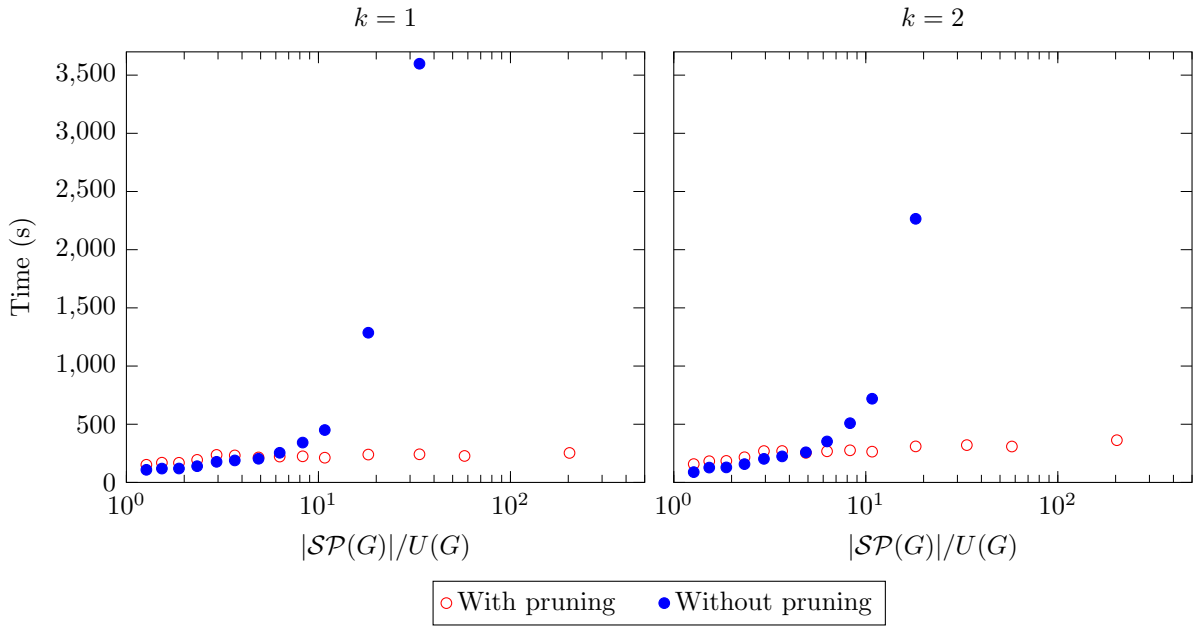

Figure~\ref{fig:pruning-comparison} clearly shows the importance of pruning in keeping the algorithm's running time polynomially bounded. As the number of unique shortest paths increases with the number of loops in the graph’s backbone, the running time without pruning grows exponentially, as expected. In contrast, with pruning enabled, the running time plateaus even as the size of the feasible set grows. However, when the relative size of the feasible set is small, e.g., when $|\mathcal{SP}(G)|/U(G)$ is less than 5, the cost of pruning exceeds the runtime when all results are kept in memory. Although the above results focus on time complexity, the space complexity also grows in a similar manner. With pruning enabled, the required memory remains stable even as the relative size of the feasible set increases. By contrast, when the algorithm runs without pruning, the required memory grows exponentially, and our machine cannot handle instances where the ratio $|\mathcal{SP}(G)|/U(G)$ exceeds 50.

This finding is particularly important for lattice-structured networks, such as street networks in grid-based cities. To see this, consider an $n \times n$ square grid graph. There are $(2n)!/(n!)^2$ shortest paths between the two corner nodes. This number grows exponentially with $n$. We also observe this empirically in the reported relative sizes of the Urban Streets instances in Table~\ref{tab:urban-streets}. For example, in the San Francisco instance, each pair of connectable nodes has, on average, 105.80 distinct shortest paths.

\subsection{Analysis on Urban Streets}

We now turn our attention to the Urban Streets instances. These spatial graphs allow us to visualise the impact of varying $k$ on graphs with physical representations. We plot the results for London and San Francisco in Figures~\ref{fig:london} and \ref{fig:san-francisco}, respectively. We chose these cities because they represent different topologies: London is a naturally sprawling city, whereas San Francisco is a planned city with grid-like structures. In these figures, we retain the cardinal directions, but the spacing has been adjusted for presentation. Note that $|P|$ and the length reported in Table~\ref{tab:urban-streets} differ, because the length of the path is one less than the number of nodes traversed. In addition to these figures, we summarise how the optimal path and neighbourhood size change as $k$ increases in Figure~\ref{fig:streets-analysis} for the two cities outlined above, as well as the three largest instances. In our experiments, we computed the $k$-step reach-central shortest paths for $k=1$ to $k=9$ across all our instances, but for graphs with over 1000 nodes, we considered only $k=1$ to $k=7$ due to the computational cost. Results for additional cities can be found in Figure~\ref{fig:streets-analysis-appendix} in the appendix.
\begin{figure}[!htb]
    \centering
    \captionsetup{justification=centering}
    \begin{subfigure}[c]{0.32\textwidth}
        \resizebox{\textwidth}{!}{\includegraphics{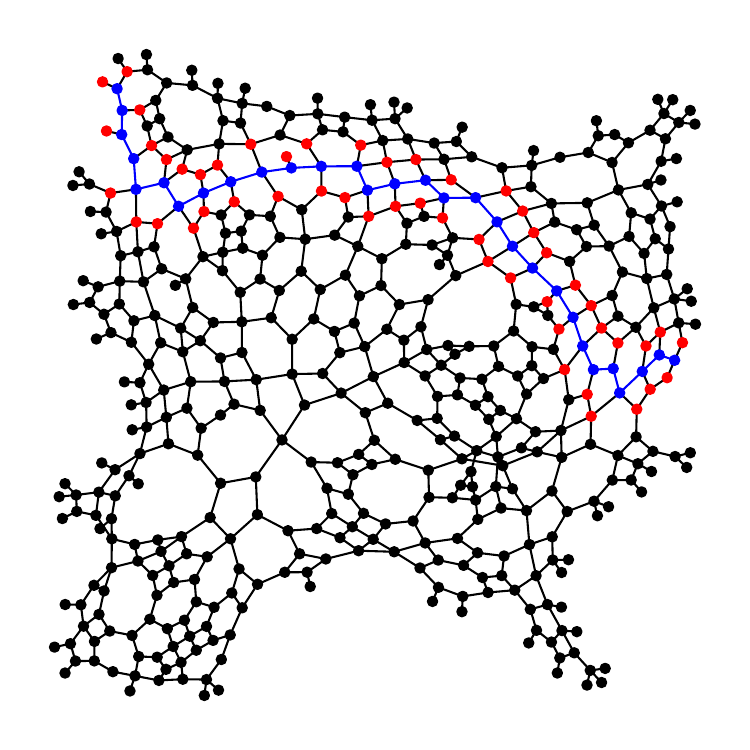}}
        \caption{$k=1$\\$C_1(P^\star) = 51$, $|P^\star| = 30$}
        \label{fig:london-1}
    \end{subfigure}%
    ~ 
    \begin{subfigure}[c]{0.32\textwidth}
        \resizebox{\textwidth}{!}{\includegraphics{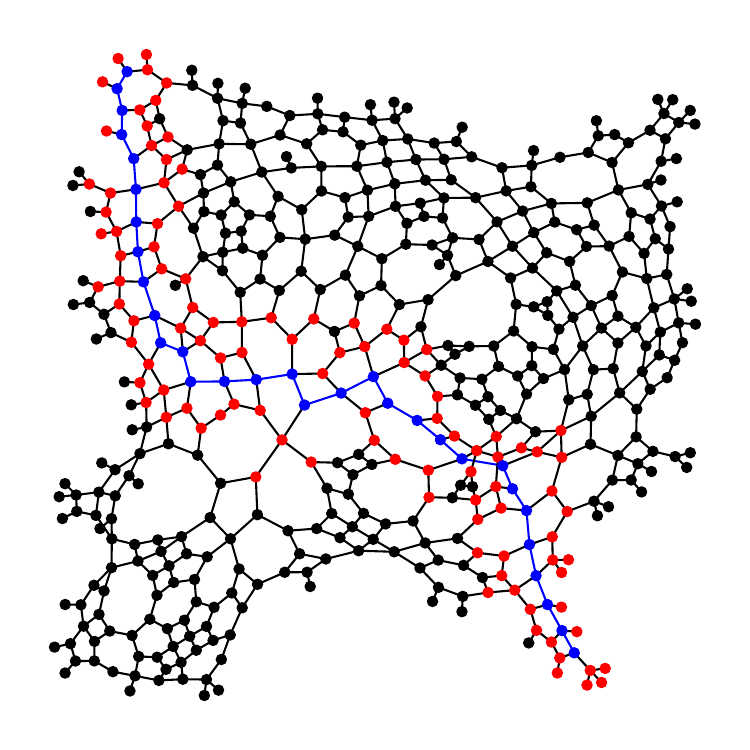}}
        \caption{$k=2$\\$C_2(P^\star) = 104$, $|P^\star| = 31$}
        \label{fig:london-2}
    \end{subfigure}%
    ~ 
    \begin{subfigure}[c]{0.32\textwidth}
        \resizebox{\textwidth}{!}{\includegraphics{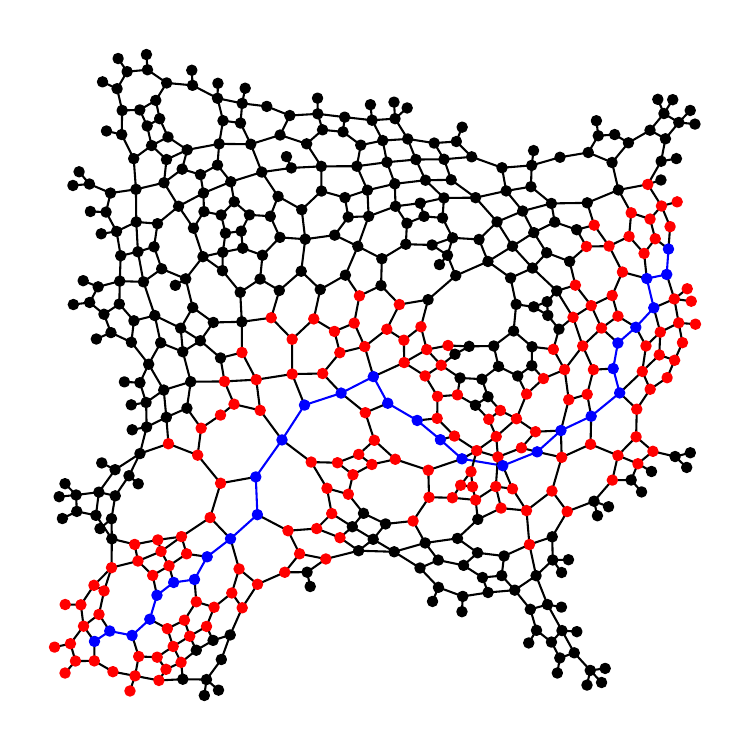}}
        \caption{$k=3$\\$C_3(P^\star) = 161$, $|P^\star| = 31$}
        \label{fig:london-3}
    \end{subfigure}%
    \vspace{1em}
    \begin{subfigure}[c]{0.32\textwidth}
        \resizebox{\textwidth}{!}{\includegraphics{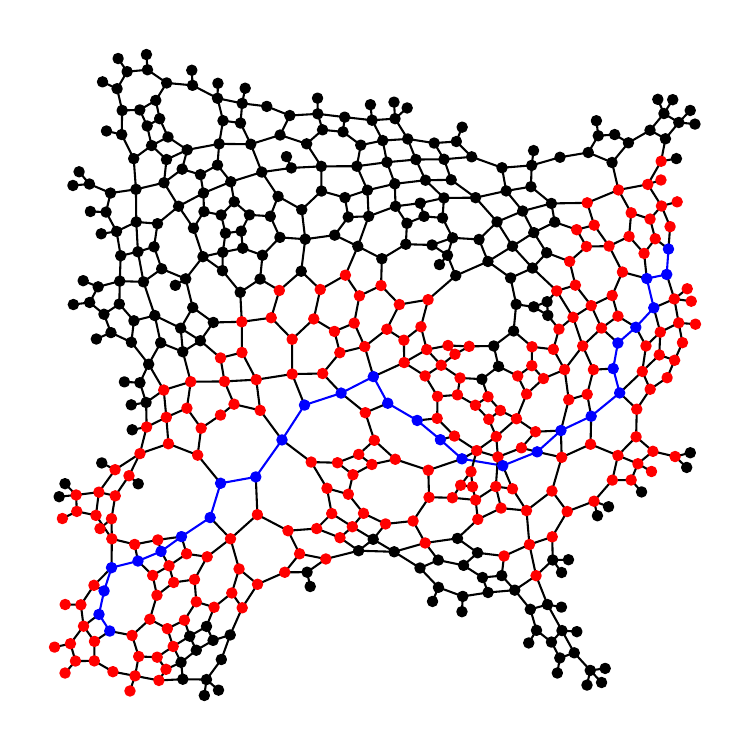}}
        \caption{$k=4$\\$C_4(P^\star) = 211$, $|P^\star| = 30$}
        \label{fig:london-4}
    \end{subfigure}%
    ~ 
    \begin{subfigure}[c]{0.32\textwidth}
        \resizebox{\textwidth}{!}{\includegraphics{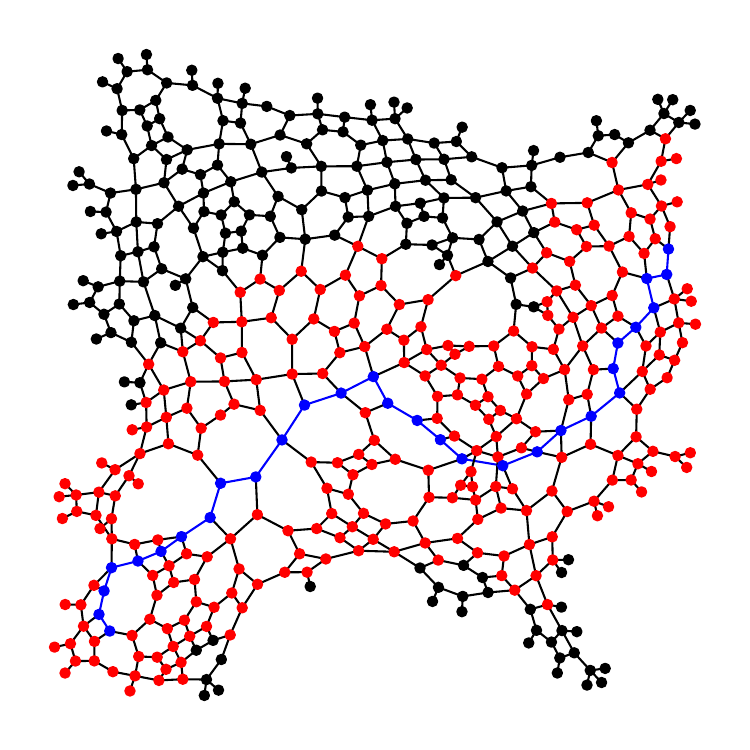}}
        \caption{$k=5$\\$C_5(P^\star) = 261$, $|P^\star| = 30$}
        \label{fig:london-5}
    \end{subfigure}%
    ~ 
    \begin{subfigure}[c]{0.32\textwidth}
        \resizebox{\textwidth}{!}{\includegraphics{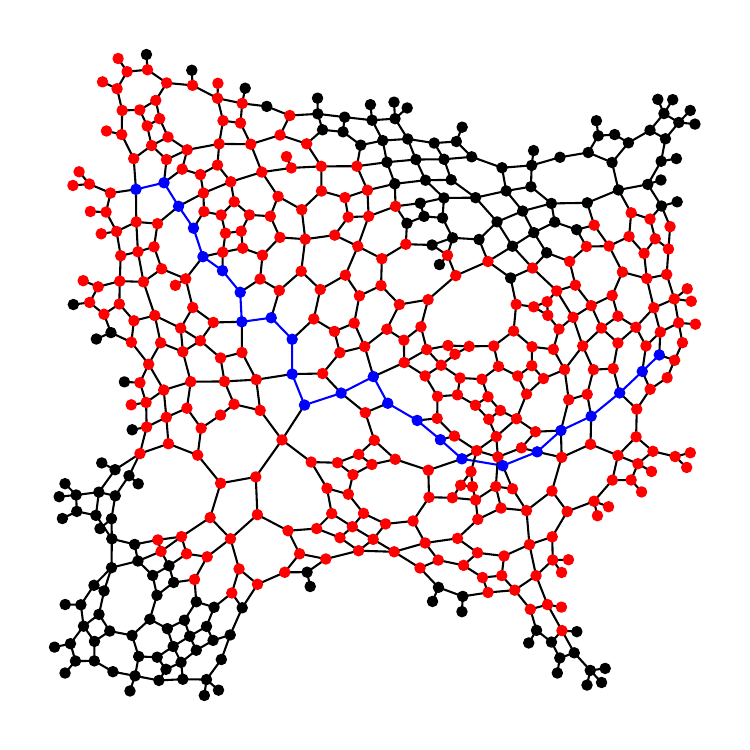}}
        \caption{$k=6$\\$C_6(P^\star) = 306$, $|P^\star| = 25$}
        \label{fig:london-6}
    \end{subfigure}%
    \vspace{1em}
    \begin{subfigure}[c]{0.32\textwidth}
        \resizebox{\textwidth}{!}{\includegraphics{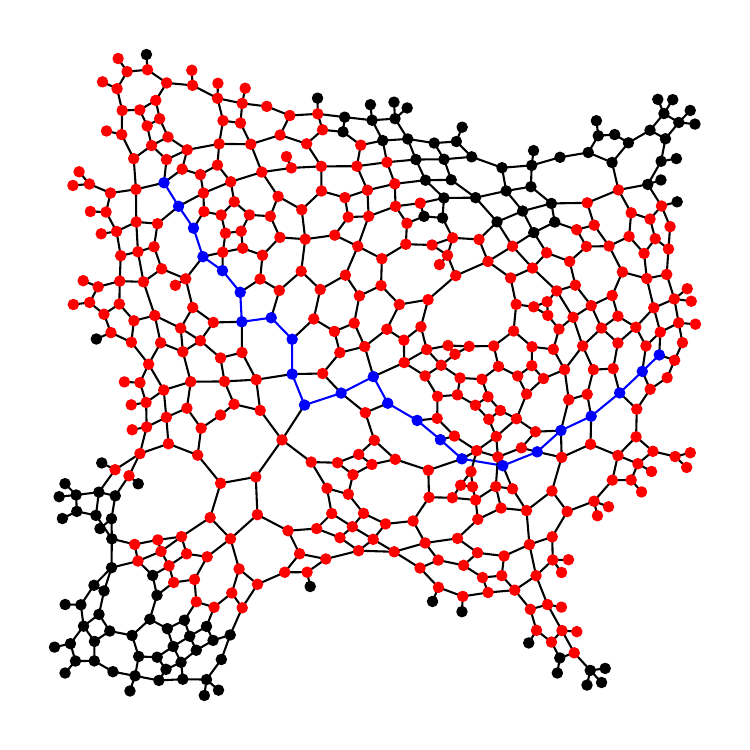}}
        \caption{$k=7$\\$C_7(P^\star) = 348$, $|P^\star| = 24$}
        \label{fig:london-7}
    \end{subfigure}%
    ~ 
    \begin{subfigure}[c]{0.32\textwidth}
        \resizebox{\textwidth}{!}{\includegraphics{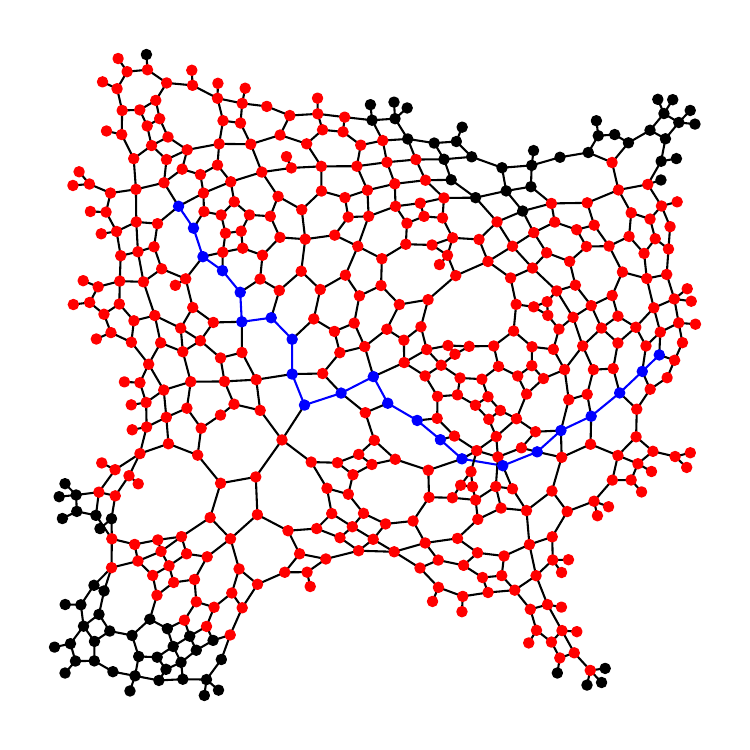}}
        \caption{$k=8$\\$C_8(P^\star) = 383$, $|P^\star| = 23$}
        \label{fig:london-8}
    \end{subfigure}%
    ~ 
    \begin{subfigure}[c]{0.32\textwidth}
        \resizebox{\textwidth}{!}{\includegraphics{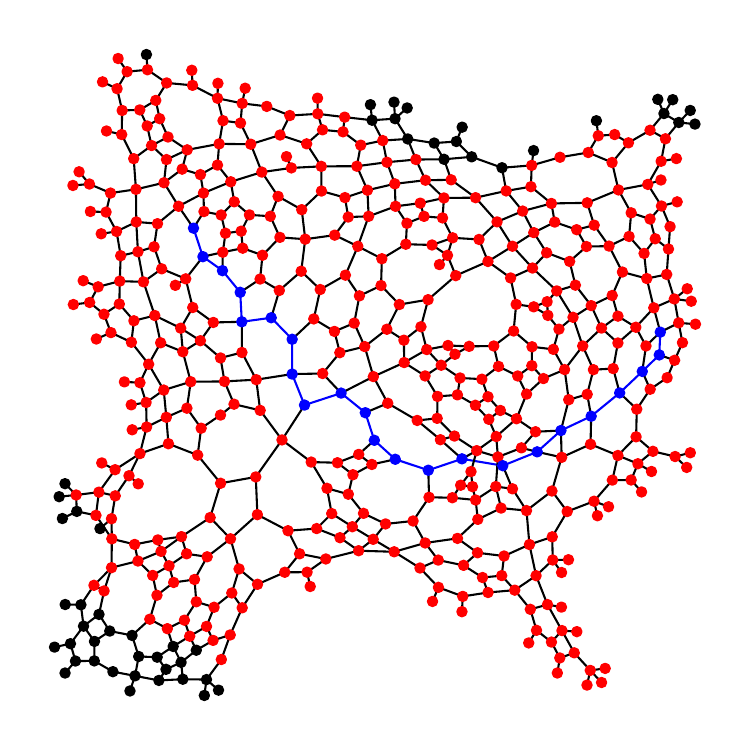}}
        \caption{$k=9$\\$C_9(P^\star) = 413$, $|P^\star| = 23$}
        \label{fig:london-9}
    \end{subfigure}
    \caption{$k$-step-central shortest paths for London (diameter: 33).}
    \label{fig:london}
\end{figure}

\begin{figure}[!htb]
    \centering
    \captionsetup{justification=centering}
    \begin{subfigure}[c]{0.32\textwidth}
        \resizebox{\textwidth}{!}{\includegraphics{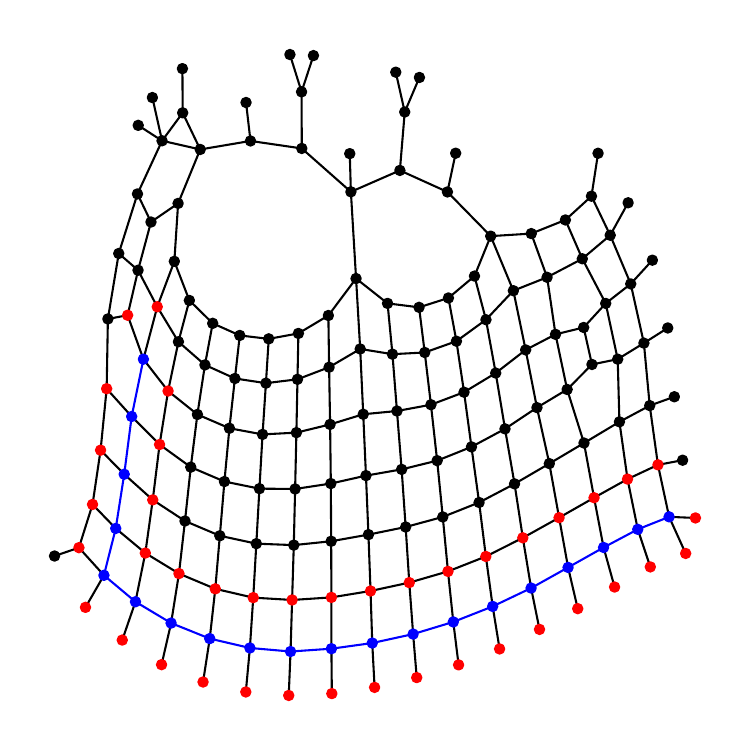}}
        \caption{$k=1$\\$C_1(P^\star) = 41$, $|P^\star| = 20$}
        \label{fig:san-francisco-1}
    \end{subfigure}%
    ~ 
    \begin{subfigure}[c]{0.32\textwidth}
        \resizebox{\textwidth}{!}{\includegraphics{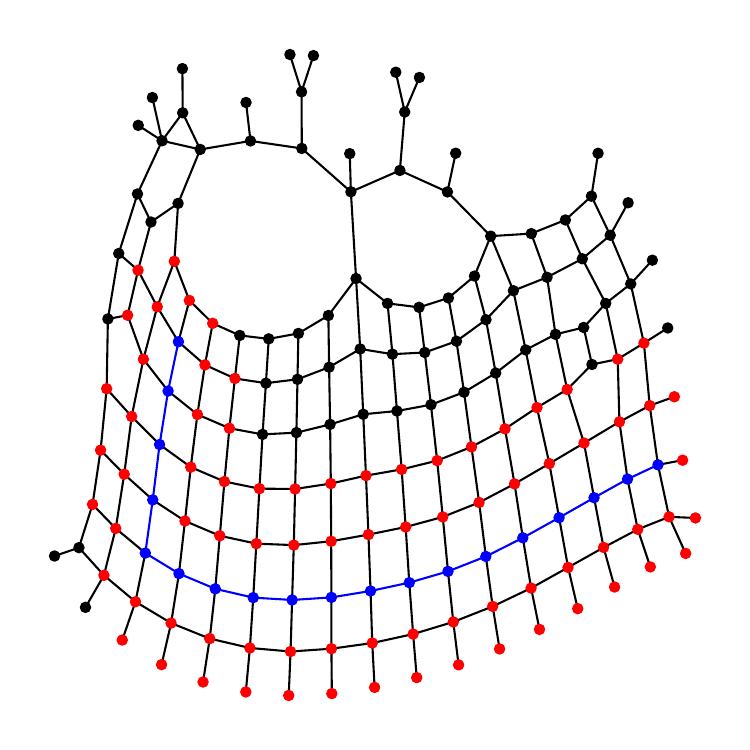}}
        \caption{$k=2$\\$C_2(P^\star) = 79$, $|P^\star| = 19$}
        \label{fig:san-francisco-2}
    \end{subfigure}%
    ~ 
    \begin{subfigure}[c]{0.32\textwidth}
        \resizebox{\textwidth}{!}{\includegraphics{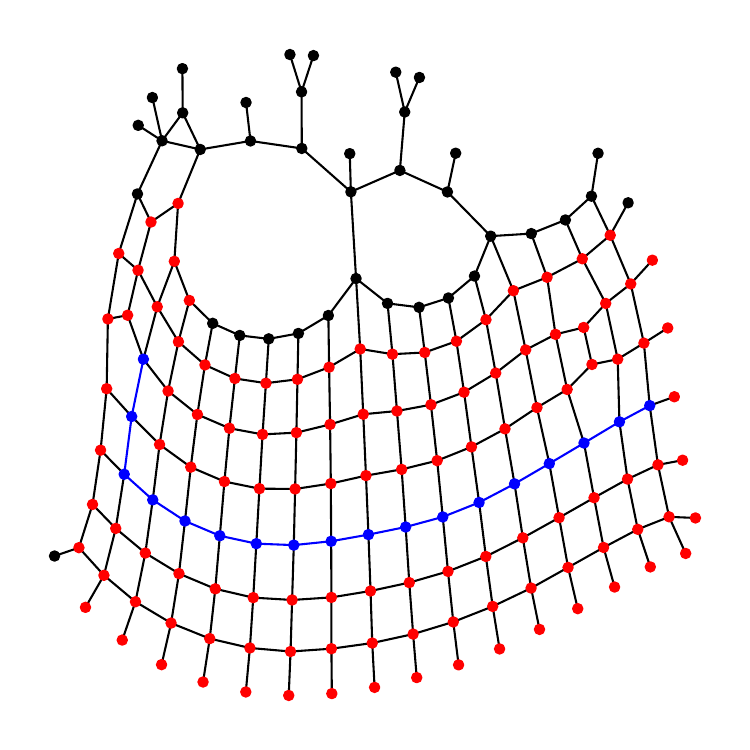}}
        \caption{$k=3$\\$C_3(P^\star) = 113$, $|P^\star| = 18$}
        \label{fig:san-francisco-3}
    \end{subfigure}%
    \vspace{1em}
    \begin{subfigure}[c]{0.32\textwidth}
        \resizebox{\textwidth}{!}{\includegraphics{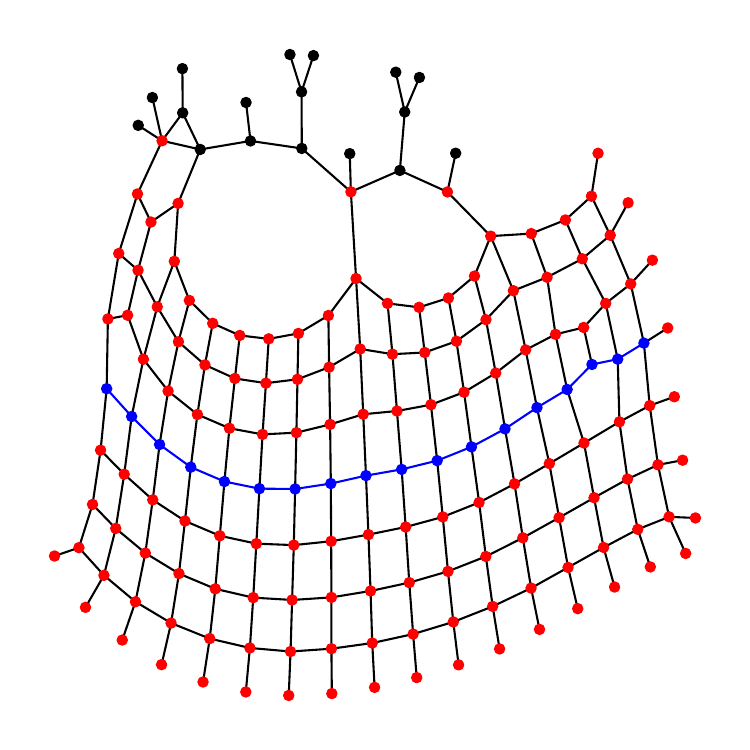}}
        \caption{$k=4$\\$C_4(P^\star) = 134$, $|P^\star| = 18$}
        \label{fig:san-francisco-4}
    \end{subfigure}%
    ~ 
    \begin{subfigure}[c]{0.32\textwidth}
        \resizebox{\textwidth}{!}{\includegraphics{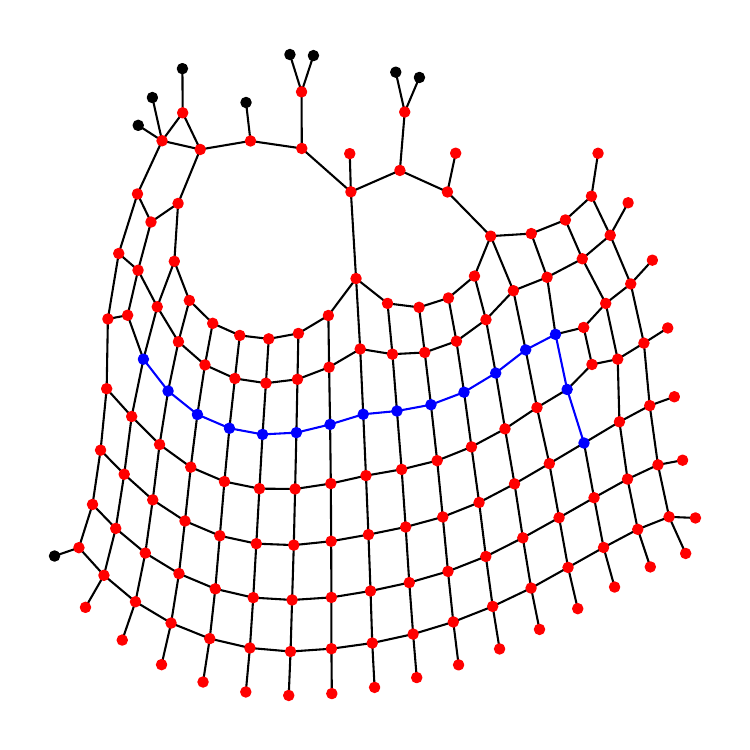}}
        \caption{$k=5$\\$C_5(P^\star) = 144$, $|P^\star| = 16$}
        \label{fig:san-francisco-5}
    \end{subfigure}%
    ~ 
    \begin{subfigure}[c]{0.32\textwidth}
        \resizebox{\textwidth}{!}{\includegraphics{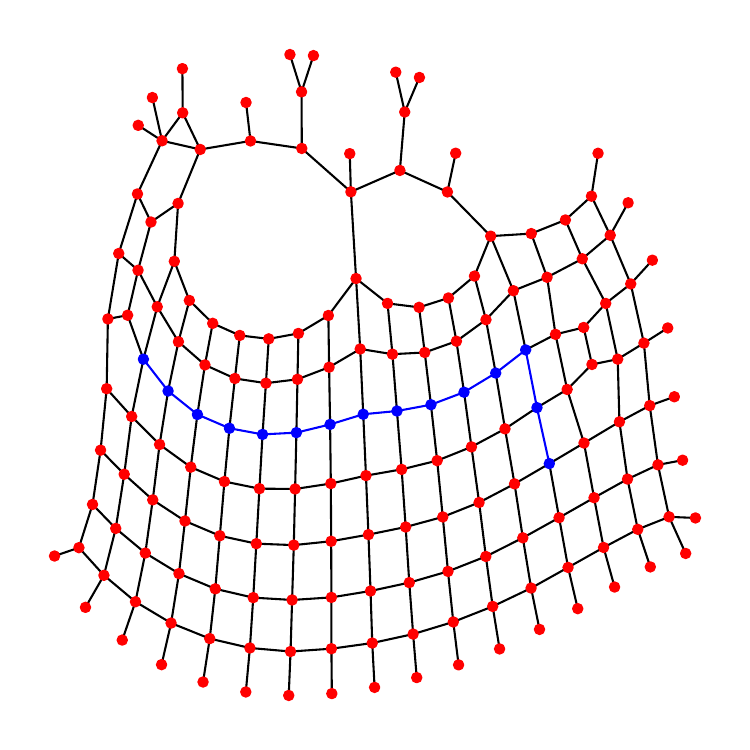}}
        \caption{$k=6$\\$C_6(P^\star) = 154$, $|P^\star| = 15$}
        \label{fig:san-francisco-6}
    \end{subfigure}%
    \vspace{1em}
    \begin{subfigure}[c]{0.32\textwidth}
        \resizebox{\textwidth}{!}{\includegraphics{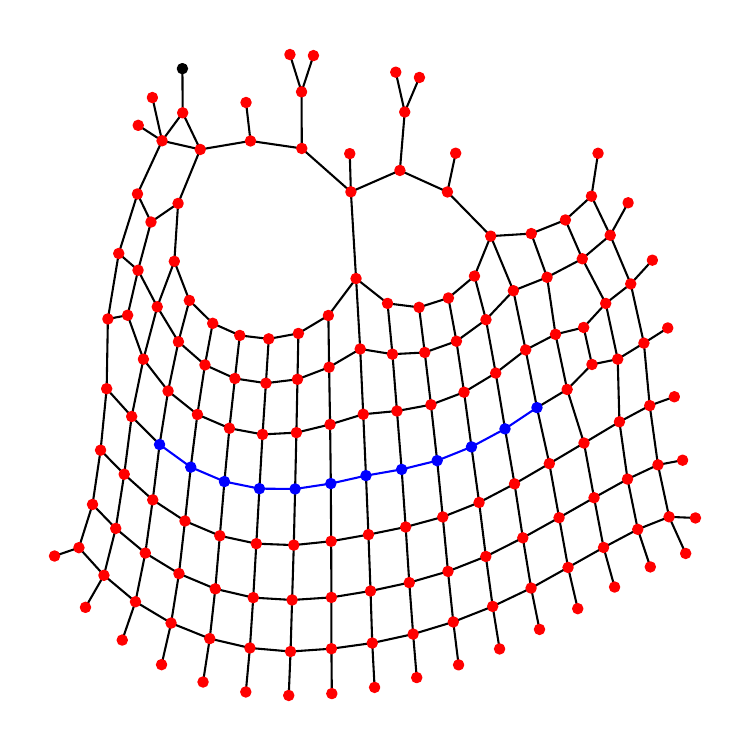}}
        \caption{$k=7$\\$C_7(P^\star) = 156$, $|P^\star| = 12$}
        \label{fig:san-francisco-7}
    \end{subfigure}%
    ~ 
    \begin{subfigure}[c]{0.32\textwidth}
        \resizebox{\textwidth}{!}{\includegraphics{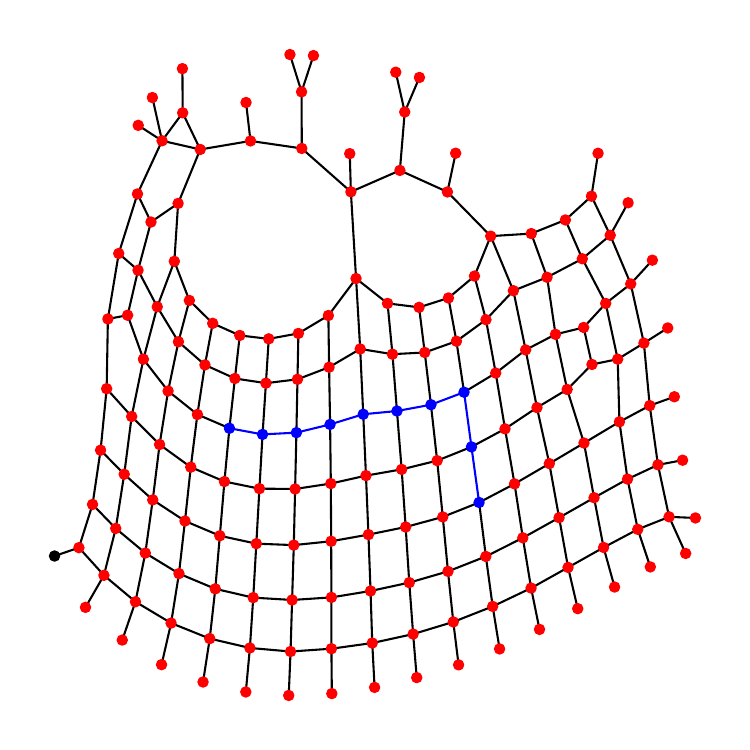}}
        \caption{$k=8$\\$C_8(P^\star) = 158$, $|P^\star| = 10$}
        \label{fig:san-francisco-8}
    \end{subfigure}%
    ~ 
    \begin{subfigure}[c]{0.32\textwidth}
        \resizebox{\textwidth}{!}{\includegraphics{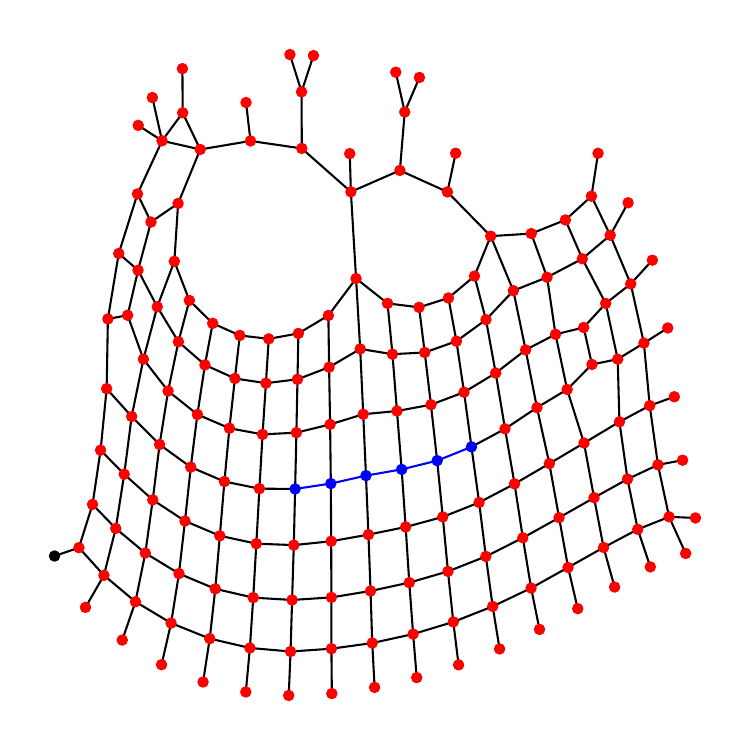}}
        \caption{$k=9$\\$C_9(P^\star) = 162$, $|P^\star| = 6$}
        \label{fig:san-francisco-9}
    \end{subfigure}
    \caption{$k$-step-central shortest paths for San Francisco (diameter: 21).}
    \label{fig:san-francisco}
\end{figure}

\begin{figure}[!htb]
    \centering
    \begin{tikzpicture}[font=\small,]
    \begin{groupplot}[
        group style={
            group size=2 by 1,
            horizontal sep=1em,
            yticklabels at=edge left,
            ylabels at=edge left,
            xlabels at=edge bottom,
        },
        x grid style={darkgray},
        xlabel={$k$},
        xmin=0, xmax=10,
        xtick style={color=black},
        y grid style={darkgray},
        ylabel={Relative length/neighbourhood of path},
        ymin=0, ymax=1.05,
        ytick style={color=black},
        ytick={0,0.2,...,1},
        legend columns=-1,
        legend style={/tikz/every even column/.append style={column sep=0.5cm}}
    ]
    \nextgroupplot[title={Relative length of path},legend to name=legendpos2]
        \addplot [draw=black, mark=o]
        table{%
        x  y
        1 0.9883720930232558
        2 0.9651162790697675
        3 0.9534883720930233
        4 0.9186046511627907
        5 0.8953488372093024
        6 0.7674418604651163
        7 0.7674418604651163
        };
    \addlegendentry{Ahmedabad}
    \addplot [draw=red, mark=x]
        table{%
        x  y
        1 0.9803921568627451
        2 0.9607843137254902
        3 0.9411764705882353
        4 0.9019607843137255
        5 0.9019607843137255
        6 0.8235294117647058
        7 0.8235294117647058
        };
    \addlegendentry{Cairo}
    \addplot [draw=blue, mark=triangle]
        table{%
        x  y
        1 0.9090909090909091
        2 0.9393939393939394
        3 0.9393939393939394
        4 0.9090909090909091
        5 0.9090909090909091
        6 0.7575757575757576
        7 0.7272727272727273
        8 0.696969696969697
        9 0.696969696969697
        };
    \addlegendentry{London}
    \addplot [draw=orange, mark=o, mark options={solid}, dashed]
        table{%
        x  y
        1 0.9523809523809523
        2 0.9047619047619048
        3 0.8571428571428571
        4 0.8571428571428571
        5 0.7619047619047619
        6 0.7142857142857143
        7 0.5714285714285714
        8 0.47619047619047616
        9 0.2857142857142857
        };
    \addlegendentry{San Francisco}
    \addplot [draw=green, mark=x, mark options={solid}, dashed]
        table{%
        x  y
        1 0.9487179487179487
        2 0.9743589743589743
        3 0.9487179487179487
        4 0.8974358974358975
        5 0.9487179487179487
        6 0.8974358974358975
        7 0.8205128205128205
        8 0.7692307692307693
        9 0.7435897435897436
        };
    \addlegendentry{Venice}
    \coordinate (c1) at (rel axis cs:0,1);
    \nextgroupplot[title={Relative size of neighbourhood}]
    \addplot [draw=red, mark=x]
        table{%
        x  y
        1 0.03832752613240418
        2 0.08606271777003484
        3 0.13658536585365855
        4 0.18954703832752615
        5 0.23902439024390243
        6 0.2951219512195122
        7 0.3480836236933798

        };
    \addplot [draw=red, mark=x]
        table{%
        x  y
        1 0.05213903743315508
        2 0.11831550802139038
        3 0.19117647058823528
        4 0.2660427807486631
        5 0.3522727272727273
        6 0.4358288770053476
        7 0.5167112299465241

        };
    \addplot [draw=blue, mark=triangle]
        table{%
        x  y
        1 0.10450819672131148
        2 0.21311475409836064
        3 0.32991803278688525
        4 0.4323770491803279
        5 0.5348360655737705
        6 0.6270491803278688
        7 0.7131147540983607
        8 0.7848360655737705
        9 0.8463114754098361

        };
    \addplot [draw=orange, mark=o, mark options={solid}, dashed]
        table{%
        x  y
        1 0.24260355029585798
        2 0.46745562130177515
        3 0.6686390532544378
        4 0.7928994082840237
        5 0.8520710059171598
        6 0.9112426035502958
        7 0.9230769230769231
        8 0.9349112426035503
        9 0.9585798816568047

        };

    \addplot [draw=brown, mark=triangle, mark options={solid}, dashed]
        table{%
        x  y
        1 0.06521739130434782
        2 0.12880434782608696
        3 0.1782608695652174
        4 0.23478260869565218
        5 0.2956521739130435
        6 0.3472826086956522
        7 0.3978260869565217

        };
        \coordinate (c2) at (rel axis cs:1,1);
    \end{groupplot}
    \coordinate (c3) at ($(c1)!.5!(c2)$);
    \node[below] at (c3 |- current bounding box.south)
      {\pgfplotslegendfromname{legendpos2}};
\end{tikzpicture}
    \caption{Effect of varying $k$ on the relative path length and size of the neighbourhood. The length of the path is divided by the diameter of the graph, while the size of the neighbourhood is divided by $|V|$.}
    \label{fig:streets-analysis}
\end{figure}

In both instances shown in Figures~\ref{fig:london} and \ref{fig:san-francisco}, the optimal paths vary drastically as $k$ changes, as evidenced by shifts in their location. For example, in the London instance, the optimal path shifts from a Northwest-Southeast routing to a Northeast-Southwest routing as $k$ increases from two to three. Similarly, in San Francisco, the optimal path shifts from a route covering the western and southern regions to an East-West route passing through the centre of the city as $k$ increases.

In San Francisco and other grid-structured cities (such as Barcelona, New York, Savannah), we observe that when $k$ is small, the $k$-step reach-central shortest path follows longer, windier routes that serve nodes on the periphery of the network, akin to ring routes. As the path's reach increases with $k$, the optimal path tends to be shorter, shift inward, and become more `central' in the network. Figures~\ref{fig:streets-analysis} and \ref{fig:streets-analysis-appendix} make the general effect of varying $k$ clear. For example, when $k$ is doubled from 1 to 2, the size of the neighbourhood doubles, even though the paths do not become longer. For smaller values of $k$ (less than 5), increasing $k$ leads to a linear increase in the size of the neighbourhood and a reduction in path length. For example, increasing $k$ from 1 to 5 results, on average, in a five-fold increase in the size of the neighbourhood and a 30\% reduction in the length of the optimal path. This result holds even with larger baseline values of $k$. Increasing the reach of a path yields more extensive coverage, but the required path length decreases. This observation demonstrates the importance of providing and improving supporting infrastructure when planning mass transit routes. When cities become more accessible, the reach of a transit route can grow, even with a shorter route, illustrating the need to balance resources between establishing the route itself and supporting accessibility projects, especially when the costs and risks of both options are considered. Construction of mass transit routes is considerably more expensive and riskier than projects that improve active transport infrastructure, implying that marginal reductions in route length may be more significant than the costs of improving accessibility.

Our results for the San Francisco network also offer insight into the subtleties of our chosen definition of $k$-step reach centrality. As described in Sections~\ref{sec:problem-definition} and \ref{sec:closeness}, there are two definitions of $k$-step reach centrality, open and closed. We use the open version in Algorithm~\ref{alg:algorithm}. From Figure~\ref{fig:san-francisco-6}, $C_6(P^\star) + |P^\star| = 154 + 15 = 169 = |V|$, so the 6-step-central shortest path and its neighbourhood cover the whole network. This implies that the closeness-central shortest path has a closeness centrality score of at most 6 (it may be smaller, as we are using the open neighbourhood definition). However, when $k$ is increased to seven, the centrality of the 7-step-central shortest path is higher, but $C_7(P^\star) + |P^\star| = 168 < |V|$. Again, this is because the centrality measure we employ does not include the path in the neighbourhood.

Using these results, we were also able to upper-bound the closeness centrality of the closeness-central path for seven additional cities. The bounds are presented in Table~\ref{tab:closeness-ub}, along with the computation time. Solutions for some of these larger instances are provided in Appendix~\ref{sec:appendix-figures}.

\begin{table}[!htb]
    \centering
    \small
    \begin{tabular}{lrrrrrrr}
        \toprule
        Instance            & $|V|$ & $|E|$ & diam & $\tilde{k}$ & $|P|$ & $C_{\tilde{k}}(P)$ & Time (s) \\ \midrule
        Barcelona           & 210   & 323   & 24   & 7           & 11    & 199                &   789.19 \\
        Irvine1             & 32    & 36    & 8    & 3           & 4     & 28                 &     0.03 \\
        Irvine2$\,^\dagger$ & 217   & 222   & 23   & 8           & 8     & 170                &    26.04 \\
        Los Angeles         & 240   & 339   & 22   & 9           & 6     & 234                &   331.13 \\
        New York            & 248   & 418   & 24   & 9           & 10    & 238                &  1152.60 \\
        Paris               & 335   & 494   & 28   & 9           & 16    & 319                &   513.80 \\
        San Francisco       & 169   & 271   & 21   & 6           & 15    & 154                &  1079.24 \\
        Walnut Creek        & 169   & 196   & 22   & 8           & 9     & 160                &    30.10 \\
        Washington          & 192   & 302   & 21   & 8           & 9     & 183                &   254.62 \\ \bottomrule
    \end{tabular}
    \caption{Upper bounds for the closeness-central shortest path for Urban Streets networks. $\tilde{k}$ is the upper bound for Problem~\eqref{eq:closest-path-problem}. $^\dagger$ To compute the result for the Irvine2 instance, we used its largest connected component.}
    \label{tab:closeness-ub}    
\end{table}

Finally, these figures also demonstrate the effect of varying $k$ in the $k$-step reach centrality. When $k$ is small, the $k$-step reach centrality acts as a local measure of centrality, emphasising the immediate neighbours of the path. As $k$ increases, the measure becomes more global. This aligns with our earlier observation that the problem of finding the closeness-central shortest path can be solved by solving the $k$-step reach-central shortest path problem multiple times.

\section{Conclusion}\label{sec:conclusion}

In this paper, we study the problem of finding the $k$-step-central shortest path. We provide a polynomial-time algorithm for this problem on unweighted graphs and show that it is NP-hard on graphs with weighted edges. We also demonstrate that our algorithm can be used to find the closeness-central shortest path because, as $k$ grows, our centrality measure becomes more global and can mimic closeness-centrality. We complement our theoretical analysis with an extensive set of numerical experiments that demonstrate the efficiency of our algorithm on graphs with up to 2000 nodes.

Our numerical experiments also yield practical insights that urban planners can use when designing direct mass-transit systems. \textit{Shorter direct transit routes with fewer stops and stations can provide broader coverage than longer, winding routes as cities become more accessible.} Increasing the number of blocks a commuter is willing to travel to a transport hub results in a linear increase in a route's reachability, even for shorter, more compact routes. This demonstrates that urban infrastructure projects must not be considered in isolation if maximising ridership is the primary objective. Improving active transport infrastructure, such as walking paths and dedicated bike lanes, can greatly enhance a transit route's overall accessibility and ridership at a significantly lower cost than constructing longer transit routes.

Future research in this area can examine alternative approaches to the closeness-central shortest path problem, such as integer programming formulations or approximation and heuristic algorithms. Additionally, weights can be assigned to nodes and edges to better model population density and topography.

\bibliographystyle{apalike}
\bibliography{refs} 

@article{Matsypura2023,
    title = {Finding the most degree-central walks and paths in a graph: Exact and heuristic approaches},
    journal = {European Journal of Operational Research},
    volume = {308},
    number = {3},
    pages = {1021-1036},
    year = {2023},
    issn = {0377-2217},
    doi = {10.1016/j.ejor.2022.12.014},
    author = {Dmytro Matsypura and Alexander Veremyev and Eduardo L. Pasiliao and Oleg A. Prokopyev}
}

@article{PhosavanhMatsypura2025,
title = {Centrality of shortest paths: algorithms and complexity results},
journal = {INFORMS Journal on Computing},
year = {2025},
doi = {10.1287/ijoc.2024.0945},
author = {Johnson Phosavanh and Dmytro Matsypura},
note = {Advance online publication.},
}

@article{BarabasiAlbert1999,
    language = {eng},
    number = {5439},
    pages = {509-512},
    publisher = {American Society for the Advancement of Science},
    title = {Emergence of Scaling in Random Networks},
    volume = {286},
    year = {1999},
    author = {Barabási, Albert-László and Albert, Réka},
    address = {Washington, DC},
    copyright = {Copyright 1999 American Association for the Advancement of Science},
    issn = {0036-8075},
    journal = {Science (American Association for the Advancement of Science)},
    doi = {10.1126/science.286.5439.509}
}

@article{WattsStrogatz1998,
    language = {eng},
    number = {6684},
    pages = {440-442},
    publisher = {Nature Publishing},
    title = {Collective dynamics of `small-world' networks},
    volume = {393},
    year = {1998},
    author = {Watts, Duncan J and Strogatz, Steven H},
    address = {London},
    copyright = {1998 INIST-CNRS},
    issn = {0028-0836},
    journal = {Nature (London)},
    doi = {10.1038/30918}
}

@article{CrucittiPaolo2006,
author = {Crucitti, Paolo and Latora, Vito and Porta, Sergio},
address = {United States},
issn = {1539-3755},
journal = {Physical review. E, Statistical, nonlinear, and soft matter physics},
language = {eng},
number = {3},
pages = {036125-},
title = {Centrality measures in spatial networks of urban streets},
volume = {73},
year = {2006},
}

@techreport{TransportChoice,
  title                    = {Improving Transport Choice — Guidelines for planning and development},
  author                   = {{NSW Government}},
  institution              = {NSW Department of Urban Affairs and Planning},
  year                     = {2001},
  location                 = {New South Wales, Australia},
  type                     = {Government report}
}

@article{Everett1999,
    author = {M. G. Everett and S. P. Borgatti},
    title = {The centrality of groups and classes},
    journal = {The Journal of Mathematical Sociology},
    volume = {23},
    number = {3},
    pages = {181--201},
    year = {1999},
    publisher = {Routledge},
    doi = {10.1080/0022250X.1999.9990219}
}

@article{deSa2015,
    ISSN = {00411655, 15265447},
    URL = {http://www.jstor.org/stable/43666754},
    author = {Elisangela Martins de Sá and Ivan Contreras and Jean-François Cordeau and Ricardo Saraiva de Camargo and Gilberto de Miranda},
    journal = {Transportation Science},
    number = {3},
    pages = {500--518},
    publisher = {INFORMS},
    title = {The Hub Line Location Problem},
    urldate = {2023-06-22},
    volume = {49},
    year = {2015},
    doi = {10.1287/trsc.2014.0576}
}

@article{Camus2024,
    title={A Survey on Optimization Studies of Group Centrality Metrics}, 
    author={Mustafa Can Camur and Chrysafis Vogiatzis},
    year={2024},
    journal = {Networks},
    volume = {2024},
    pages = {1-18},
    doi = {10.1002/net.22248},
    URL = {10.1002/net.22248},
    publisher = {Wiley},
}

@book{GareyJohnson1979,
    lccn = {78012361},
    publisher = {W. H. Freeman},
    series = {A Series of books in the mathematical sciences.},
    title = {{Computers and intractability: a guide to the theory of NP-completeness}},
    year = {1979},
    author = {Garey, Michael R. and Johnson, David S.},
    address = {San Francisco},
    booktitle = {{Computers and intractability: A guide to the theory of NP-completeness}},
    isbn = {0716710447},
    language = {eng},
}

@book{Vazirani2003,
    author = {Vazirani, Vijay V.},
    address = {Berlin, Heidelberg},
    booktitle = {Approximation Algorithms},
    edition = {2nd},
    isbn = {3-662-04565-6},
    language = {eng},
    publisher = {Springer Berlin Heidelberg},
    title = {Approximation Algorithms },
    year = {2003},
}

@article{Slater1982,
 ISSN = {00411655, 15265447},
 author = {Peter J. Slater},
 journal = {Transportation Science},
 number = {1},
 pages = {1--18},
 publisher = {INFORMS},
 title = {Locating Central Paths in a Graph},
 urldate = {2024-08-27},
 volume = {16},
 year = {1982},
 doi = {10.1287/trsc.16.1.1}
}

@article{hedetniemi1981jordan,
  author  = {Hedetniemi, S. Mitchell and Cockayne, E. J. and Hedetniemi, S. T.},
  title   = {{Linear Algorithms for Finding the Jordan Center and Path Center of a Tree}},
  journal = {Transportation Science},
  year    = {1981},
  volume  = {15},
  number  = {2},
  pages   = {98--114},
  publisher = {INFORMS},
  doi     = {10.1287/trsc.15.2.98}
}

@article{Gomez2023,
title = {Path eccentricity of graphs},
journal = {Discrete Applied Mathematics},
volume = {337},
pages = {1-13},
year = {2023},
issn = {0166-218X},
doi = {10.1016/j.dam.2023.04.012},
author = {Renzo G\'{o}mez and Juan Guti\'{e}rrez}
}

@article{dawande2012,
author = {Dawande, Milind and Mookerjee, Vijay and Sriskandarajah, Chelliah and Zhu, Yunxia},
title = {Structural Search and Optimization in Social Networks},
journal = {INFORMS Journal on Computing},
volume = {24},
number = {4},
pages = {611-623},
year = {2012},
doi = {10.1287/ijoc.1110.0473},
URL = {10.1287/ijoc.1110.0473},
eprint = {10.1287/ijoc.1110.0473}
}

@article{Borgatti2006,
author = {Borgatti, Stephen P.},
issn = {1381-298X},
journal = {Computational and Mathematical Organization Theory},
language = {eng ; jpn},
number = {1},
pages = {21-34},
title = {Identifying sets of key players in a social network},
volume = {12},
year = {2006},
}

@article{Camur2021,
    author = {Camur, Mustafa Can and Sharkey, Thomas and Vogiatzis, Chrysafis},
    title = {The Star Degree Centrality Problem: A Decomposition Approach},
    journal = {INFORMS Journal on Computing},
    volume = {34},
    number = {1},
    pages = {93-112},
    year = {2021},
    doi = {10.1287/ijoc.2021.1074},
    URL = {10.1287/ijoc.2021.1074}
}

@article{Vogiatzis2014,
    author = {Vogiatzis, Chrysafis and Veremyev, Alexander and Pasiliao, Eduardo and Pardalos, Panos},
    year = {2014},
    month = {04},
    pages = {1-19},
    title = {An integer programming approach for finding the most and the least central cliques},
    volume = {9},
    journal = {Optimization Letters},
    doi = {10.1007/s11590-014-0782-2}
}

@article{Vogiatzis2019,
    author = {Vogiatzis, Chrysafis and Camur, Mustafa Can},
    title = {Identification of Essential Proteins Using Induced Stars in Protein–Protein Interaction Networks},
    journal = {INFORMS Journal on Computing},
    volume = {31},
    number = {4},
    pages = {703-718},
    year = {2019},
    doi = {10.1287/ijoc.2018.0872},
    URL = {10.1287/ijoc.2018.0872}
}

@techreport{NSWTransportAccessibility,
    title        = {Public transport accessibility - Built environment indicator},
    author       = {{NSW Government}},
    institution  = {NSW Government},
    year         = {2023},
    location     = {New South Wales, Australia},
    type         = {Government report},
    url          = {https://www.movementandplace.nsw.gov.au/sites/default/files/2023-02/BEI%204_Public%20transport%20accessibility.pdf}
}

@article{Taylor2024Distances,
    author = {Taylor, Michael A. P. and Somenahalli, Sekhar},
    copyright = {2024 The Author(s). Published with license by Taylor & Francis Group, LLC 2024},
    issn = {1556-8318},
    journal = {International Journal of Sustainable Transportation},
    language = {eng},
    number = {7},
    pages = {576-588},
    publisher = {Taylor & Francis},
    title = {{Distributions of walking access to public transport in Melbourne, Australia - Evidence on acceptable and tolerable walking distances}},
    volume = {18},
    year = {2024},
}

@article{puerto2018updatedreview,
  author  = {Puerto, Justo and Ricca, Federica and Scozzari, Andrea},
  title   = {Extensive Facility Location Problems on Networks: An Updated Review},
  journal = {Transactions in Operations Research},
  year    = {2018},
  volume  = {26},
  number  = {2},
  pages   = {187--226},
  doi     = {10.1007/s11750-018-0476-5}
}

@article{Zeng2019coveringpath,
    author = {Liwei Zeng and Sunil Chopra and Karen Smilowitz},
    journal = {Transportation Science},
    number = {6},
    pages = {1656--1672},
    publisher = {INFORMS},
    title = {The Covering Path Problem on a Grid},
    urldate = {2026-06-03},
    volume = {53},
    year = {2019}
}

@article{Current1989coveringTSP,
    author = {John R. Current and David A. Schilling},
    journal = {Transportation Science},
    number = {3},
    pages = {208--213},
    publisher = {INFORMS},
    title = {The Covering Salesman Problem},
    urldate = {2026-06-03},
    volume = {23},
    year = {1989}
}

\newpage
\appendix

\section*{Appendix}
\section{Additional results for Urban Streets}\label{sec:appendix-tables}

\begin{table}[!htb]
    \centering
    \small
\begin{tabular}{lrrrrrrr}
\toprule
  ~                        & Barcelona  & Bologna  & Brasilia  & Irvine1  & Irvine2 & Los Angeles & New Delhi \\ \cmidrule(lr){2-8}
  $|V|$                    & 210        & 541      & 179       & 32       & 217     & 240         & 252       \\
  $|E|$                    & 323        & 771      & 230       & 36       & 222     & 339         & 328       \\
  diam                     & 24         & 38       & 22        & 8        & 23      & 22          & 25        \\
  $|\mathcal{SP}(G)|/U(G)$ & 65.03      & 10.88    & 2.15      & 1.30     & 1.31    & 10.80       & 2.95      \\ \midrule \multicolumn{1}{c}{} & \multicolumn{7}{c}{$k=1$} \\ \cmidrule(lr){2-8}
  diam centrality          & 43         & 55       & 26        & 12       & 28      & 38          & 33        \\
  path length              & 22         & 35       & 19        & 6        & 21      & 20          & 23        \\
  path centrality          & 44         & 57       & 28        & 14       & 30      & 40          & 35        \\
  time                     & 4.07       & 71.49    & 1.88      & 0.01     & 2.02    & 5.01        & 5.70      \\ \midrule \multicolumn{1}{c}{} & \multicolumn{7}{c}{$k=2$} \\ \cmidrule(lr){2-8}
  diam centrality          & 87         & 116      & 50        & 18       & 58      & 83          & 64        \\
  path length              & 22         & 33       & 15        & 4        & 17      & 18          & 19        \\
  path centrality          & 88         & 120      & 56        & 24       & 65      & 86          & 69        \\
  time                     & 9.94       & 99.73    & 2.47      & 0.01     & 2.38    & 9.05        & 7.79      \\ \midrule \multicolumn{1}{c}{} & \multicolumn{7}{c}{$k=3$} \\ \cmidrule(lr){2-8}
  diam centrality          & 126        & 180      & 74        & 23       & 88      & 132         & 93        \\
  path length              & 18         & 31       & 15        & 3        & 15      & 18          & 19        \\
  path centrality          & 130        & 187      & 85        & 28       & 101     & 134         & 100       \\
  time                     & 26.39      & 161.11   & 3.27      & 0.01     & 2.97    & 17.73       & 11.08     \\ \midrule \multicolumn{1}{c}{} & \multicolumn{7}{c}{$k=4$} \\ \cmidrule(lr){2-8}
  diam centrality          & 149        & 238      & 94        & 23       & 112     & 169         & 120       \\
  path length              & 16         & 30       & 13        & 1        & 13      & 16          & 18        \\
  path centrality          & 157        & 248      & 109       & 30       & 128     & 175         & 132       \\
  time                     & 62.59      & 271.94   & 3.91      & 0.01     & 3.57    & 30.72       & 14.93     \\ \midrule \multicolumn{1}{c}{} & \multicolumn{7}{c}{$k=5$} \\ \cmidrule(lr){2-8}
  diam centrality          & 168        & 283      & 106       & 23       & 130     & 195         & 149       \\
  path length              & 15         & 27       & 12        & 0        & 11      & 14          & 17        \\
  path centrality          & 177        & 296      & 128       & 31       & 141     & 203         & 165       \\
  time                     & 126.64     & 435.37   & 4.38      & 0.02     & 3.85    & 44.16       & 18.63     \\ 
\bottomrule
\end{tabular}
    \caption{Further results for Urban Streets graphs.}
    \label{tab:urban-streets-appendix1}
\end{table}

\begin{table}[!htb]
    \centering
    \small
\begin{tabular}{lrrrrrrrr}
\toprule
  ~                        & New York  & Paris  & Richmond  & Savannah & Seoul    & Vienna  & Walnut Creek  & Washington \\ \cmidrule(lr){2-9}
  $|V|$                    & 248       & 335    & 697       & 584      & 869      & 467     & 169           & 192        \\
  $|E|$                    & 418       & 494    & 1084      & 958      & 1307     & 691     & 196           & 302        \\
  diam                     & 24        & 28     & 49        & 39       & 50       & 32      & 22            & 21         \\
  $|\mathcal{SP}(G)|/U(G)$ & 40.34     & 5.35   & 14.39     & 76.10    & 22.20    & 14.43   & 2.16          & 22.85      \\ \midrule \multicolumn{1}{c}{} & \multicolumn{8}{c}{$k=1$} \\ \cmidrule(lr){2-9}
  diam centrality          & 45        & 44     & 79        & 71       & 75       & 48      & 28            & 42         \\
  path length              & 22        & 25     & 46        & 36       & 48       & 30      & 19            & 19         \\
  path centrality          & 46        & 48     & 81        & 73       & 77       & 50      & 31            & 44         \\
  time                     & 5.90      & 14.12  & 165.85    & 106.51   & 378.42   & 44.98   & 1.67          & 2.86       \\ \midrule \multicolumn{1}{c}{} & \multicolumn{8}{c}{$k=2$} \\ \cmidrule(lr){2-9}
  diam centrality          & 84        & 90     & 165       & 136      & 156      & 106     & 53            & 73         \\
  path length              & 22        & 24     & 45        & 37       & 46       & 28      & 18            & 19         \\
  path centrality          & 86        & 96     & 168       & 139      & 158      & 108     & 60            & 75         \\
  time                     & 11.81     & 19.78  & 232.30    & 185.92   & 501.88   & 67.24   & 2.26          & 6.53       \\ \midrule \multicolumn{1}{c}{} & \multicolumn{8}{c}{$k=3$} \\ \cmidrule(lr){2-9}
  diam centrality          & 123       & 141    & 255       & 211      & 246      & 157     & 76            & 105        \\
  path length              & 22        & 23     & 43        & 36       & 44       & 28      & 17            & 17         \\
  path centrality          & 126       & 147    & 259       & 213      & 251      & 161     & 89            & 109        \\
  time                     & 25.06     & 32.06  & 397.84    & 391.88   & 794.65   & 117.37  & 3.23          & 15.31      \\ \midrule \multicolumn{1}{c}{} & \multicolumn{8}{c}{$k=4$} \\ \cmidrule(lr){2-9}
  diam centrality          & 152       & 179    & 347       & 272      & 329      & 207     & 98            & 126        \\
  path length              & 18        & 23     & 41        & 34       & 42       & 26      & 15            & 16         \\
  path centrality          & 156       & 191    & 352       & 292      & 336      & 212     & 113           & 135        \\
  time                     & 49.58     & 48.90  & 683.11    & 801.47   & 1347.74  & 210.35  & 4.20          & 28.73      \\ \midrule \multicolumn{1}{c}{} & \multicolumn{8}{c}{$k=5$} \\ \cmidrule(lr){2-9}
  diam centrality          & 171       & 213    & 415       & 325      & 404      & 256     & 117           & 143        \\
  path length              & 16        & 20     & 40        & 36       & 41       & 24      & 14            & 13         \\
  path centrality          & 181       & 229    & 421       & 370      & 413      & 263     & 131           & 153        \\
  time                     & 86.73     & 66.45  & 1094.52   & 1530.80  & 2258.93  & 344.19  & 4.80          & 43.51      \\ 
\bottomrule
\end{tabular}
    \caption{Further results for Urban Streets graphs.}
    \label{tab:urban-streets-appendix2}
\end{table}

\clearpage

\section{Additional figures for Urban Streets}\label{sec:appendix-figures}

\begin{figure}[!htb]
    \centering
    \includestandalone[width=\linewidth]{figures-tikz/analysis-streets-appendix}
    \caption{Effect of varying $k$ on the relative path length and size of the neighbourhood. The length of the path is divided by the diameter of the graph, while the size of the neighbourhood is divided by $|V|$.}
    \label{fig:streets-analysis-appendix}
\end{figure}

\begin{figure*}[!htb]
    \centering
    \captionsetup{justification=centering}
    \begin{subfigure}[c]{0.32\textwidth}
        \resizebox{\textwidth}{!}{\includegraphics{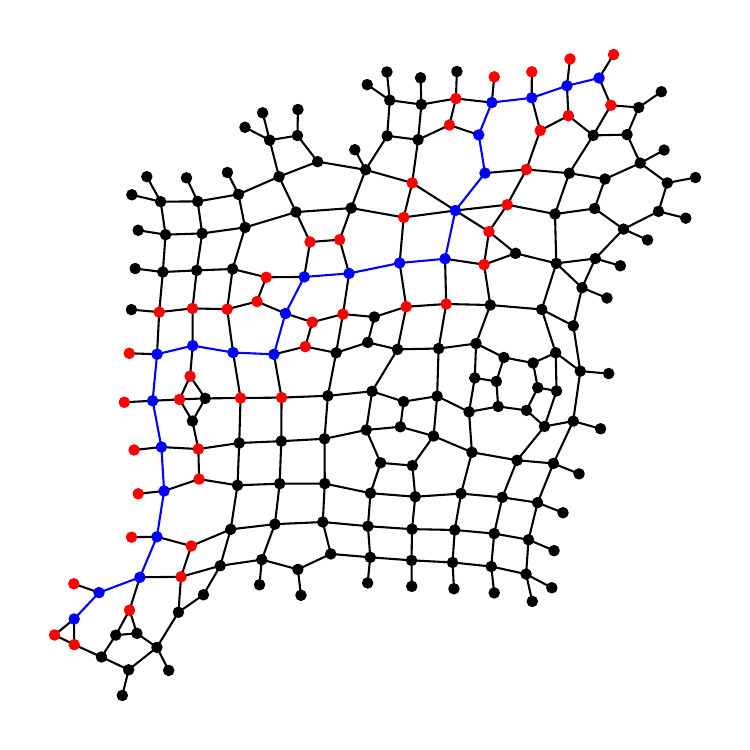}}
        \caption{$k=1$\\$C_1(P^\star) = 44$, $|P^\star| = 23$}
        \label{fig:barcelona-1}
    \end{subfigure}%
    ~ 
    \begin{subfigure}[c]{0.32\textwidth}
        \resizebox{\textwidth}{!}{\includegraphics{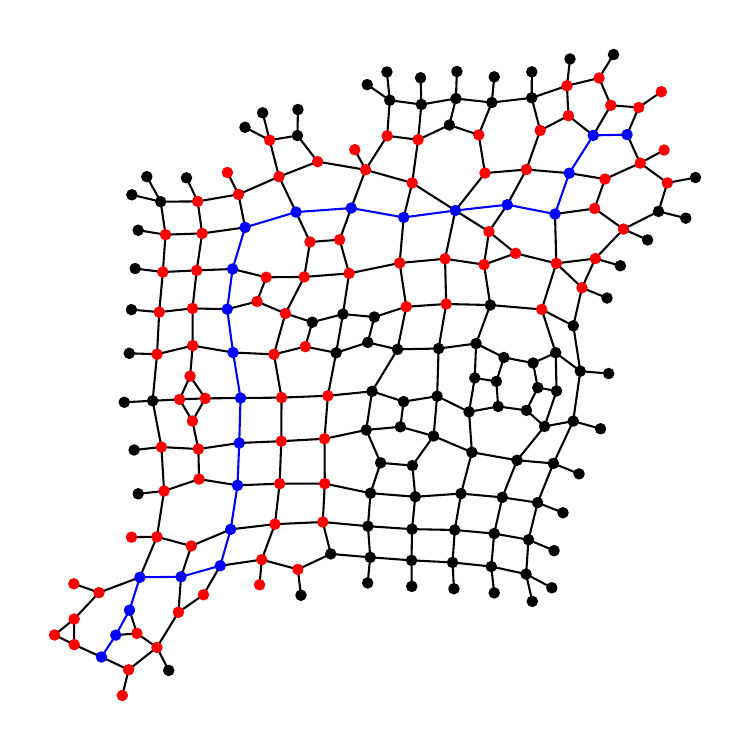}}
        \caption{$k=2$\\$C_2(P^\star) = 88$, $|P^\star| = 23$}
        \label{fig:barcelona-2}
    \end{subfigure}%
    ~ 
    \begin{subfigure}[c]{0.32\textwidth}
        \resizebox{\textwidth}{!}{\includegraphics{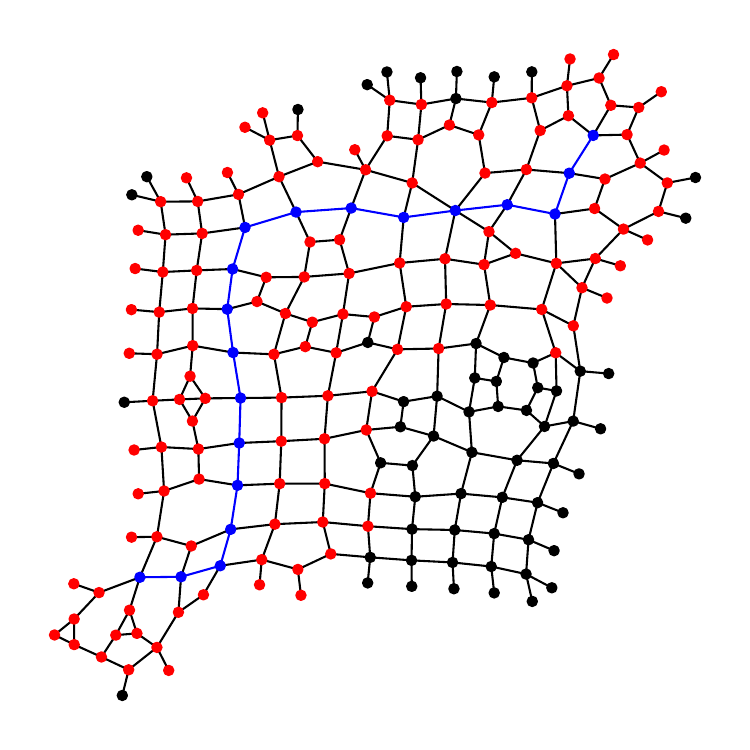}}
        \caption{$k=3$\\$C_3(P^\star) = 130$, $|P^\star| = 19$}
        \label{fig:barcelona-3}
    \end{subfigure}%
    \vspace{1em}
    \begin{subfigure}[c]{0.32\textwidth}
        \resizebox{\textwidth}{!}{\includegraphics{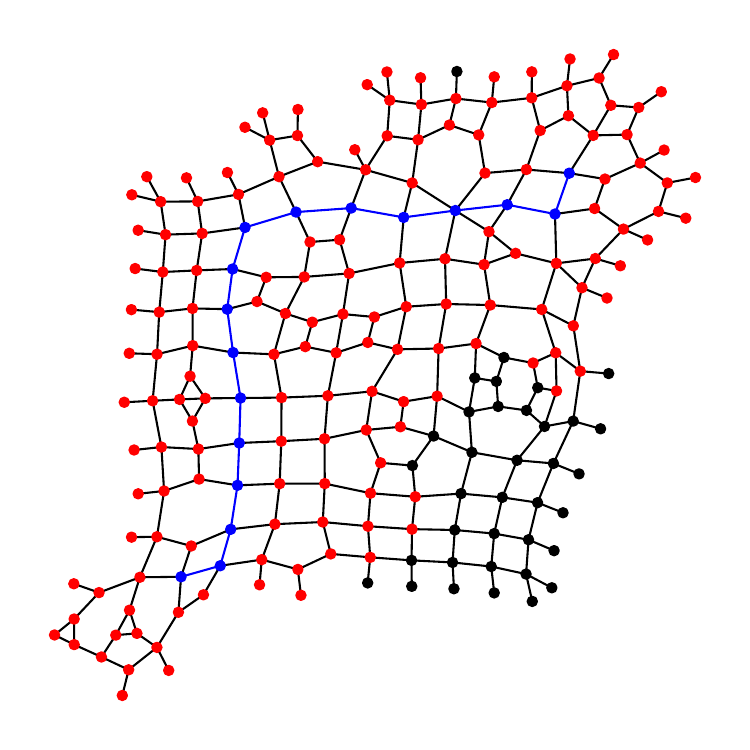}}
        \caption{$k=4$\\$C_4(P^\star) = 157$, $|P^\star| = 17$}
        \label{fig:barcelona-4}
    \end{subfigure}%
    ~ 
    \begin{subfigure}[c]{0.32\textwidth}
        \resizebox{\textwidth}{!}{\includegraphics{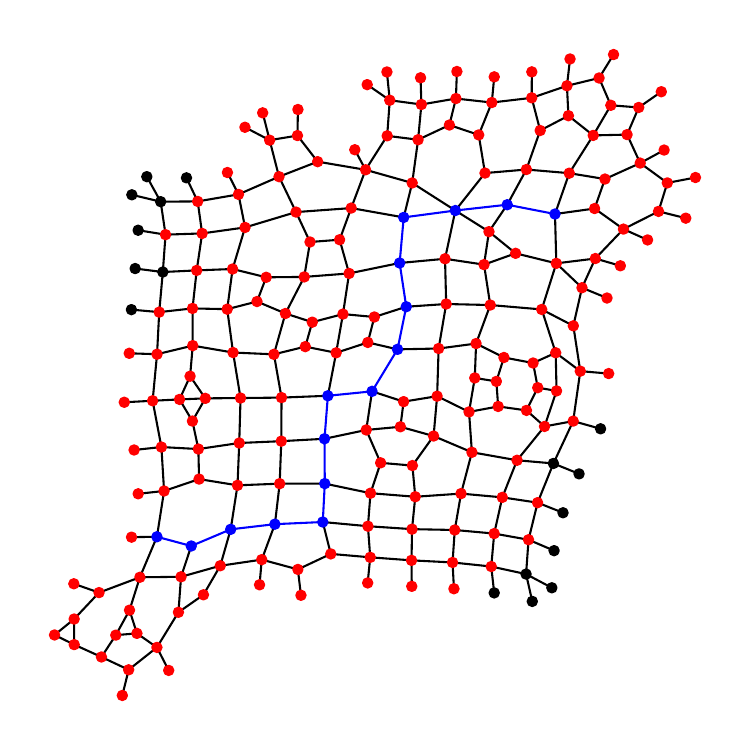}}
        \caption{$k=5$\\$C_5(P^\star) = 177$, $|P^\star| = 16$}
        \label{fig:barcelona-5}
    \end{subfigure}%
    ~ 
    \begin{subfigure}[c]{0.32\textwidth}
        \resizebox{\textwidth}{!}{\includegraphics{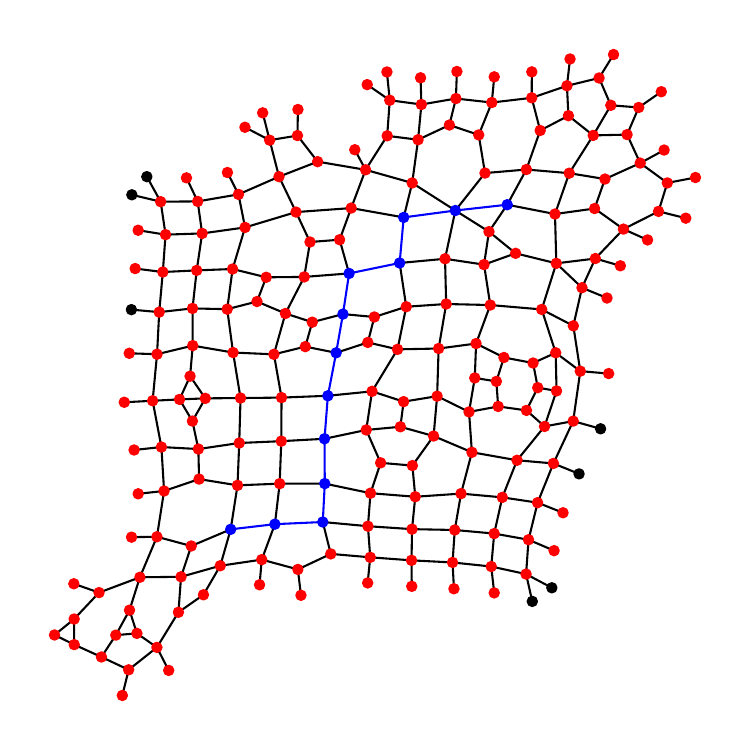}}
        \caption{$k=6$\\$C_6(P^\star) = 190$, $|P^\star| = 13$}
        \label{fig:barcelona-6}
    \end{subfigure}%
    \vspace{1em}
    \begin{subfigure}[c]{0.32\textwidth}
        \resizebox{\textwidth}{!}{\includegraphics{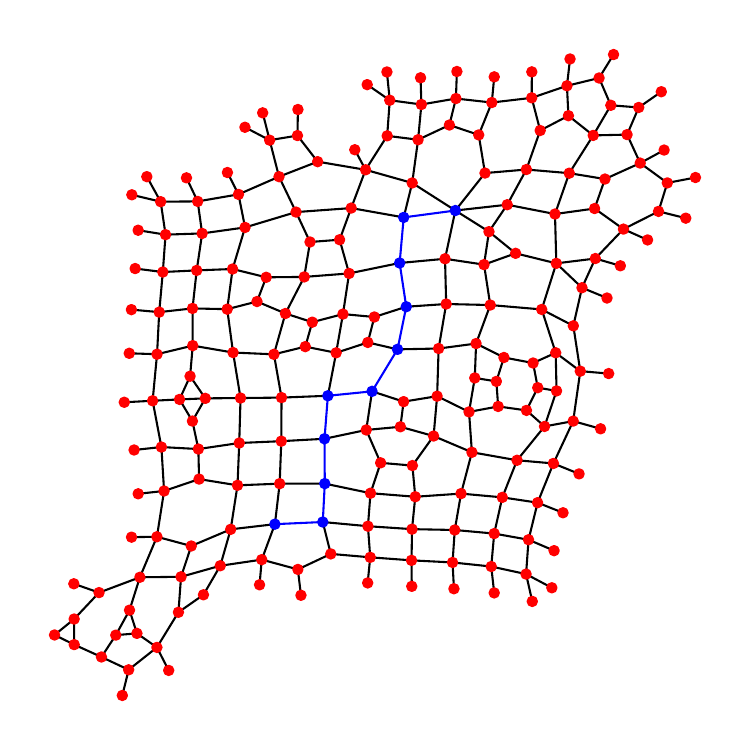}}
        \caption{$k=7$\\$C_7(P^\star) = 199$, $|P^\star| = 11$}
        \label{fig:barcelona-7}
    \end{subfigure}%
    ~ 
    \begin{subfigure}[c]{0.32\textwidth}
        \resizebox{\textwidth}{!}{\includegraphics{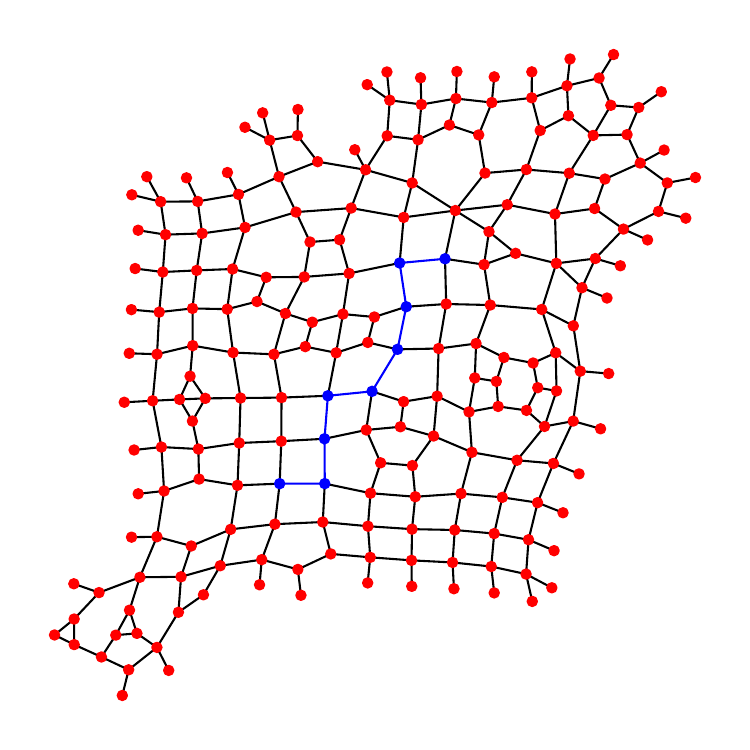}}
        \caption{$k=8$\\$C_8(P^\star) = 201$, $|P^\star| = 9$}
        \label{fig:barcelona-8}
    \end{subfigure}%
    ~ 
    \begin{subfigure}[c]{0.32\textwidth}
        \resizebox{\textwidth}{!}{\includegraphics{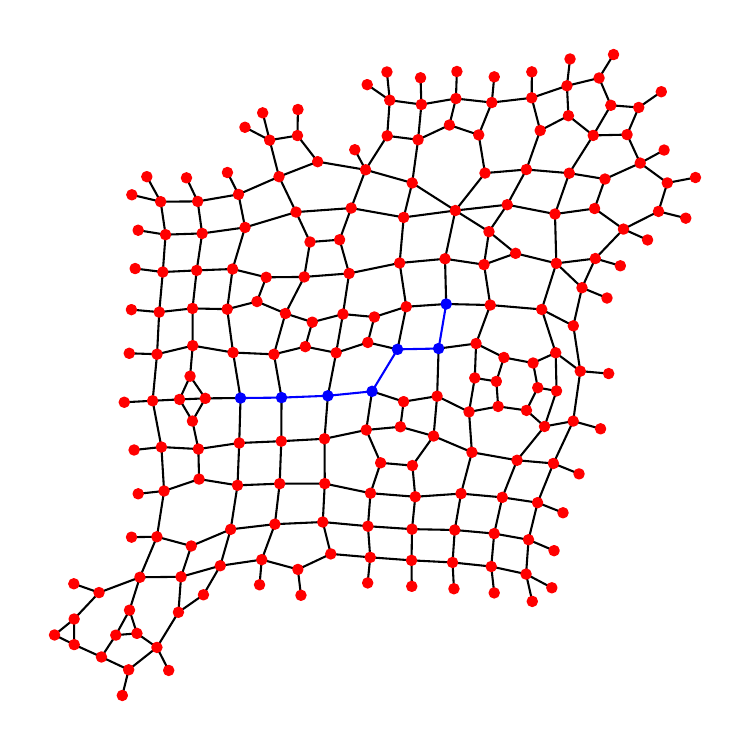}}
        \caption{$k=9$\\$C_9(P^\star) = 203$, $|P^\star| = 7$}
        \label{fig:barcelona-9}
    \end{subfigure}
    \caption{Most $k$-step reach-central shortest paths for Barcelona (diameter: 24).}
    \label{fig:barcelona}
\end{figure*}

\begin{figure*}[!htb]
    \centering
    \captionsetup{justification=centering}
    \begin{subfigure}[c]{0.32\textwidth}
        \resizebox{\textwidth}{!}{\includegraphics{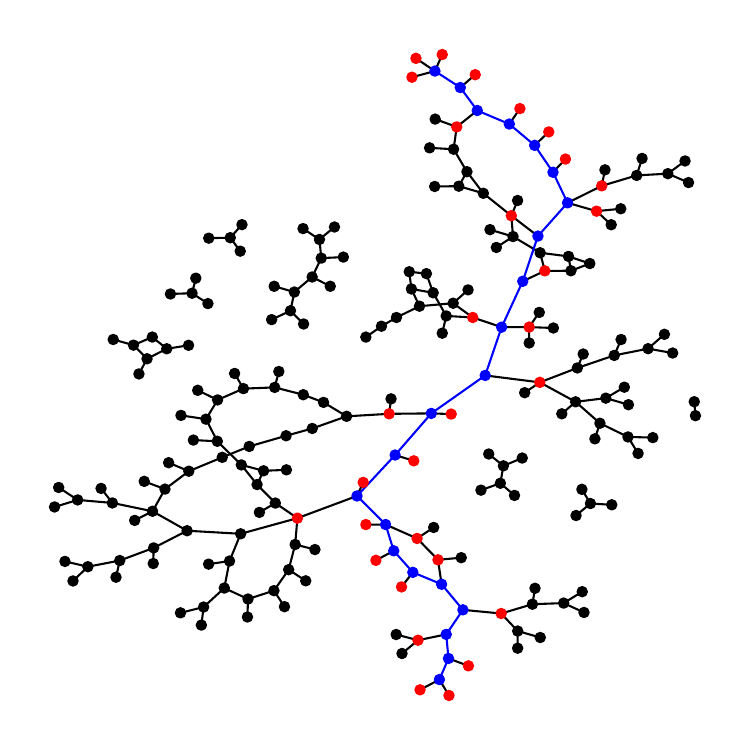}}
        \caption{$k=1$\\$C_1(P^\star) = 30$, $|P^\star| = 22$}
        \label{fig:irvine2-1}
    \end{subfigure}%
    ~ 
    \begin{subfigure}[c]{0.32\textwidth}
        \resizebox{\textwidth}{!}{\includegraphics{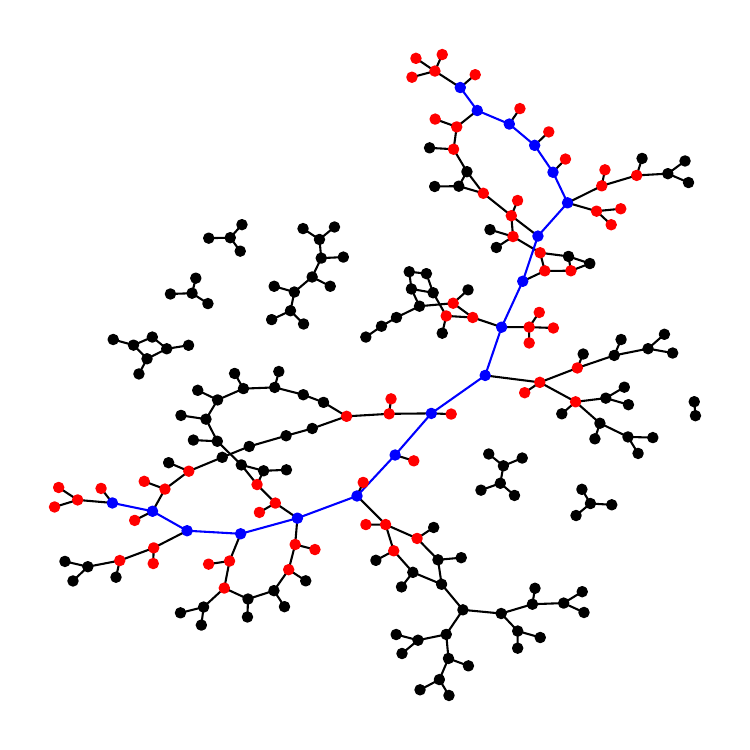}}
        \caption{$k=2$\\$C_2(P^\star) = 65$, $|P^\star| = 18$}
        \label{fig:irvine2-2}
    \end{subfigure}%
    ~ 
    \begin{subfigure}[c]{0.32\textwidth}
        \resizebox{\textwidth}{!}{\includegraphics{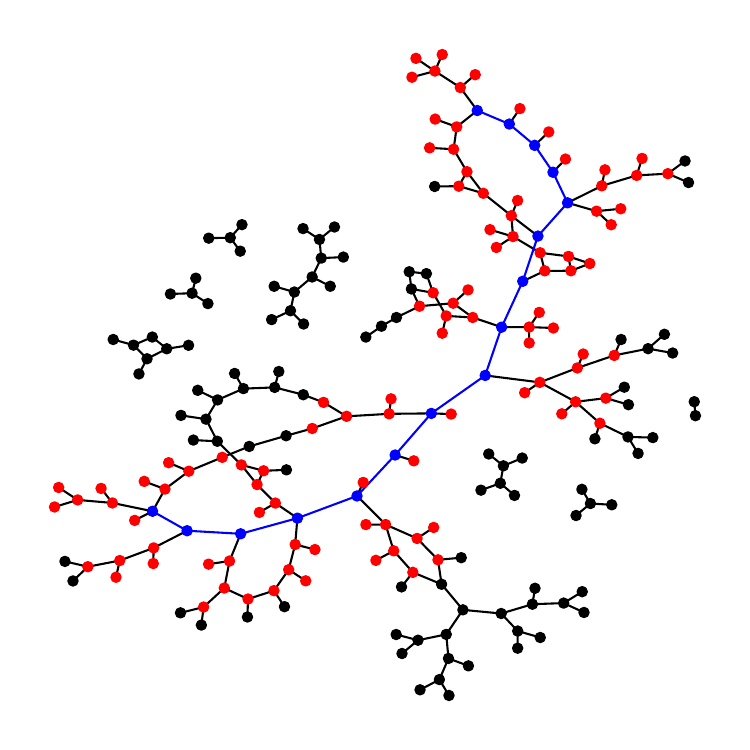}}
        \caption{$k=3$\\$C_3(P^\star) = 101$, $|P^\star| = 16$}
        \label{fig:irvine2-3}
    \end{subfigure}%
    \vspace{1em}
    \begin{subfigure}[c]{0.32\textwidth}
        \resizebox{\textwidth}{!}{\includegraphics{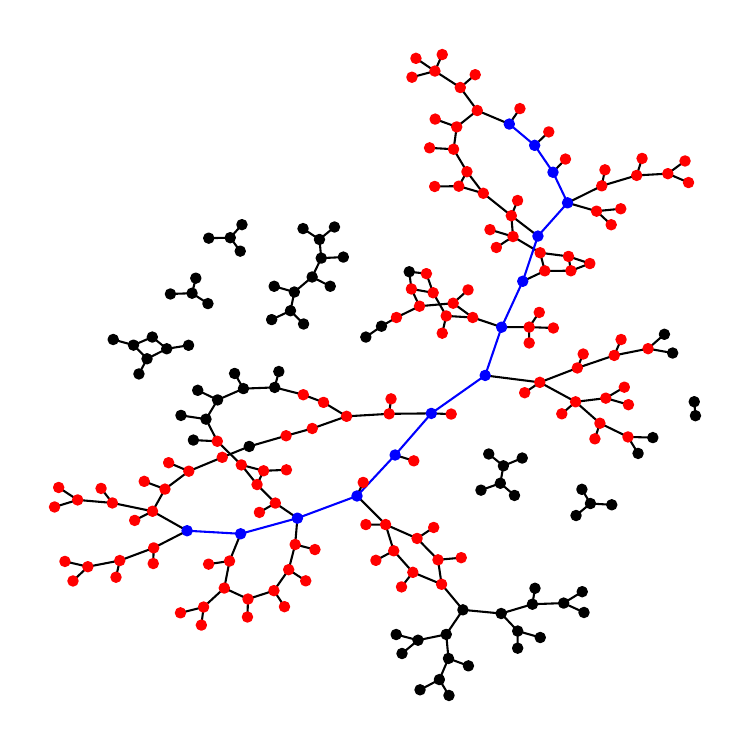}}
        \caption{$k=4$\\$C_4(P^\star) = 128$, $|P^\star| = 14$}
        \label{fig:irvine2-4}
    \end{subfigure}%
    ~ 
    \begin{subfigure}[c]{0.32\textwidth}
        \resizebox{\textwidth}{!}{\includegraphics{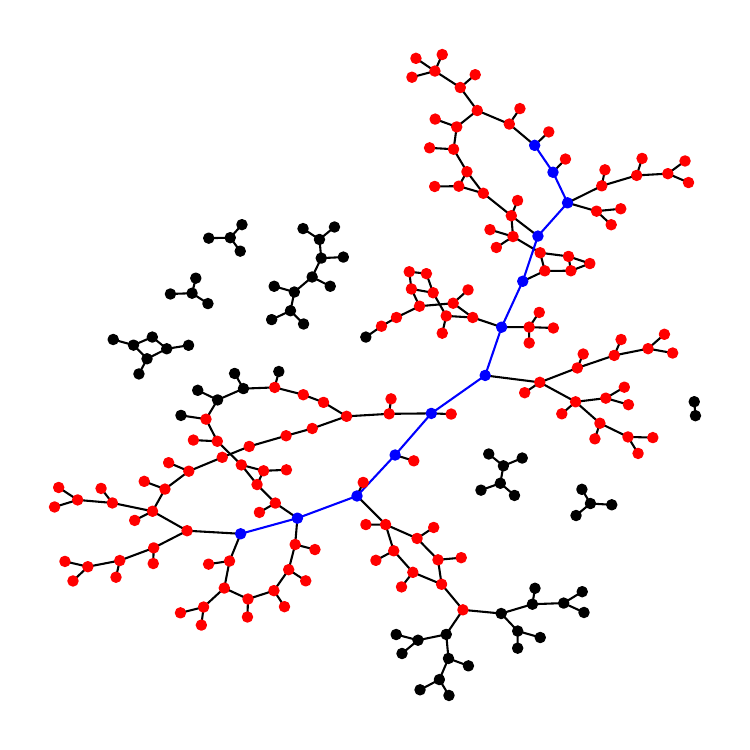}}
        \caption{$k=5$\\$C_5(P^\star) = 141$, $|P^\star| = 12$}
        \label{fig:irvine2-5}
    \end{subfigure}%
    ~ 
    \begin{subfigure}[c]{0.32\textwidth}
        \resizebox{\textwidth}{!}{\includegraphics{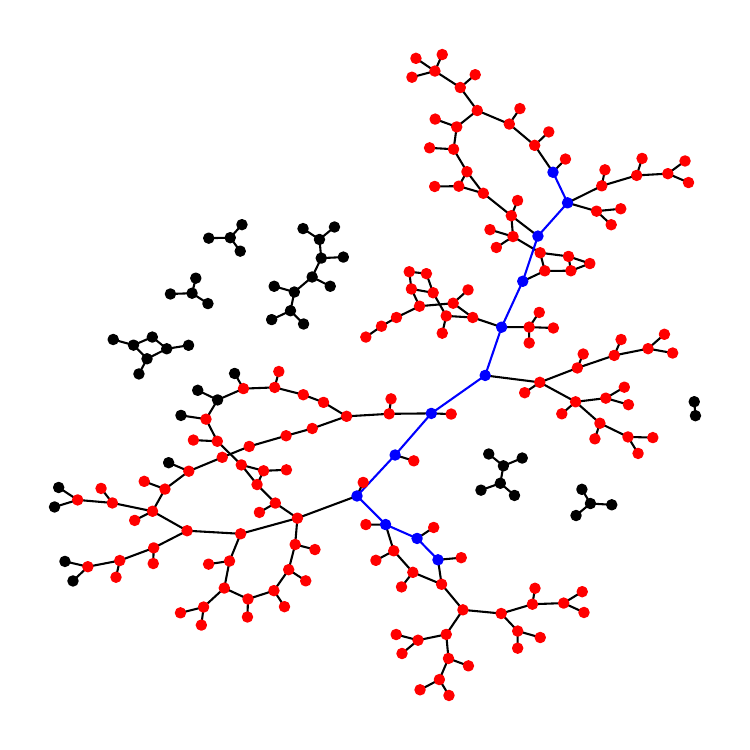}}
        \caption{$k=6$\\$C_6(P^\star) = 157$, $|P^\star| = 12$}
        \label{fig:irvine2-6}
    \end{subfigure}%
    \vspace{1em}
    \begin{subfigure}[c]{0.32\textwidth}
        \resizebox{\textwidth}{!}{\includegraphics{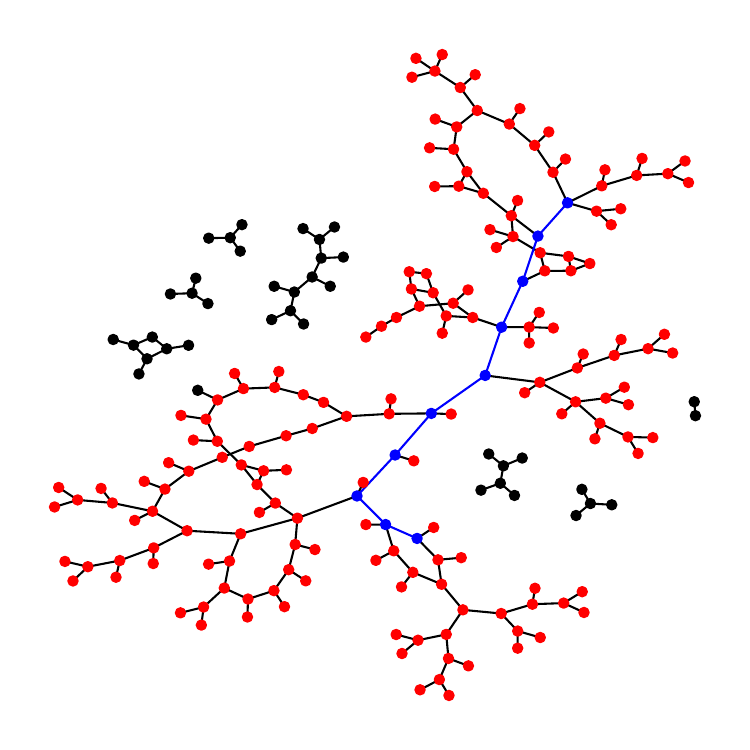}}
        \caption{$k=7$\\$C_7(P^\star) = 167$, $|P^\star| = 10$}
        \label{fig:irvine2-7}
    \end{subfigure}%
    ~ 
    \begin{subfigure}[c]{0.32\textwidth}
        \resizebox{\textwidth}{!}{\includegraphics{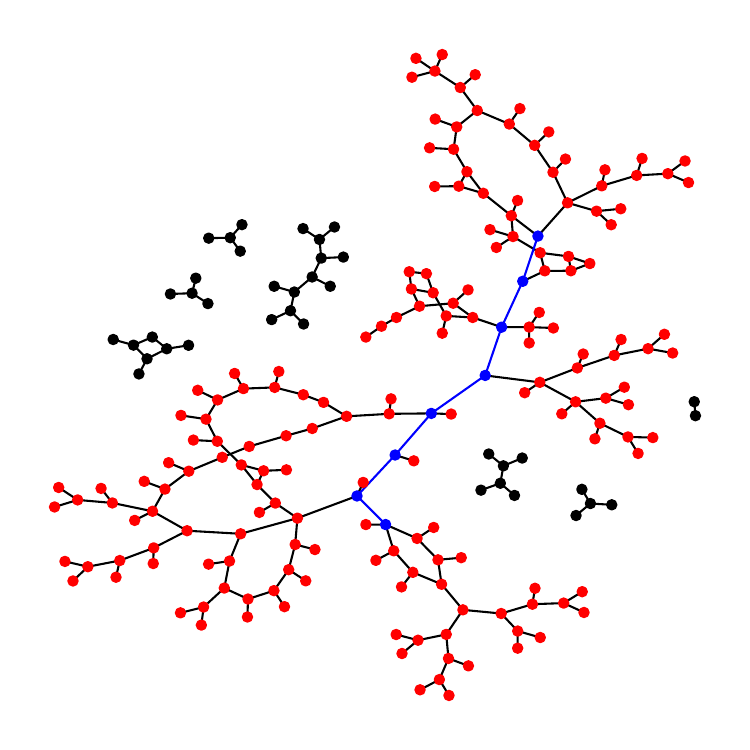}}
        \caption{$k=8$\\$C_8(P^\star) = 170$, $|P^\star| = 8$}
        \label{fig:irvine2-8}
    \end{subfigure}%
    ~ 
    \begin{subfigure}[c]{0.32\textwidth}
        \resizebox{\textwidth}{!}{\includegraphics{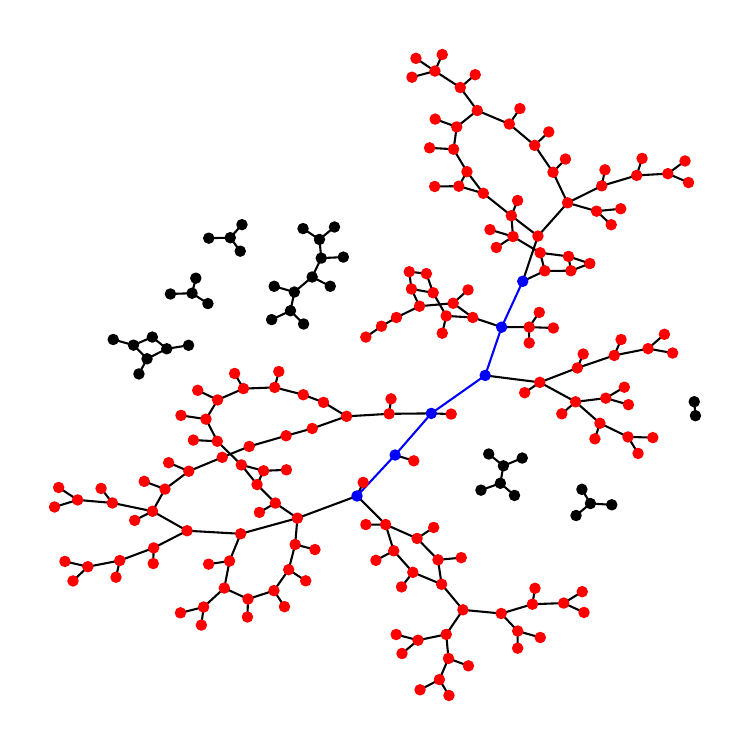}}
        \caption{$k=9$\\$C_9(P^\star) = 172$, $|P^\star| = 6$}
        \label{fig:irvine2-9}
    \end{subfigure}
    \caption{$k$-step-central shortest paths for Irvine2 (diameter: 23).}
    \label{fig:irvine2}
\end{figure*}

\begin{figure*}[!htb]
    \centering
    \captionsetup{justification=centering}
    \begin{subfigure}[c]{0.32\textwidth}
        \resizebox{\textwidth}{!}{\includegraphics{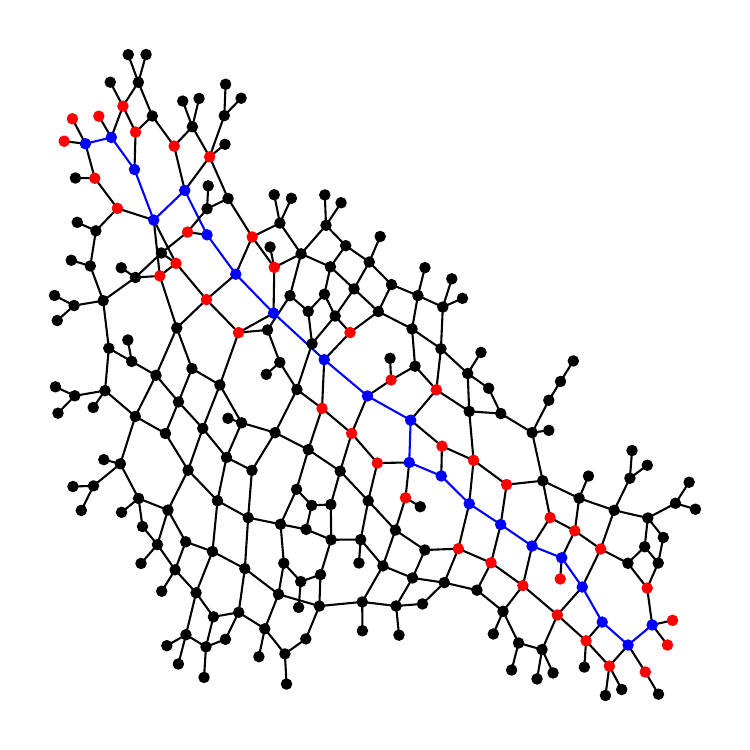}}
        \caption{$k=1$\\$C_1(P^\star) = 40$, $|P^\star| = 21$}
        \label{fig:los-angeles-1}
    \end{subfigure}%
    ~ 
    \begin{subfigure}[c]{0.32\textwidth}
        \resizebox{\textwidth}{!}{\includegraphics{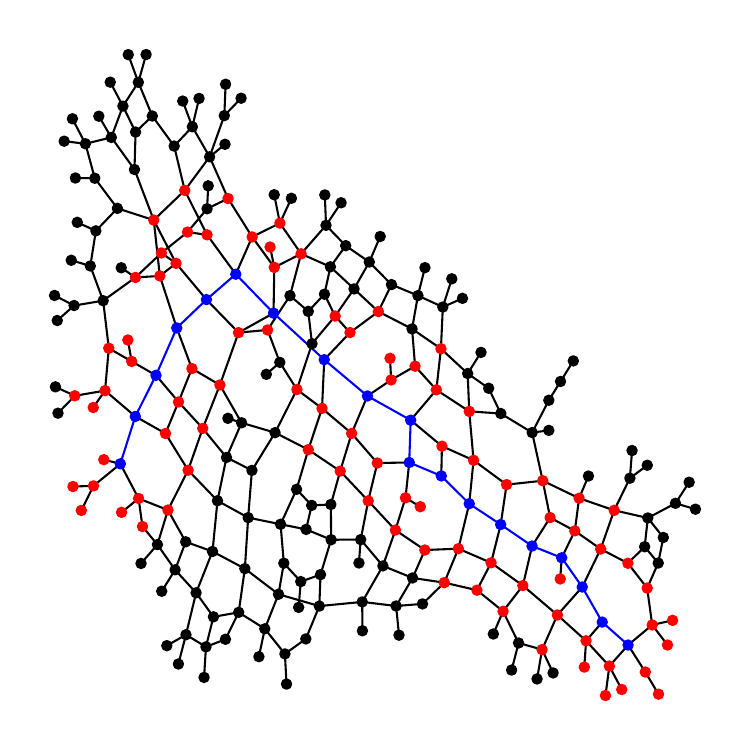}}
        \caption{$k=2$\\$C_2(P^\star) = 86$, $|P^\star| = 19$}
        \label{fig:los-angeles-2}
    \end{subfigure}%
    ~ 
    \begin{subfigure}[c]{0.32\textwidth}
        \resizebox{\textwidth}{!}{\includegraphics{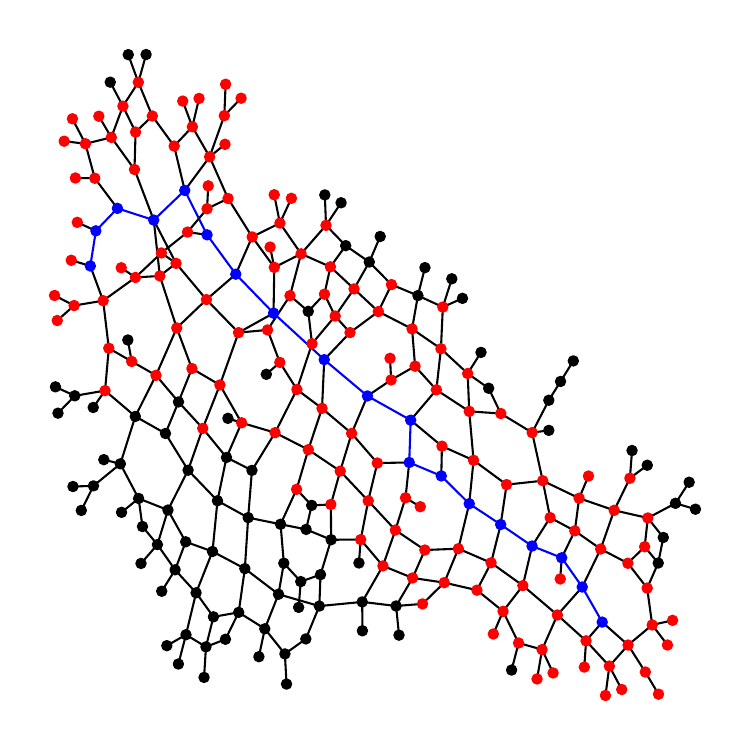}}
        \caption{$k=3$\\$C_3(P^\star) = 134$, $|P^\star| = 19$}
        \label{fig:los-angeles-3}
    \end{subfigure}%
    \vspace{1em}
    \begin{subfigure}[c]{0.32\textwidth}
        \resizebox{\textwidth}{!}{\includegraphics{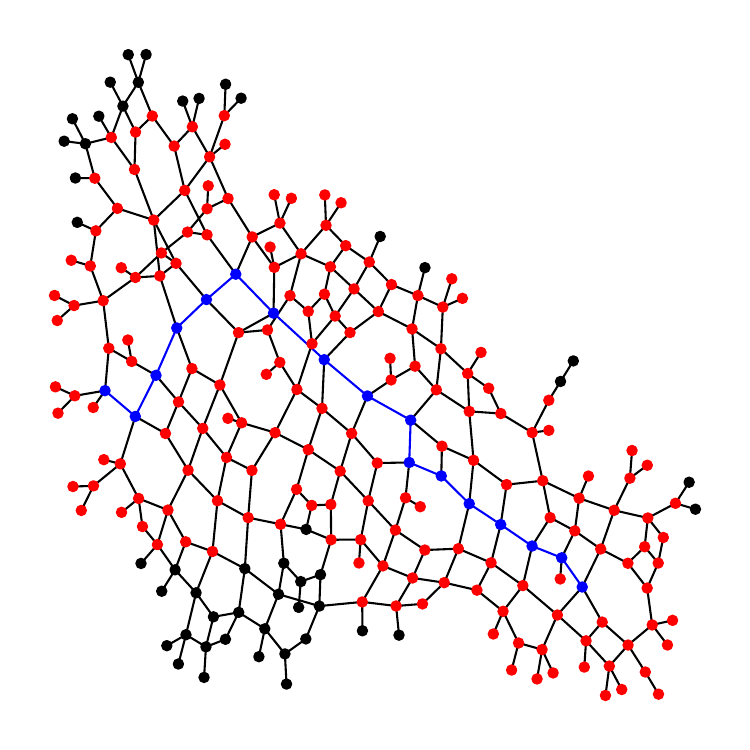}}
        \caption{$k=4$\\$C_4(P^\star) = 175$, $|P^\star| = 17$}
        \label{fig:los-angeles-4}
    \end{subfigure}%
    ~ 
    \begin{subfigure}[c]{0.32\textwidth}
        \resizebox{\textwidth}{!}{\includegraphics{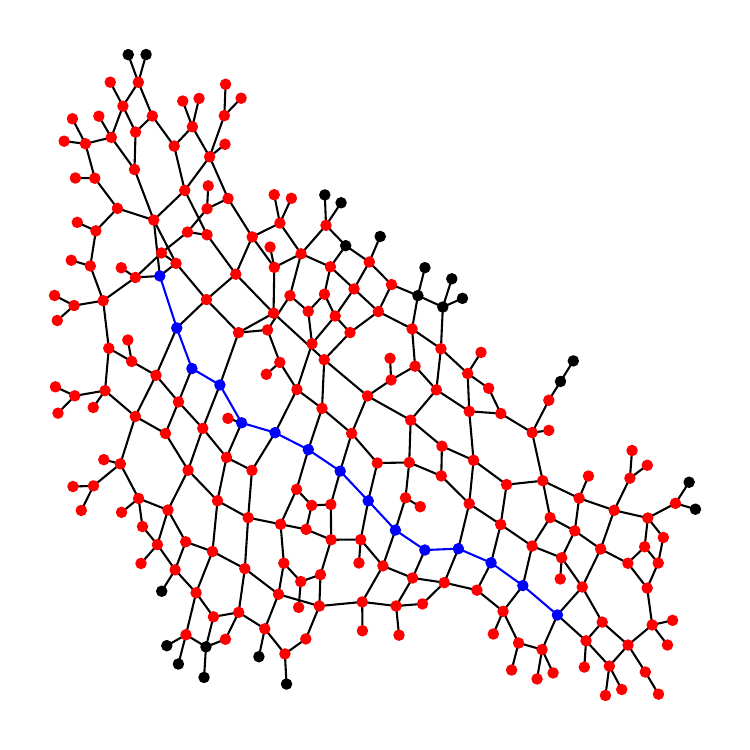}}
        \caption{$k=5$\\$C_5(P^\star) = 203$, $|P^\star| = 15$}
        \label{fig:los-angeles-5}
    \end{subfigure}%
    ~ 
    \begin{subfigure}[c]{0.32\textwidth}
        \resizebox{\textwidth}{!}{\includegraphics{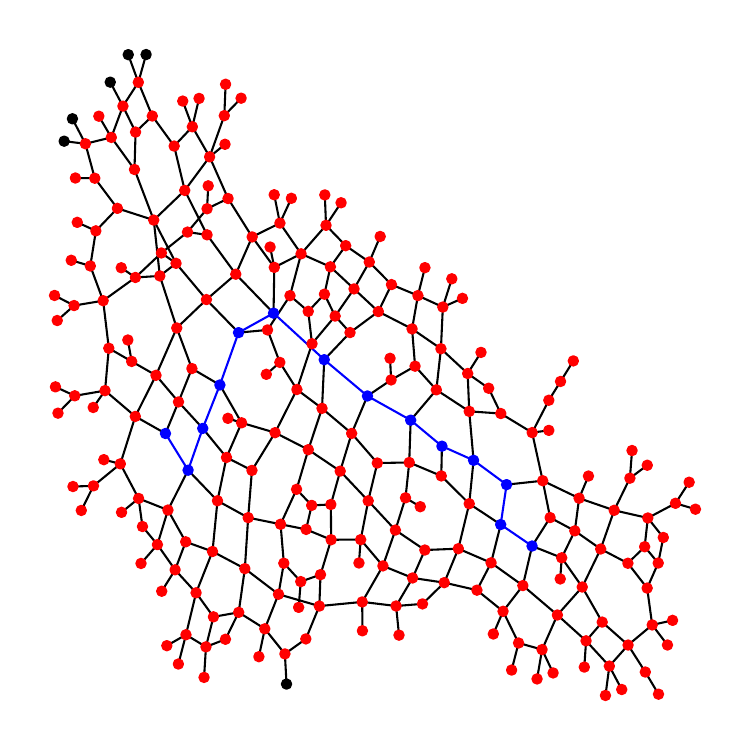}}
        \caption{$k=6$\\$C_6(P^\star) = 220$, $|P^\star| = 14$}
        \label{fig:los-angeles-6}
    \end{subfigure}%
    \vspace{1em}
    \begin{subfigure}[c]{0.32\textwidth}
        \resizebox{\textwidth}{!}{\includegraphics{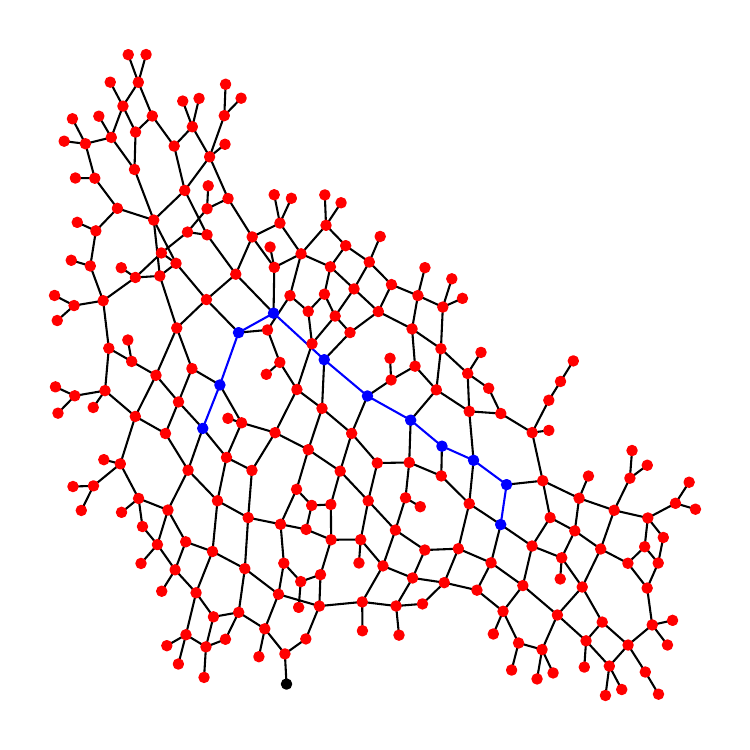}}
        \caption{$k=7$\\$C_7(P^\star) = 228$, $|P^\star| = 11$}
        \label{fig:los-angeles-7}
    \end{subfigure}%
    ~ 
    \begin{subfigure}[c]{0.32\textwidth}
        \resizebox{\textwidth}{!}{\includegraphics{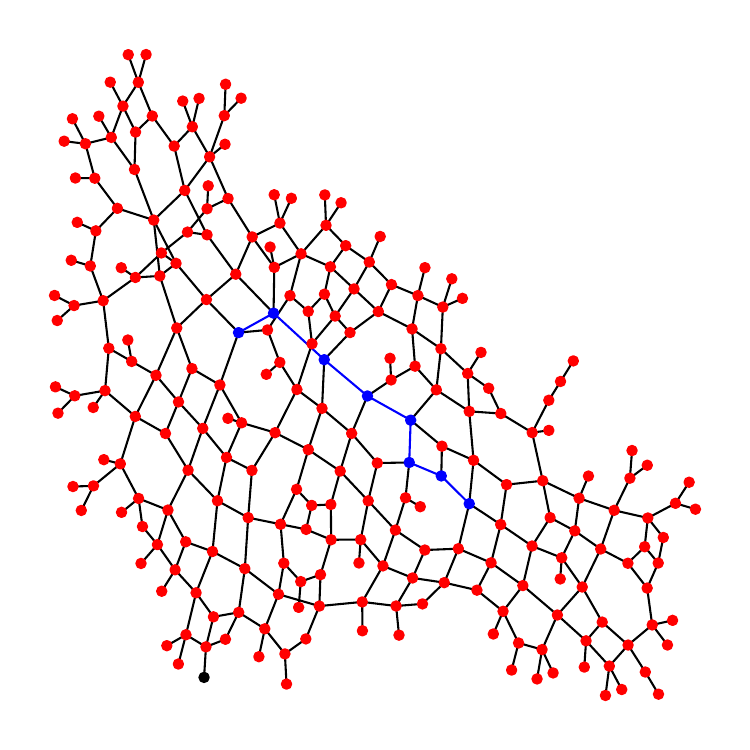}}
        \caption{$k=8$\\$C_8(P^\star) = 231$, $|P^\star| = 8$}
        \label{fig:los-angeles-8}
    \end{subfigure}%
    ~ 
    \begin{subfigure}[c]{0.32\textwidth}
        \resizebox{\textwidth}{!}{\includegraphics{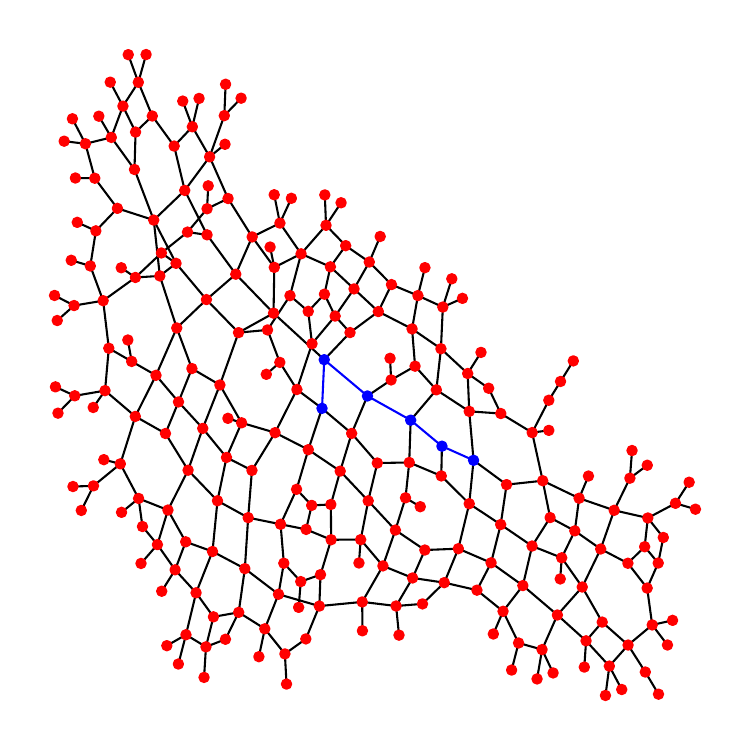}}
        \caption{$k=9$\\$C_9(P^\star) = 234$, $|P^\star| = 6$}
        \label{fig:los-angeles-9}
    \end{subfigure}
    \caption{$k$-step-central shortest paths for Los Angeles (diameter: 22).}
    \label{fig:los-angeles}
\end{figure*}

\begin{figure*}[!htb]
    \centering
    \captionsetup{justification=centering}
    \begin{subfigure}[c]{0.32\textwidth}
        \resizebox{\textwidth}{!}{\includegraphics{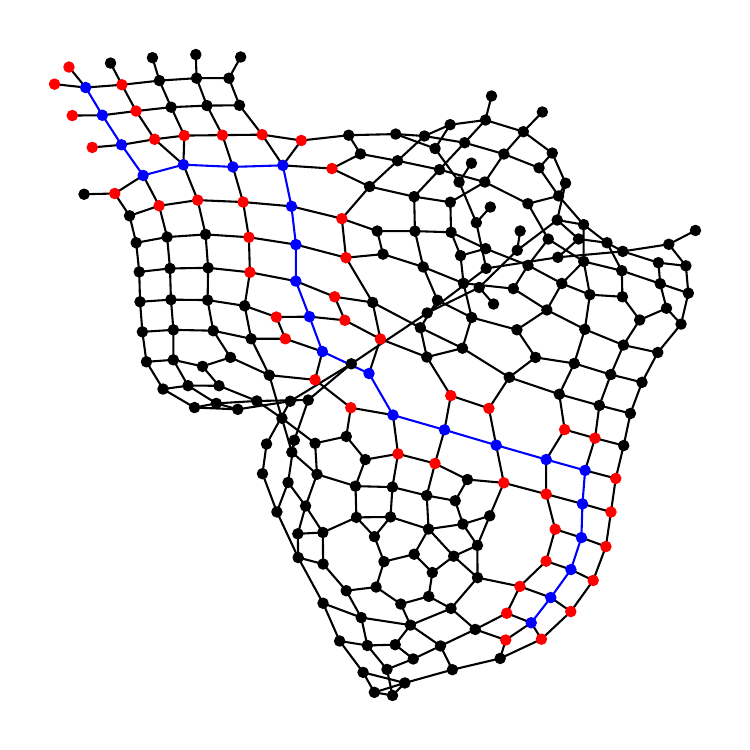}}
        \caption{$k=1$\\$C_1(P^\star) = 46$, $|P^\star| = 23$}
        \label{fig:new-york-1}
    \end{subfigure}%
    ~ 
    \begin{subfigure}[c]{0.32\textwidth}
        \resizebox{\textwidth}{!}{\includegraphics{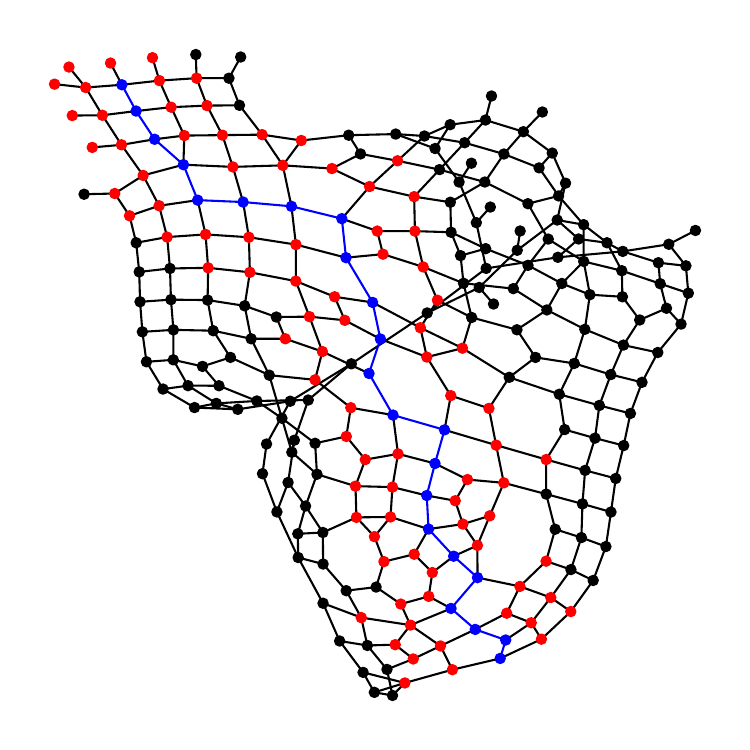}}
        \caption{$k=2$\\$C_2(P^\star) = 86$, $|P^\star| = 23$}
        \label{fig:new-york-2}
    \end{subfigure}%
    ~ 
    \begin{subfigure}[c]{0.32\textwidth}
        \resizebox{\textwidth}{!}{\includegraphics{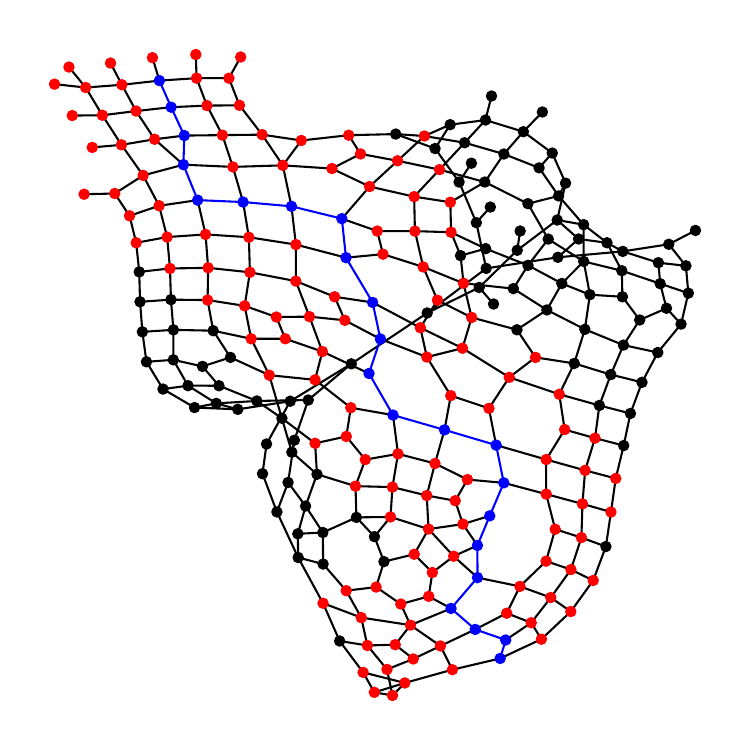}}
        \caption{$k=3$\\$C_3(P^\star) = 126$, $|P^\star| = 23$}
        \label{fig:new-york-3}
    \end{subfigure}%
    \vspace{1em}
    \begin{subfigure}[c]{0.32\textwidth}
        \resizebox{\textwidth}{!}{\includegraphics{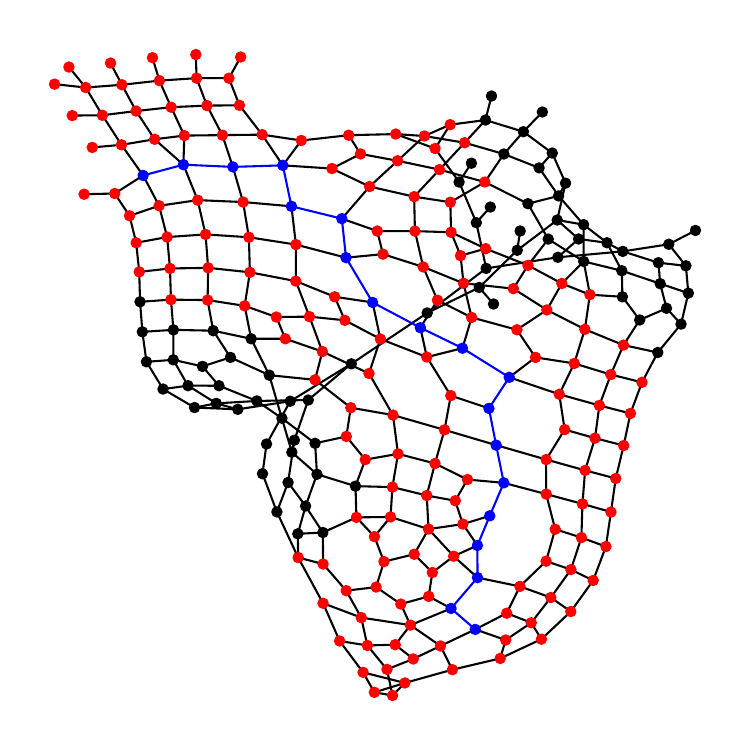}}
        \caption{$k=4$\\$C_4(P^\star) = 156$, $|P^\star| = 19$}
        \label{fig:new-york-4}
    \end{subfigure}%
    ~ 
    \begin{subfigure}[c]{0.32\textwidth}
        \resizebox{\textwidth}{!}{\includegraphics{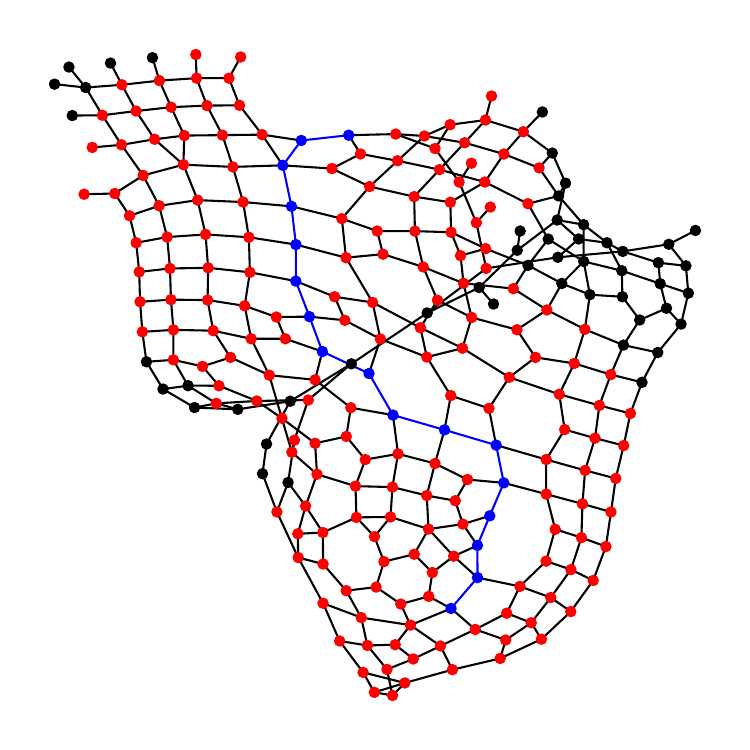}}
        \caption{$k=5$\\$C_5(P^\star) = 181$, $|P^\star| = 17$}
        \label{fig:new-york-5}
    \end{subfigure}%
    ~ 
    \begin{subfigure}[c]{0.32\textwidth}
        \resizebox{\textwidth}{!}{\includegraphics{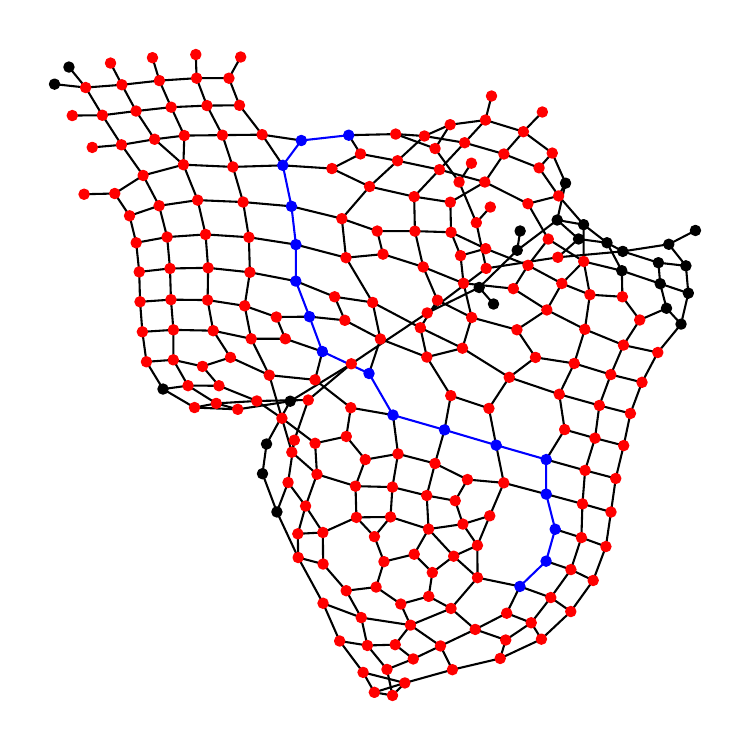}}
        \caption{$k=6$\\$C_6(P^\star) = 205$, $|P^\star| = 17$}
        \label{fig:new-york-6}
    \end{subfigure}%
    \vspace{1em}
    \begin{subfigure}[c]{0.32\textwidth}
        \resizebox{\textwidth}{!}{\includegraphics{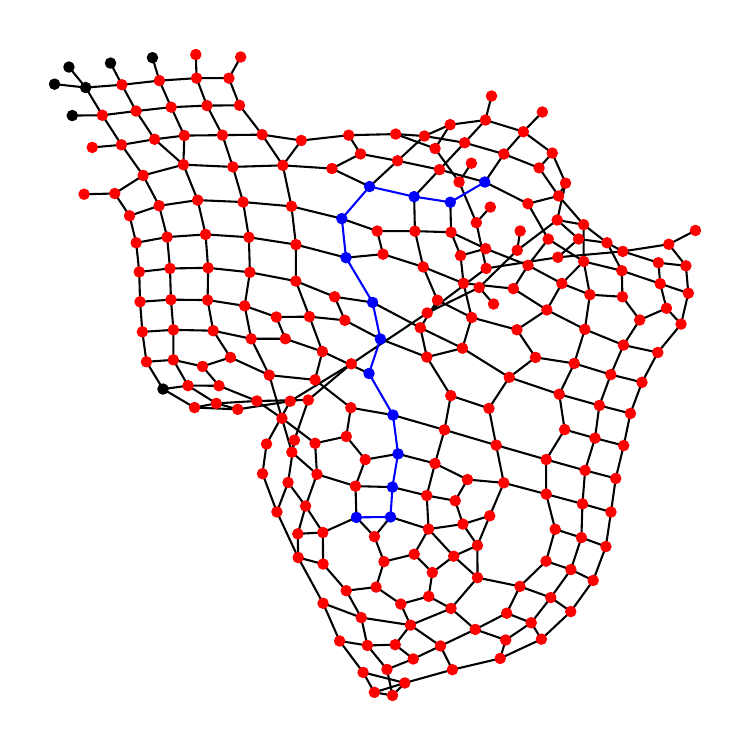}}
        \caption{$k=7$\\$C_7(P^\star) = 227$, $|P^\star| = 14$}
        \label{fig:new-york-7}
    \end{subfigure}%
    ~ 
    \begin{subfigure}[c]{0.32\textwidth}
        \resizebox{\textwidth}{!}{\includegraphics{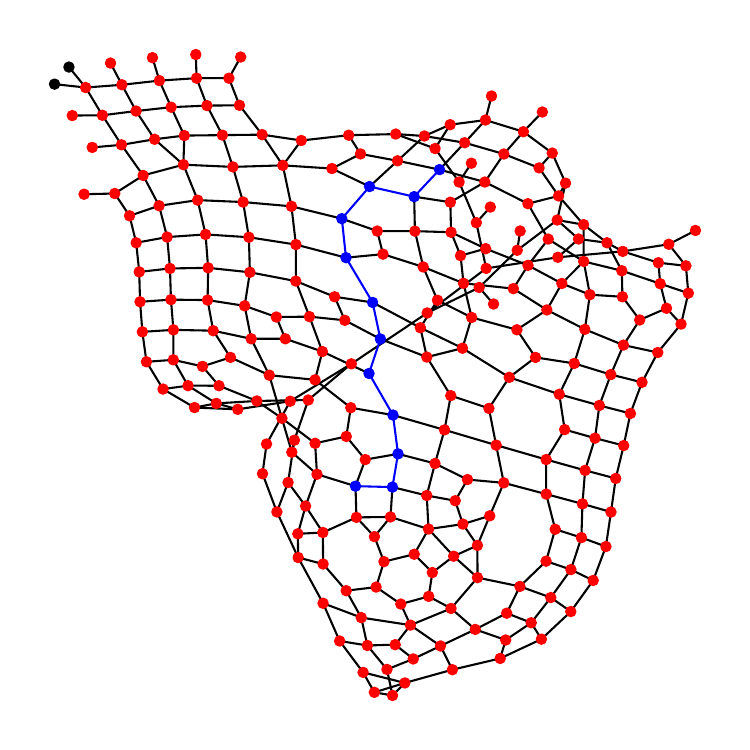}}
        \caption{$k=8$\\$C_8(P^\star) = 234$, $|P^\star| = 12$}
        \label{fig:new-york-8}
    \end{subfigure}%
    ~ 
    \begin{subfigure}[c]{0.32\textwidth}
        \resizebox{\textwidth}{!}{\includegraphics{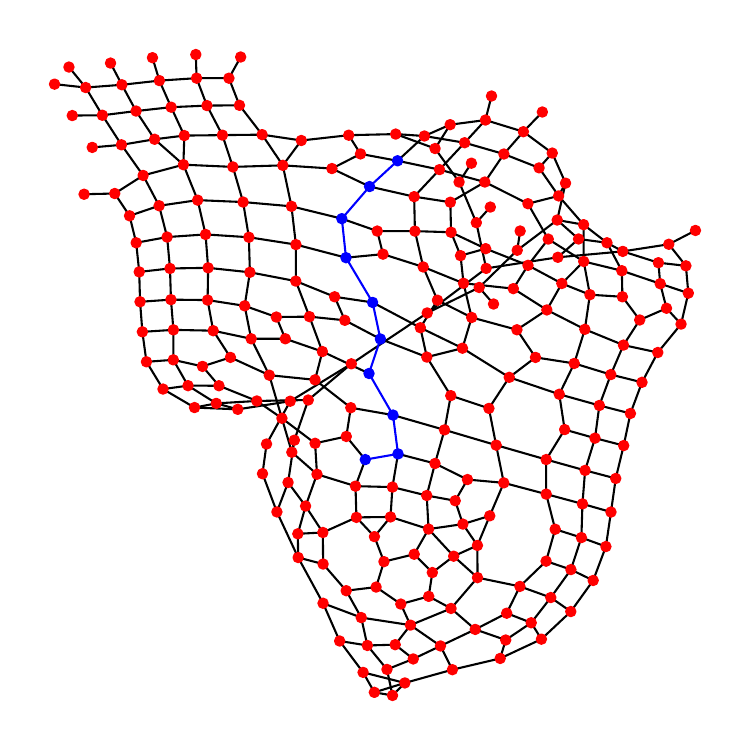}}
        \caption{$k=9$\\$C_9(P^\star) = 238$, $|P^\star| = 10$}
        \label{fig:new-york-9}
    \end{subfigure}
    \caption{$k$-step-central shortest paths for New York (diameter: 24).}
    \label{fig:new-york}
\end{figure*}

\begin{figure*}[!htb]
    \centering
    \captionsetup{justification=centering}
    \begin{subfigure}[c]{0.32\textwidth}
        \resizebox{\textwidth}{!}{\includegraphics{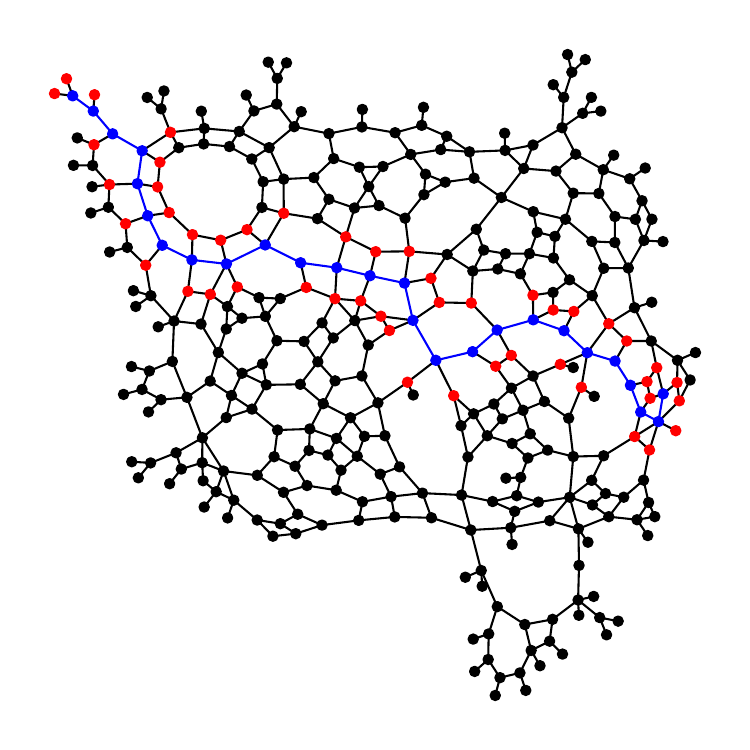}}
        \caption{$k=1$\\$C_1(P^\star) = 48$, $|P^\star| = 26$}
        \label{fig:paris-1}
    \end{subfigure}%
    ~ 
    \begin{subfigure}[c]{0.32\textwidth}
        \resizebox{\textwidth}{!}{\includegraphics{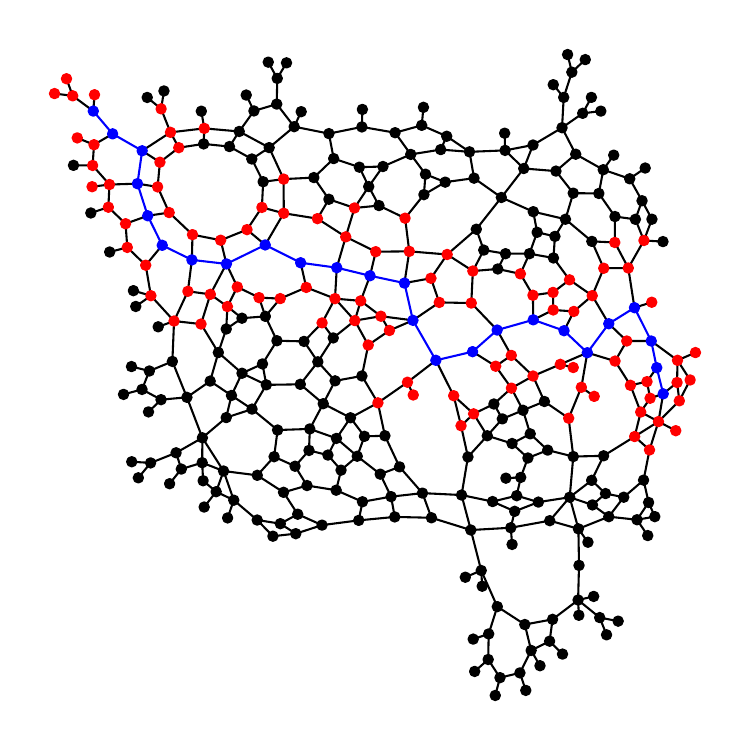}}
        \caption{$k=2$\\$C_2(P^\star) = 96$, $|P^\star| = 25$}
        \label{fig:paris-2}
    \end{subfigure}%
    ~ 
    \begin{subfigure}[c]{0.32\textwidth}
        \resizebox{\textwidth}{!}{\includegraphics{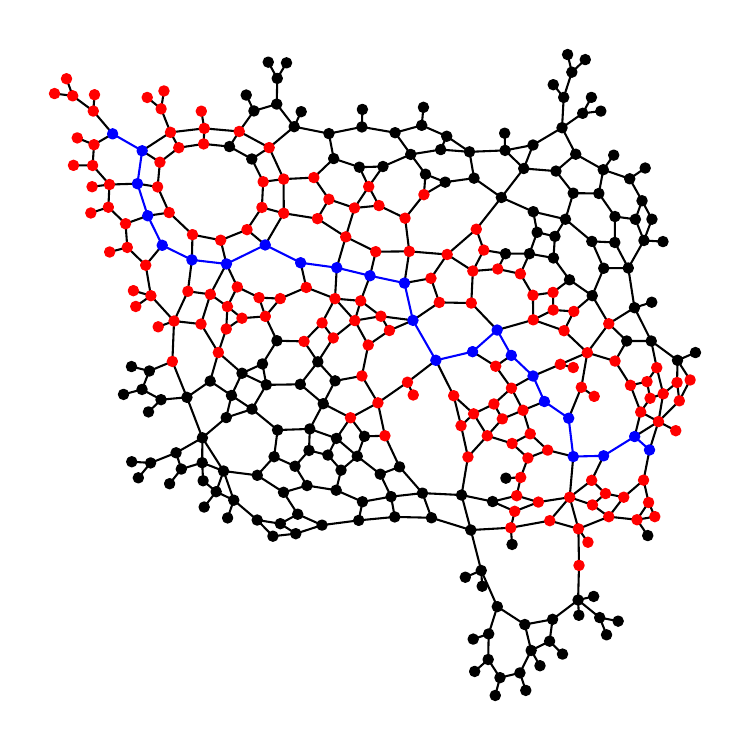}}
        \caption{$k=3$\\$C_3(P^\star) = 147$, $|P^\star| = 24$}
        \label{fig:paris-3}
    \end{subfigure}%
    \vspace{1em}
    \begin{subfigure}[c]{0.32\textwidth}
        \resizebox{\textwidth}{!}{\includegraphics{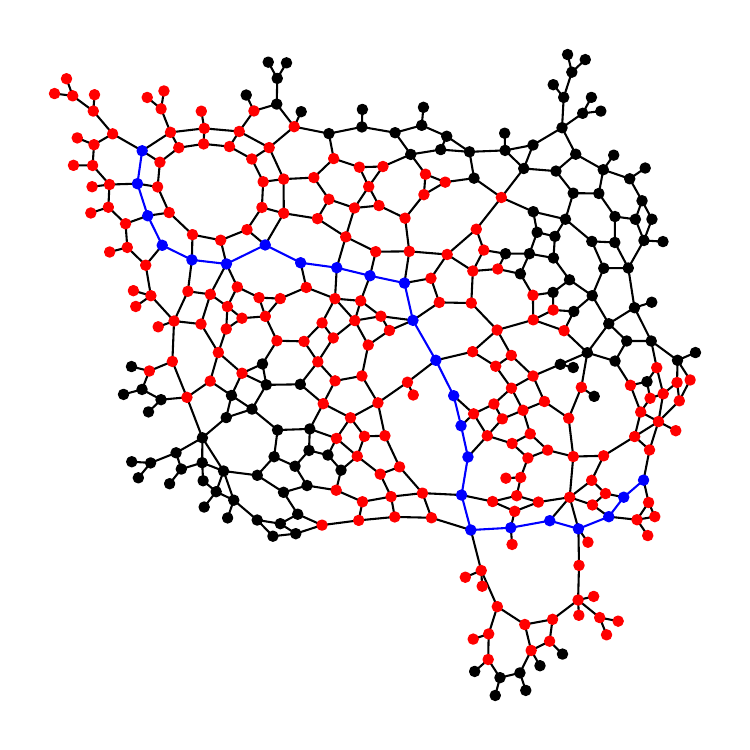}}
        \caption{$k=4$\\$C_4(P^\star) = 191$, $|P^\star| = 24$}
        \label{fig:paris-4}
    \end{subfigure}%
    ~ 
    \begin{subfigure}[c]{0.32\textwidth}
        \resizebox{\textwidth}{!}{\includegraphics{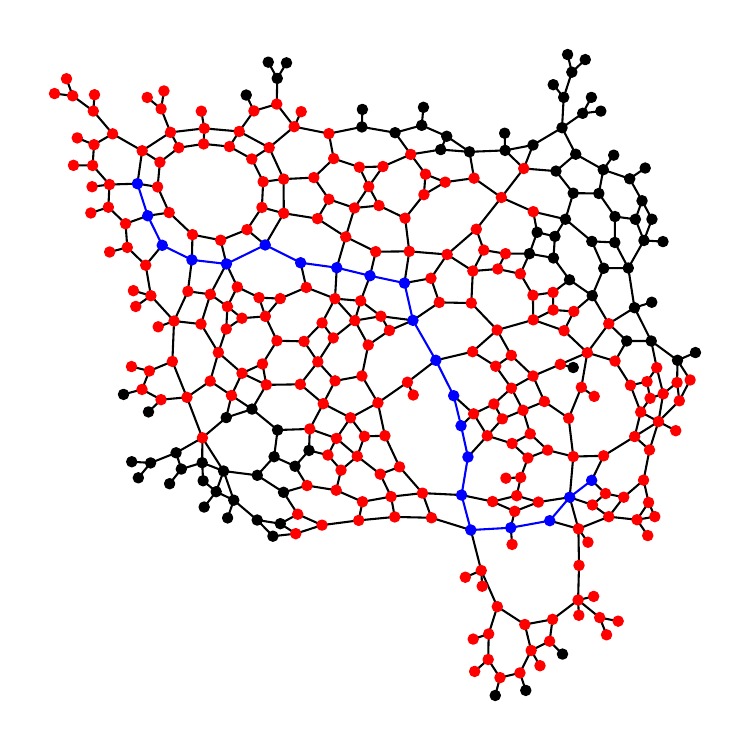}}
        \caption{$k=5$\\$C_5(P^\star) = 229$, $|P^\star| = 21$}
        \label{fig:paris-5}
    \end{subfigure}%
    ~ 
    \begin{subfigure}[c]{0.32\textwidth}
        \resizebox{\textwidth}{!}{\includegraphics{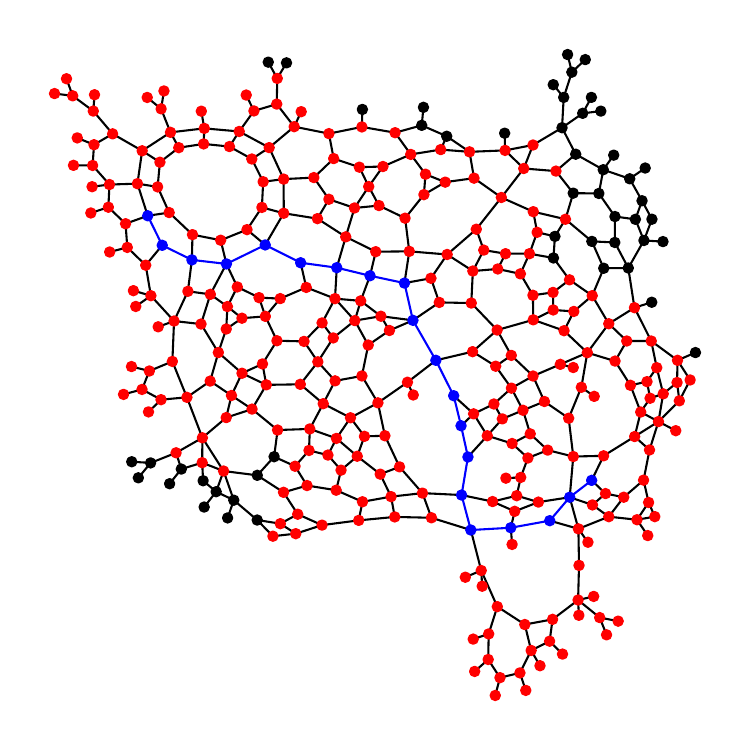}}
        \caption{$k=6$\\$C_6(P^\star) = 265$, $|P^\star| = 20$}
        \label{fig:paris-6}
    \end{subfigure}%
    \vspace{1em}
    \begin{subfigure}[c]{0.32\textwidth}
        \resizebox{\textwidth}{!}{\includegraphics{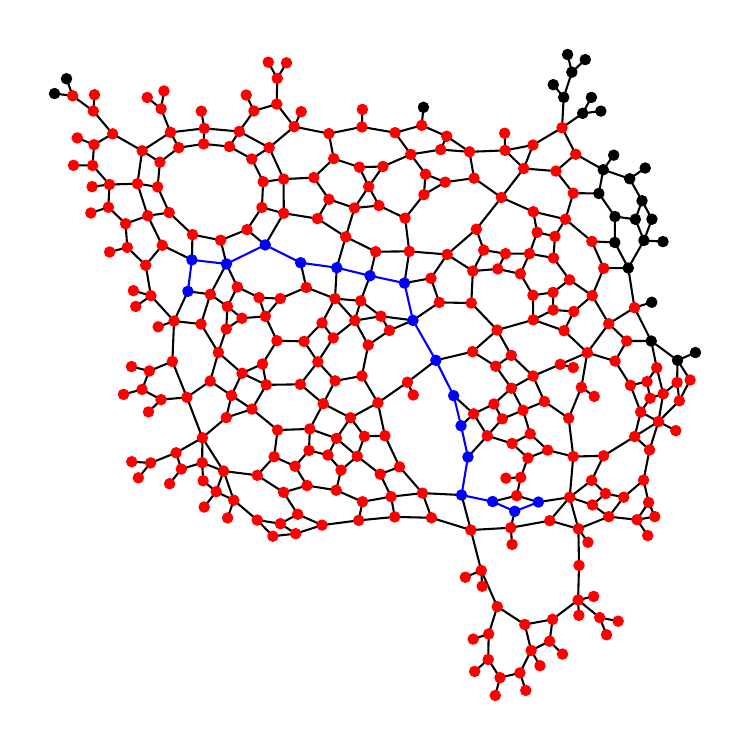}}
        \caption{$k=7$\\$C_7(P^\star) = 290$, $|P^\star| = 17$}
        \label{fig:paris-7}
    \end{subfigure}%
    ~ 
    \begin{subfigure}[c]{0.32\textwidth}
        \resizebox{\textwidth}{!}{\includegraphics{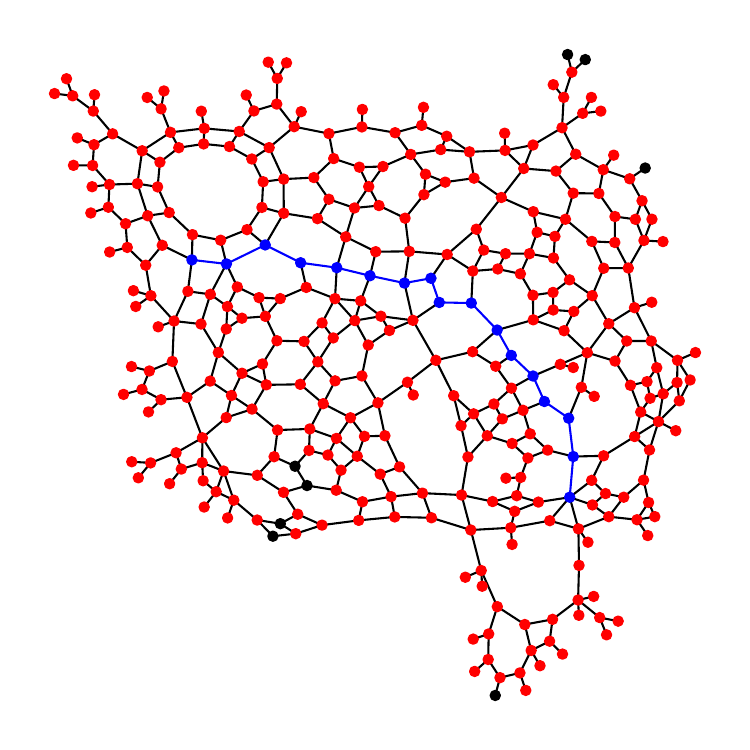}}
        \caption{$k=8$\\$C_8(P^\star) = 310$, $|P^\star| = 17$}
        \label{fig:paris-8}
    \end{subfigure}%
    ~ 
    \begin{subfigure}[c]{0.32\textwidth}
        \resizebox{\textwidth}{!}{\includegraphics{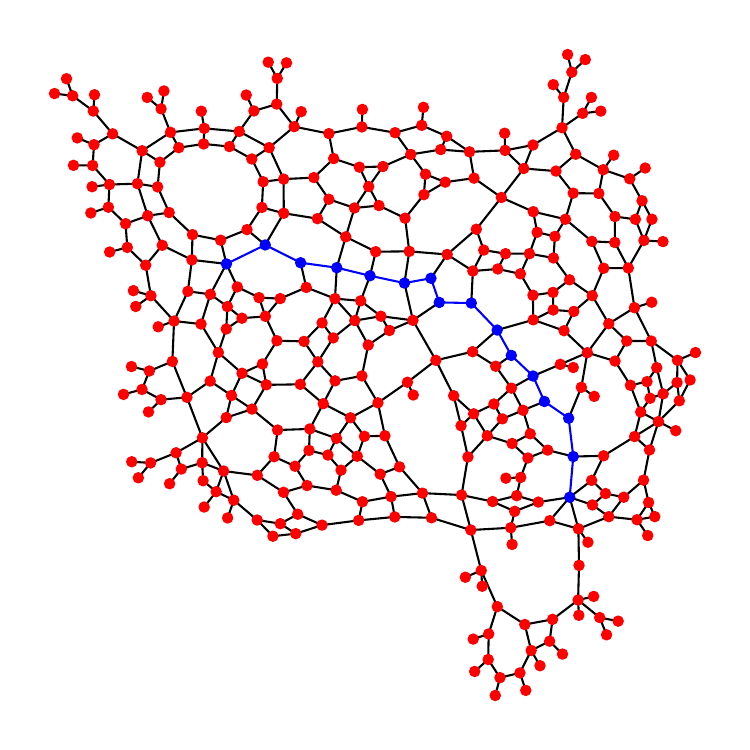}}
        \caption{$k=9$\\$C_9(P^\star) = 319$, $|P^\star| = 16$}
        \label{fig:paris-9}
    \end{subfigure}
    \caption{$k$-step-central shortest paths for Paris (diameter: 28).}
    \label{fig:paris}
\end{figure*}

\begin{figure}[!htb]
    \centering
    \captionsetup{justification=centering}
    \begin{subfigure}[c]{0.32\textwidth}
        \resizebox{\textwidth}{!}{\includegraphics{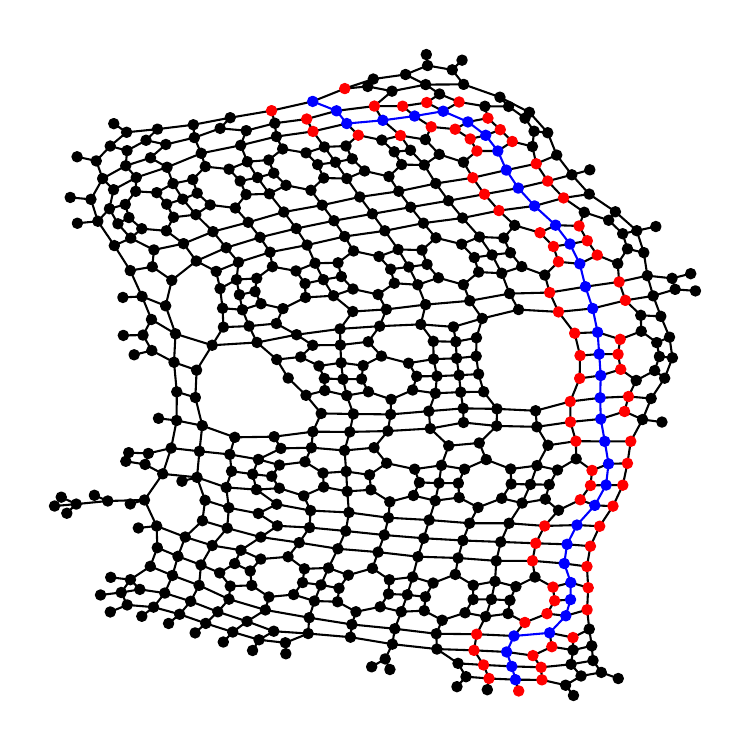}}
        \caption{$k=1$\\$C_1(P^\star) = 73$, $|P^\star| = 37$}
        \label{fig:savannah-1}
    \end{subfigure}%
    ~ 
    \begin{subfigure}[c]{0.32\textwidth}
        \resizebox{\textwidth}{!}{\includegraphics{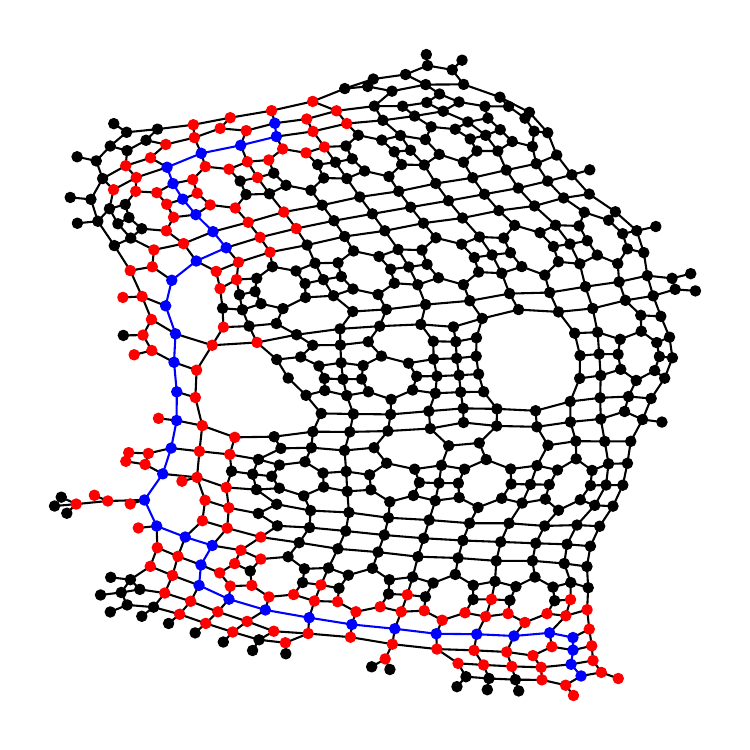}}
        \caption{$k=2$\\$C_2(P^\star) = 139$, $|P^\star| = 38$}
        \label{fig:savannah-2}
    \end{subfigure}%
    ~ 
    \begin{subfigure}[c]{0.32\textwidth}
        \resizebox{\textwidth}{!}{\includegraphics{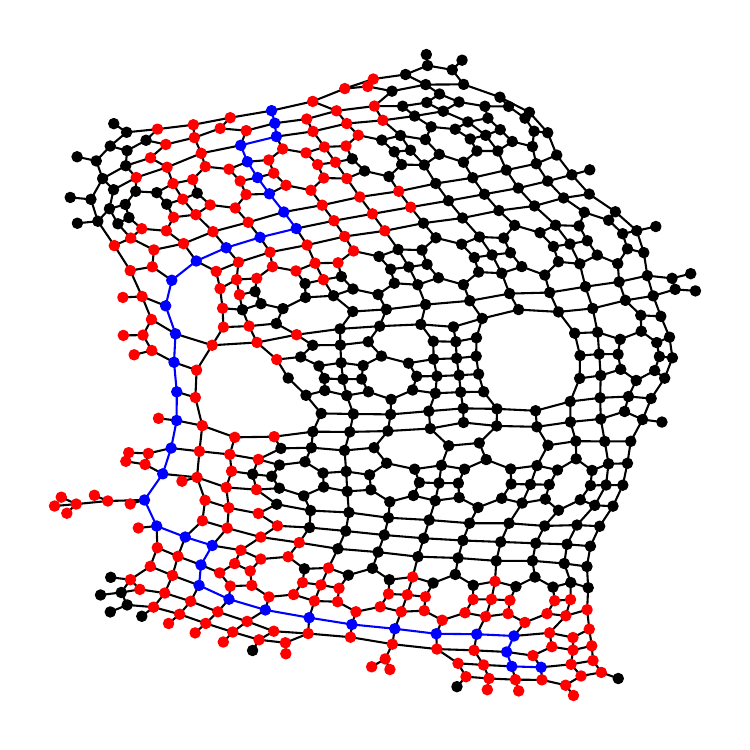}}
        \caption{$k=3$\\$C_3(P^\star) = 213$, $|P^\star| = 37$}
        \label{fig:savannah-3}
    \end{subfigure}%
    \vspace{1em}
    \begin{subfigure}[c]{0.32\textwidth}
        \resizebox{\textwidth}{!}{\includegraphics{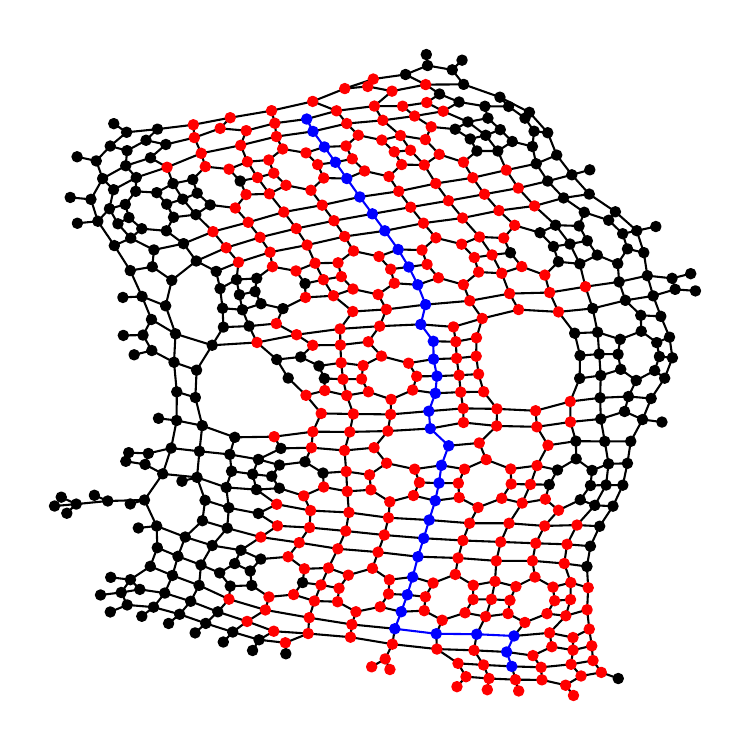}}
        \caption{$k=4$\\$C_4(P^\star) = 292$, $|P^\star| = 35$}
        \label{fig:savannah-4}
    \end{subfigure}%
    ~ 
    \begin{subfigure}[c]{0.32\textwidth}
        \resizebox{\textwidth}{!}{\includegraphics{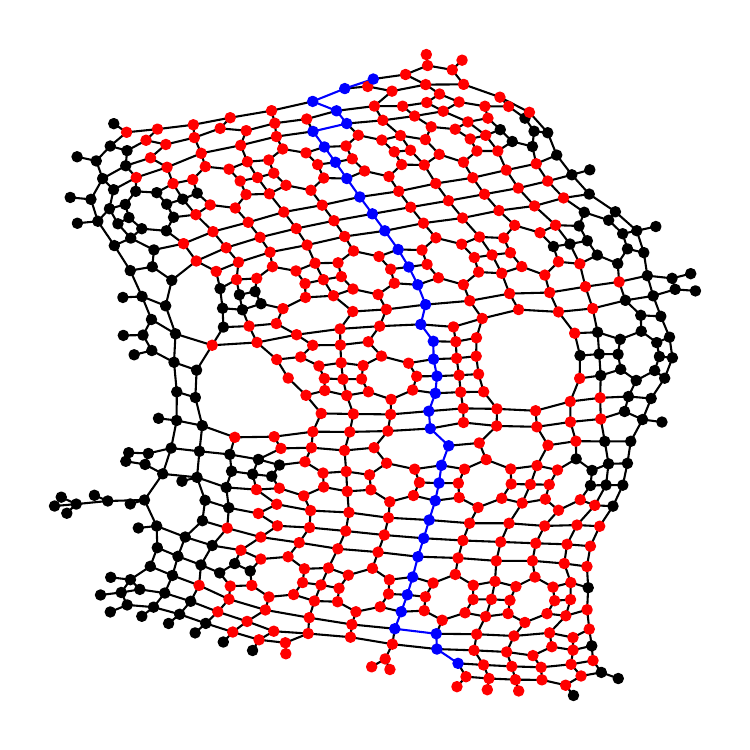}}
        \caption{$k=5$\\$C_5(P^\star) = 370$, $|P^\star| = 37$}
        \label{fig:savannah-5}
    \end{subfigure}%
    ~ 
    \begin{subfigure}[c]{0.32\textwidth}
        \resizebox{\textwidth}{!}{\includegraphics{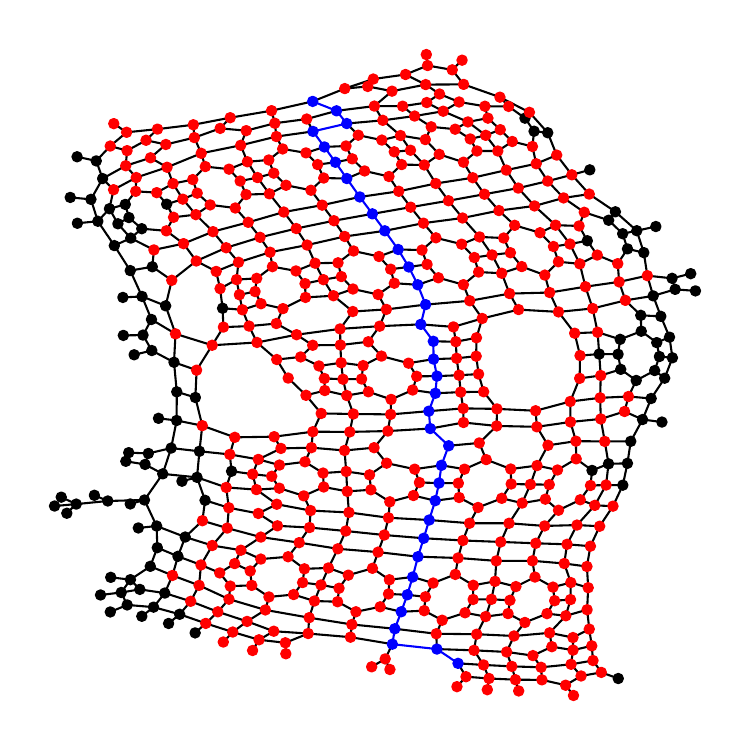}}
        \caption{$k=6$\\$C_6(P^\star) = 440$, $|P^\star| = 35$}
        \label{fig:savannah-6}
    \end{subfigure}%
    \vspace{1em}
    \begin{subfigure}[c]{0.32\textwidth}
        \resizebox{\textwidth}{!}{\includegraphics{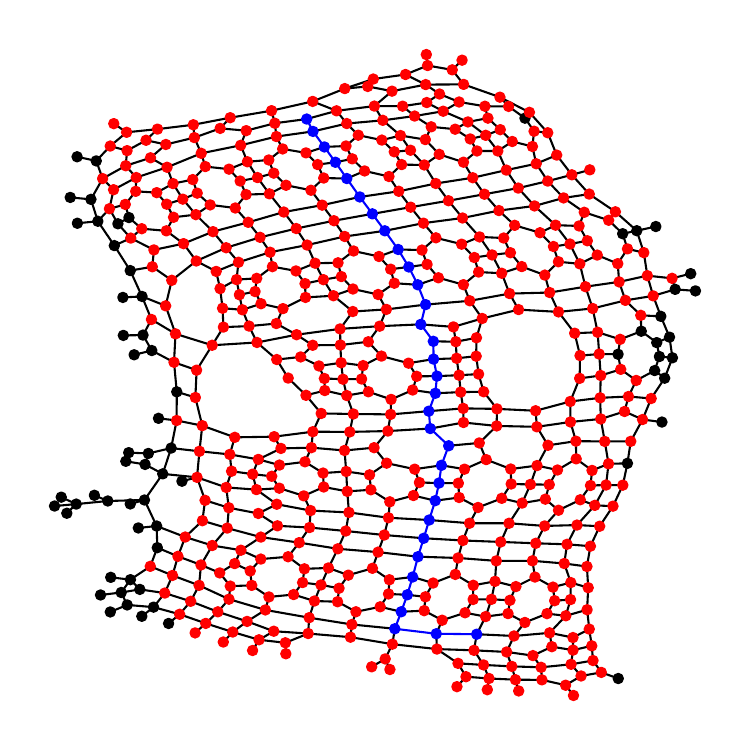}}
        \caption{$k=7$\\$C_7(P^\star) = 487$, $|P^\star| = 32$}
        \label{fig:savannah-7}
    \end{subfigure}%
    ~ 
    \begin{subfigure}[c]{0.32\textwidth}
        \resizebox{\textwidth}{!}{\includegraphics{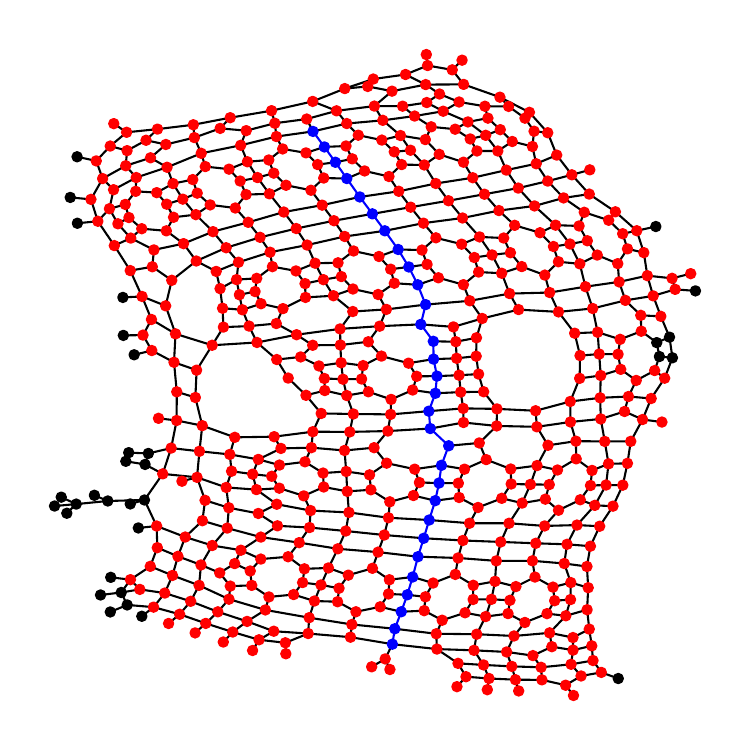}}
        \caption{$k=8$\\$C_8(P^\star) = 522$, $|P^\star| = 30$}
        \label{fig:savannah-8}
    \end{subfigure}%
    ~ 
    \begin{subfigure}[c]{0.32\textwidth}
        \resizebox{\textwidth}{!}{\includegraphics{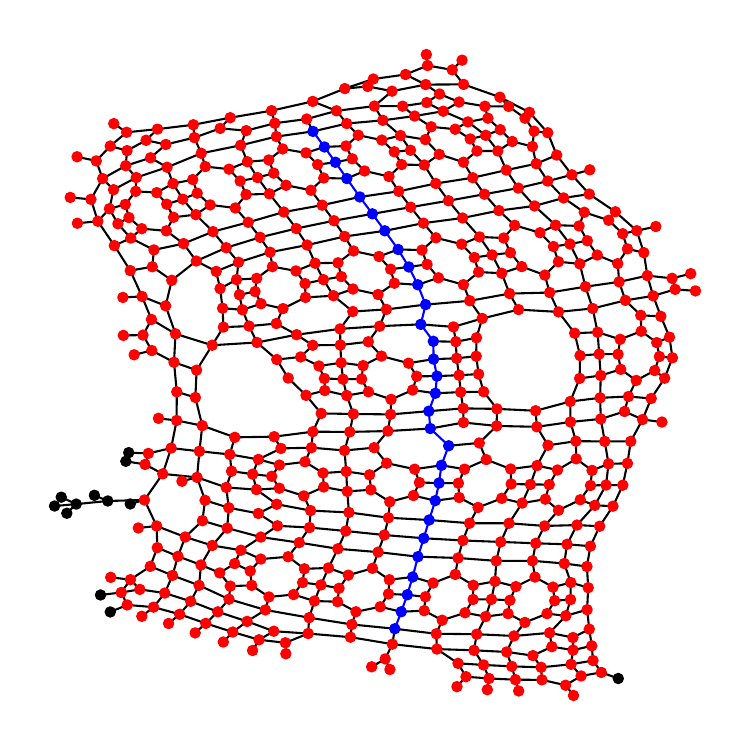}}
        \caption{$k=9$\\$C_9(P^\star) = 543$, $|P^\star| = 29$}
        \label{fig:savannah-9}
    \end{subfigure}
    \caption{$k$-step-central shortest paths for Savannah (diameter: 39).}
    \label{fig:savannah}
\end{figure}

\begin{figure*}[!htb]
    \centering
    \captionsetup{justification=centering}
    \begin{subfigure}[c]{0.32\textwidth}
        \resizebox{\textwidth}{!}{\includegraphics{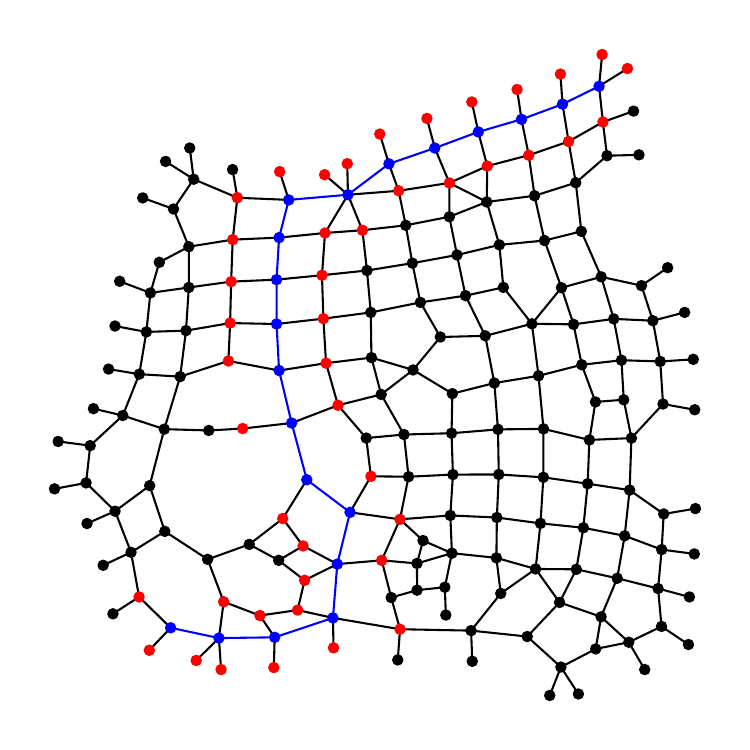}}
        \caption{$k=1$\\$C_1(P^\star) = 44$, $|P^\star| = 20$}
        \label{fig:washington-1}
    \end{subfigure}%
    ~ 
    \begin{subfigure}[c]{0.32\textwidth}
        \resizebox{\textwidth}{!}{\includegraphics{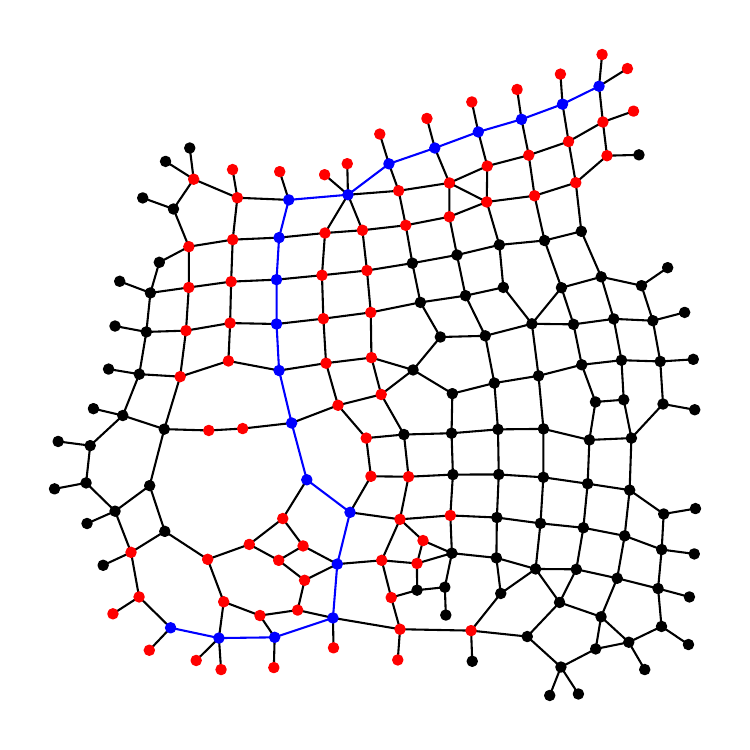}}
        \caption{$k=2$\\$C_2(P^\star) = 75$, $|P^\star| = 20$}
        \label{fig:washington-2}
    \end{subfigure}%
    ~ 
    \begin{subfigure}[c]{0.32\textwidth}
        \resizebox{\textwidth}{!}{\includegraphics{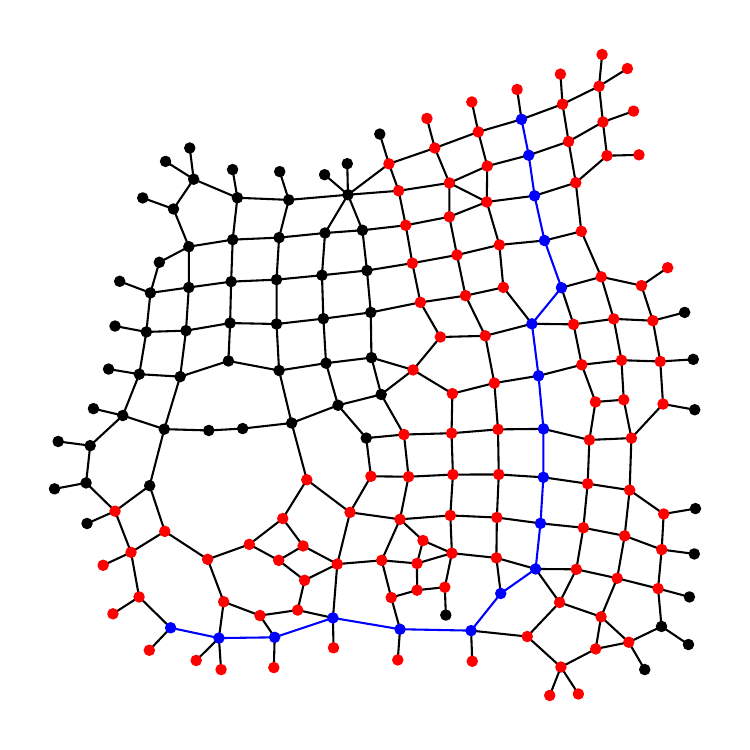}}
        \caption{$k=3$\\$C_3(P^\star) = 109$, $|P^\star| = 18$}
        \label{fig:washington-3}
    \end{subfigure}%
    \vspace{1em}
    \begin{subfigure}[c]{0.32\textwidth}
        \resizebox{\textwidth}{!}{\includegraphics{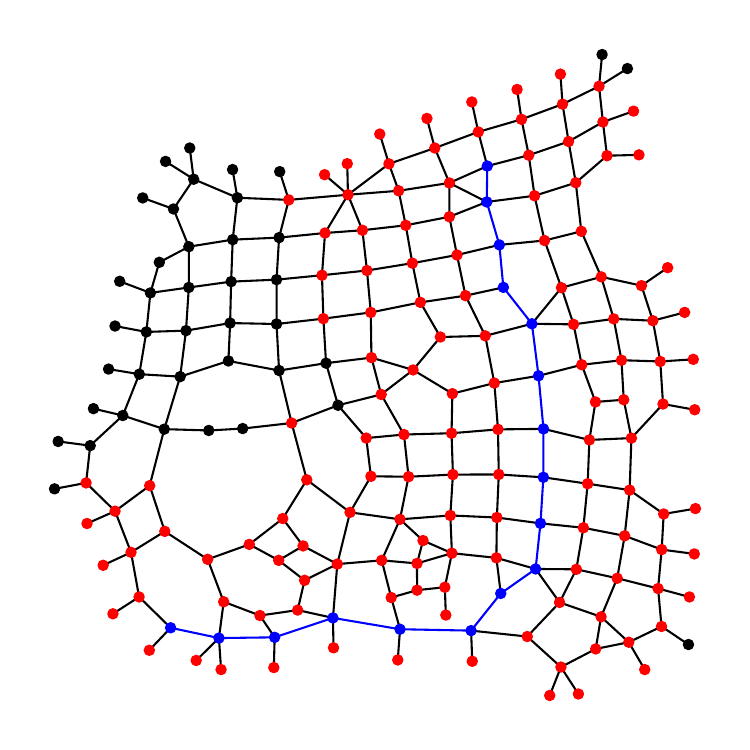}}
        \caption{$k=4$\\$C_4(P^\star) = 135$, $|P^\star| = 17$}
        \label{fig:washington-4}
    \end{subfigure}%
    ~ 
    \begin{subfigure}[c]{0.32\textwidth}
        \resizebox{\textwidth}{!}{\includegraphics{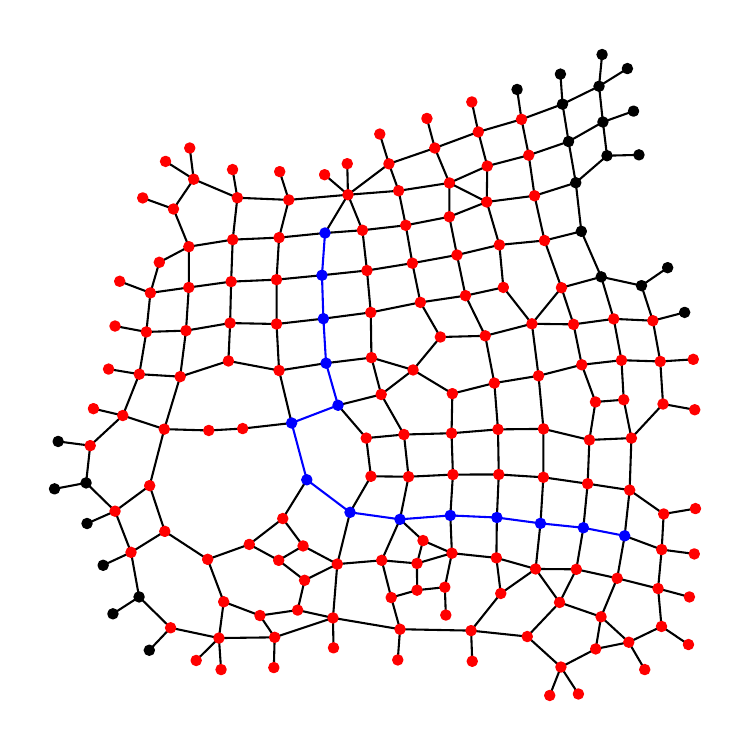}}
        \caption{$k=5$\\$C_5(P^\star) = 153$, $|P^\star| = 14$}
        \label{fig:washington-5}
    \end{subfigure}%
    ~ 
    \begin{subfigure}[c]{0.32\textwidth}
        \resizebox{\textwidth}{!}{\includegraphics{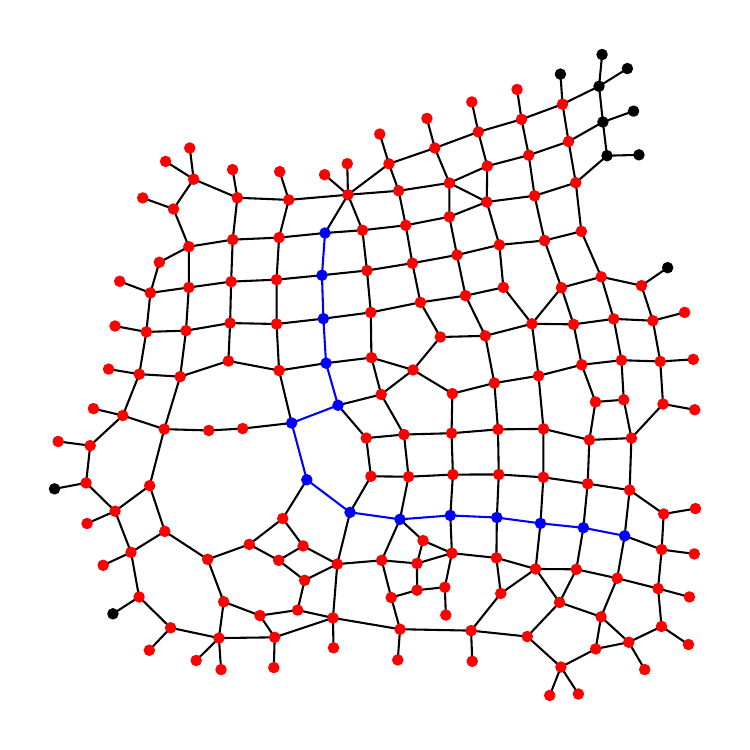}}
        \caption{$k=6$\\$C_6(P^\star) = 167$, $|P^\star| = 14$}
        \label{fig:washington-6}
    \end{subfigure}%
    \vspace{1em}
    \begin{subfigure}[c]{0.32\textwidth}
        \resizebox{\textwidth}{!}{\includegraphics{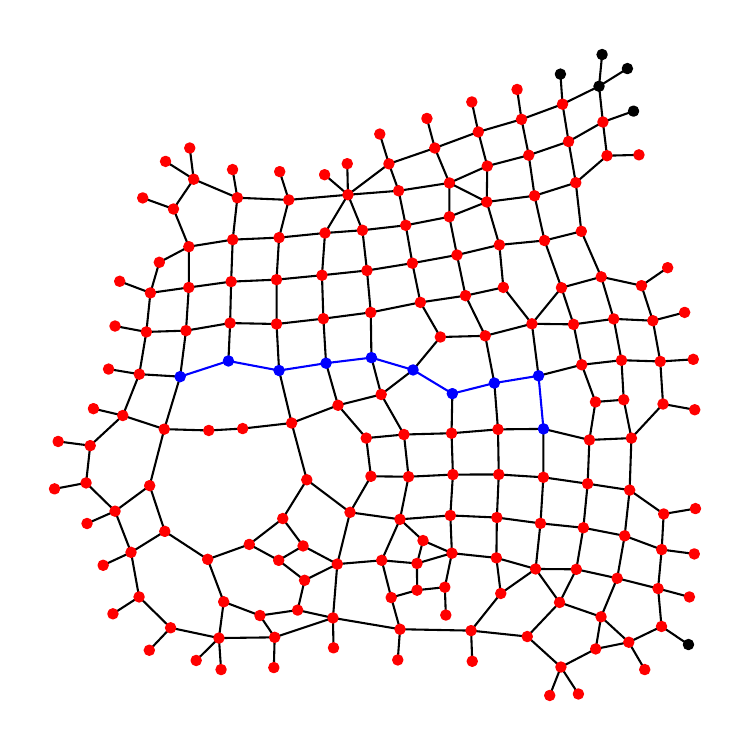}}
        \caption{$k=7$\\$C_7(P^\star) = 176$, $|P^\star| = 10$}
        \label{fig:washington-7}
    \end{subfigure}%
    ~ 
    \begin{subfigure}[c]{0.32\textwidth}
        \resizebox{\textwidth}{!}{\includegraphics{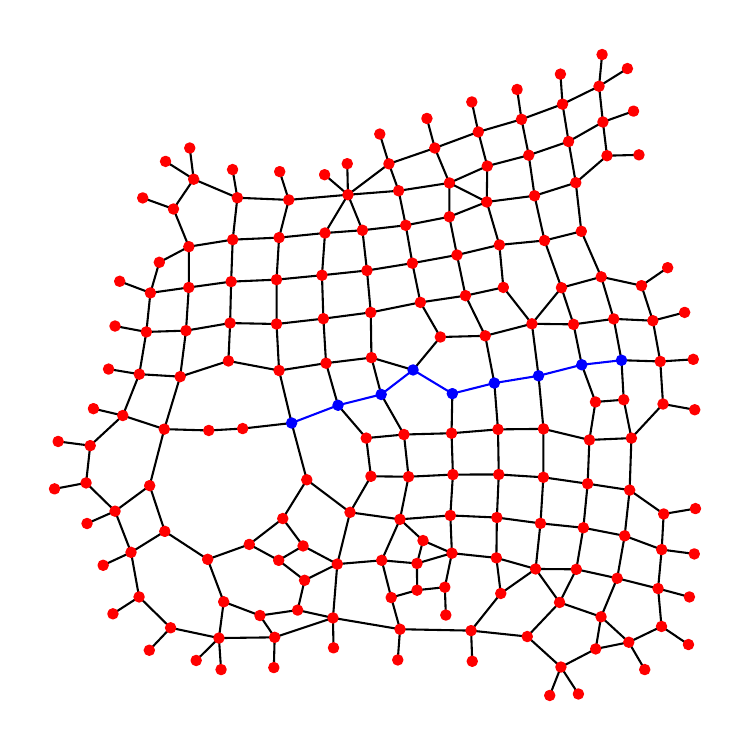}}
        \caption{$k=8$\\$C_8(P^\star) = 183$, $|P^\star| = 9$}
        \label{fig:washington-8}
    \end{subfigure}%
    ~ 
    \begin{subfigure}[c]{0.32\textwidth}
        \resizebox{\textwidth}{!}{\includegraphics{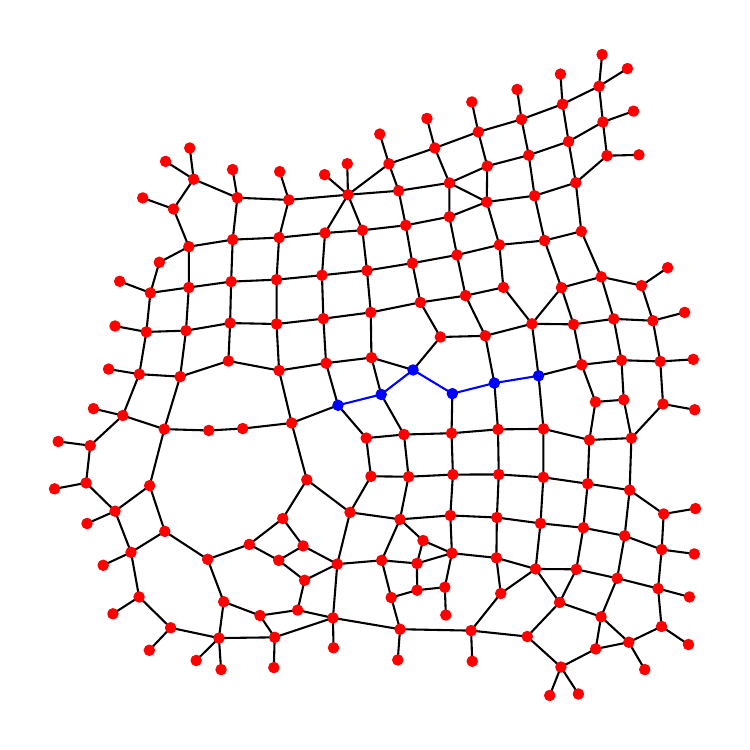}}
        \caption{$k=9$\\$C_9(P^\star) = 186$, $|P^\star| = 6$}
        \label{fig:washington-9}
    \end{subfigure}
    \caption{$k$-step-central shortest paths for Washington (diameter: 21).}
    \label{fig:washington}
    
\end{figure*}

\end{document}


\begin{tikzpicture}[font=\small,]
    \begin{groupplot}[
        group style={
            group size=2 by 1,
            horizontal sep=1em,
            yticklabels at=edge left,
            ylabels at=edge left,
            xlabels at=edge bottom,
        },
        x grid style={darkgray},
        xlabel={$k$},
        xmin=0, xmax=10,
        xtick style={color=black},
        y grid style={darkgray},
        ylabel={Relative length/neighbourhood of path},
        ymin=0, ymax=1.05,
        ytick style={color=black},
        ytick={0,0.2,...,1},
        legend columns=4,
        legend style={/tikz/every even column/.append style={column sep=0.5cm}}
    ]
    \nextgroupplot[title={Relative length of path},legend to name=legendpos1]
    \addplot [draw=black, mark=o, mark options={solid}, solid]
        table{%
        x  y
        1 0.9583333333333334
        2 0.9583333333333334
        3 0.7916666666666666
        4 0.7083333333333334
        5 0.6666666666666666
        6 0.5416666666666666
        7 0.4583333333333333
        8 0.375
        9 0.2916666666666667
        };
    \addlegendentry{Barcelona}
    \addplot [draw=red, mark=x, mark options={solid}, solid]
        table{%
        x  y
        1 0.9473684210526315
        2 0.8947368421052632
        3 0.8421052631578947
        4 0.8157894736842105
        5 0.7368421052631579
        6 0.7105263157894737
        7 0.631578947368421
        8 0.631578947368421
        9 0.631578947368421
        };
    \addlegendentry{Bologna}
    \addplot [draw=blue, mark=triangle, mark options={solid}, solid]
        table{%
        x  y
        1 0.9090909090909091
        2 0.7272727272727273
        3 0.7272727272727273
        4 0.6363636363636364
        5 0.5909090909090909
        6 0.5454545454545454
        7 0.45454545454545453
        8 0.45454545454545453
        9 0.4090909090909091
        };
    \addlegendentry{Brasilia}
    \addplot [draw=orange, mark=diamond, mark options={solid}, solid]
        table{%
        x  y
        1 0.875
        2 0.625
        3 0.5
        4 0.25
        5 0.125
        6 0.125
        7 0.125
        8 0.125
        9 0.125
        };
    \addlegendentry{Irvine1}
    \addplot [draw=green, mark=o, mark options={solid}, dashed]
        table{%
        x  y
        1 0.9565217391304348
        2 0.782608695652174
        3 0.6956521739130435
        4 0.6086956521739131
        5 0.5217391304347826
        6 0.5217391304347826
        7 0.43478260869565216
        8 0.34782608695652173
        9 0.2608695652173913
        };
    \addlegendentry{Irvine2}
    \addplot [draw=brown, mark=x, mark options={solid}, dashed]
        table{%
        x  y
        1 0.9545454545454546
        2 0.8636363636363636
        3 0.8636363636363636
        4 0.7727272727272727
        5 0.6818181818181818
        6 0.6363636363636364
        7 0.5
        8 0.36363636363636365
        9 0.2727272727272727
        };
    \addlegendentry{Los Angeles}
    \addplot [draw=violet, mark=triangle, mark options={solid}, dashed]
        table{%
        x  y
        1 0.96
        2 0.8
        3 0.8
        4 0.76
        5 0.72
        6 0.68
        7 0.56
        8 0.48
        9 0.36
        };
    \addlegendentry{New Delhi}
    \addplot [draw=cyan, mark=diamond, mark options={solid}, dashed]
        table{%
        x  y
        1 0.9583333333333334
        2 0.9583333333333334
        3 0.9583333333333334
        4 0.7916666666666666
        5 0.7083333333333334
        6 0.7083333333333334
        7 0.5833333333333334
        8 0.5
        9 0.4166666666666667
        };
    \addlegendentry{New York}
    \addplot [draw=olive, mark=o, mark options={solid}, dotted]
        table{%
        x  y
        1 0.9285714285714286
        2 0.8928571428571429
        3 0.8571428571428571
        4 0.8571428571428571
        5 0.75
        6 0.7142857142857143
        7 0.6071428571428571
        8 0.6071428571428571
        9 0.5714285714285714
        };
    \addlegendentry{Paris}
    \addplot [draw=purple, mark=x, mark options={solid}, dotted]
        table{%
        x  y
        1 0.9591836734693877
        2 0.9387755102040817
        3 0.8979591836734694
        4 0.8571428571428571
        5 0.8367346938775511
        6 0.5714285714285714
        7 0.5918367346938775
        8 0.6122448979591837
        9 0.6122448979591837
        };
    \addlegendentry{Richmond}
    \addplot [draw=teal, mark=triangle, mark options={solid}, dotted]
        table{%
        x  y
        1 0.98
        2 0.94
        3 0.9
        4 0.86
        5 0.84
        6 0.82
        7 0.8
        8 0.78
        9 0.78
        };
    \addlegendentry{Savannah}
    \addplot [draw=brown, mark=diamond, mark options={solid}, dotted]
        table{%
        x  y
        1 0.9313725490196079
        2 0.9215686274509803
        3 0.8921568627450981
        4 0.6470588235294118
        5 0.6274509803921569
        6 0.6176470588235294
        7 0.6078431372549019
        };
    \addlegendentry{Seoul}
    \addplot [draw=lime, mark=o, mark options={solid}, dashdotted]
        table{%
        x  y
        1 0.96875
        2 0.90625
        3 0.90625
        4 0.84375
        5 0.78125
        6 0.71875
        7 0.65625
        8 0.65625
        9 0.625
        };
    \addlegendentry{Vienna}
    \addplot [draw=magenta, mark=x, mark options={solid}, dashdotted]
        table{%
        x  y
        1 0.9090909090909091
        2 0.8636363636363636
        3 0.8181818181818182
        4 0.7272727272727273
        5 0.6818181818181818
        6 0.5909090909090909
        7 0.5
        8 0.4090909090909091
        9 0.3181818181818182
        };
    \addlegendentry{Walnut Creek}
    \addplot [draw=gray, mark=triangle, mark options={solid}, dashdotted]
        table{%
        x  y
        1 0.9523809523809523
        2 0.9523809523809523
        3 0.8571428571428571
        4 0.8095238095238095
        5 0.6666666666666666
        6 0.6666666666666666
        7 0.47619047619047616
        8 0.42857142857142855
        9 0.2857142857142857
        };
    \addlegendentry{Washington}
    \coordinate (c1) at (rel axis cs:0,1);
    \nextgroupplot[title={Relative size of neighbourhood}]
    \addplot [draw=black, mark=o, mark options={solid}, solid]
        table{%
        x  y
        1 0.20952380952380953
        2 0.41904761904761906
        3 0.6190476190476191
        4 0.7476190476190476
        5 0.8428571428571429
        6 0.9047619047619048
        7 0.9476190476190476
        8 0.9571428571428572
        9 0.9666666666666667
        };
    \addplot [draw=red, mark=x, mark options={solid}, solid]
        table{%
        x  y
        1 0.10536044362292052
        2 0.22181146025878004
        3 0.3456561922365989
        4 0.4584103512014787
        5 0.5471349353049908
        6 0.6303142329020333
        7 0.7060998151571165
        8 0.7911275415896488
        9 0.8502772643253235
        };
    \addplot [draw=blue, mark=triangle, mark options={solid}, solid]
        table{%
        x  y
        1 0.1564245810055866
        2 0.3128491620111732
        3 0.4748603351955307
        4 0.6089385474860335
        5 0.7150837988826816
        6 0.7821229050279329
        7 0.8547486033519553
        8 0.9050279329608939
        9 0.9385474860335196
        };
    \addplot [draw=orange, mark=diamond, mark options={solid}, solid]
        table{%
        x  y
        1 0.4375
        2 0.75
        3 0.875
        4 0.9375
        5 0.96875
        6 0.96875
        7 0.96875
        8 0.96875
        9 0.96875
        };
    \addplot [draw=green, mark=o, mark options={solid}, dashed]
        table{%
        x  y
        1 0.1382488479262673
        2 0.2995391705069124
        3 0.46543778801843316
        4 0.5898617511520737
        5 0.6497695852534562
        6 0.7235023041474654
        7 0.7695852534562212
        8 0.783410138248848
        9 0.7926267281105991
        };
    \addplot [draw=brown, mark=x, mark options={solid}, dashed]
        table{%
        x  y
        1 0.16666666666666666
        2 0.35833333333333334
        3 0.5583333333333333
        4 0.7291666666666666
        5 0.8458333333333333
        6 0.9166666666666666
        7 0.95
        8 0.9625
        9 0.975
        };
    \addplot [draw=violet, mark=triangle, mark options={solid}, dashed]
        table{%
        x  y
        1 0.1388888888888889
        2 0.27380952380952384
        3 0.3968253968253968
        4 0.5238095238095238
        5 0.6547619047619048
        6 0.7658730158730159
        7 0.8531746031746031
        8 0.9087301587301587
        9 0.9404761904761905
        };
    \addplot [draw=cyan, mark=diamond, mark options={solid}, dashed]
        table{%
        x  y
        1 0.18548387096774194
        2 0.3467741935483871
        3 0.5080645161290323
        4 0.6290322580645161
        5 0.7298387096774194
        6 0.8266129032258065
        7 0.9153225806451613
        8 0.9435483870967742
        9 0.9596774193548387
        };
    \addplot [draw=olive, mark=o, mark options={solid}, dotted]
        table{%
        x  y
        1 0.14328358208955225
        2 0.2865671641791045
        3 0.4388059701492537
        4 0.5701492537313433
        5 0.6835820895522388
        6 0.7910447761194029
        7 0.8656716417910447
        8 0.9253731343283582
        9 0.9522388059701492
        };
    \addplot [draw=purple, mark=x, mark options={solid}, dotted]
        table{%
        x  y
        1 0.11621233859397417
        2 0.24103299856527977
        3 0.3715925394548063
        4 0.5050215208034433
        5 0.6040172166427547
        6 0.703012912482066
        7 0.8034433285509326
        8 0.8737446197991392
        9 0.926829268292683
        };
    \addplot [draw=teal, mark=triangle, mark options={solid}, dotted]
        table{%
        x  y
        1 0.125
        2 0.238013698630137
        3 0.3647260273972603
        4 0.5
        5 0.6335616438356164
        6 0.7534246575342466
        7 0.833904109589041
        8 0.8938356164383562
        9 0.9297945205479452
        };
    \addplot [draw=brown, mark=diamond, mark options={solid}, dotted]
        table{%
        x  y
        1 0.08860759493670886
        2 0.18181818181818182
        3 0.28883774453394706
        4 0.3866513233601841
        5 0.475258918296893
        6 0.5581127733026467
        7 0.6317606444188723
        8 0.7031070195627158
        9 0.759493670886076
        };
    \addplot [draw=lime, mark=o, mark options={solid}, dashdotted]
        table{%
        x  y
        1 0.10706638115631692
        2 0.23126338329764454
        3 0.34475374732334046
        4 0.4539614561027837
        5 0.563169164882227
        6 0.6638115631691649
        7 0.7473233404710921
        8 0.8137044967880086
        9 0.8715203426124197
        };
    \addplot [draw=magenta, mark=x, mark options={solid}, dashdotted]
        table{%
        x  y
        1 0.1834319526627219
        2 0.35502958579881655
        3 0.5266272189349113
        4 0.6686390532544378
        5 0.7751479289940828
        6 0.8520710059171598
        7 0.9112426035502958
        8 0.9467455621301775
        9 0.9585798816568047
        };
    \addplot [draw=gray, mark=triangle, mark options={solid}, dashdotted]
        table{%
        x  y
        1 0.22916666666666666
        2 0.390625
        3 0.5677083333333334
        4 0.703125
        5 0.796875
        6 0.8697916666666666
        7 0.9166666666666666
        8 0.953125
        9 0.96875
        };
        \coordinate (c2) at (rel axis cs:1,1);
    \end{groupplot}
    \coordinate (c3) at ($(c1)!.5!(c2)$);
    \node[below] at (c3 |- current bounding box.south)
      {\pgfplotslegendfromname{legendpos1}};
\end{tikzpicture}